\documentclass[12pt]{article}
\usepackage[utf8]{inputenc}

\usepackage{subcaption}

\makeatletter
\renewcommand{\p@subfigure}{\thefigure} 
\makeatother
\DeclareCaptionLabelFormat{subfig}{Figure~\thefigure\alph{subfigure}}
\usepackage{float} 
\usepackage{lineno}

\usepackage{natbib}
\usepackage{geometry,amsmath,mathtools,blkarray,amssymb,amsthm,setspace,xcolor,tikz-cd,xparse,amsfonts,graphics,comment}

\usetikzlibrary{shapes.geometric}
\usetikzlibrary{calc}
\usepackage{tikz,verbatim}

\usepackage[shortlabels]{enumitem}
\usepackage{mathrsfs}  

\usepackage[colorlinks,linkcolor=blue,citecolor=blue,urlcolor=blue,bookmarks=false,hypertexnames=true]{hyperref}
\newtheorem{theorem}{Theorem}[section]
\newtheorem{definition}{Definition}[section]

\newtheorem{proposition}{Proposition}[section]
\newtheorem{lemma}{Lemma}[section]
\newtheorem{corollary}{Corollary}[section]
\newtheorem{example}{Example}[section]
\newtheorem{remark}{Remark}[section]

\theoremstyle{definition}

\def\Re{\mathbf{R}}

\def\L{\mathscr{L}}
\def\A{\mathscr{A}}
\def\X{\mathscr{X}}
\def\D{\mathscr{D}}
\def\E{\mathscr{E}}
\def\muA{\mu_{{\scriptscriptstyle \A}}}
\def\muX{\mu_{{\scriptscriptstyle  \X}}}
\def\rhoa{\rho}

\def\rhoX{\rho_{\scriptscriptstyle \X}}

\newcommand{\bit}{\begin{itemize}}

\newcommand{\eit}{\end{itemize}}
\newcommand{\co}{\text{co.}}

\title{Random Utility  with   Aggregation}

\author{%
\begin{tabular}{@{}p{0.32\textwidth}p{0.32\textwidth}p{0.32\textwidth}@{}}
\centering Yuexin Liao &
\centering Kota Saito &
\centering Alec Sandroni\footnote{%
Liao: \textsf{yliao@caltech.edu}. Saito: \textsf{saito@caltech.edu}. Sandroni: \textsf{alecsandroni01@gmail.com}.%
We thank participants at presentations at the Caltech Lunch Seminar, MIT Theory Lunch Seminar, Waseda University, Keio University, the University of Tokyo, and the University of Osaka. We also thank Isaiah Andrews, Drew Fudenberg, Whitney Newey, and Parag Pathak for insightful comments. Alec Sandroni acknowledges financial support from the Summer Undergraduate Research Fellowships (SURF) program at Caltech (2022--2024). Yuexin Liao acknowledges financial support from the SURF program at Caltech (2025).%
} \tabularnewline
\centering\small California Institute of Technology &
\centering\small California Institute of Technology &
\centering\small Massachusetts Institute of Technology
\end{tabular}
}

\begin{document}

\maketitle

\begin{abstract}
We study random utility (RU) rationality with aggregation when the underlying alternatives in each aggregate vary across consumers and are unobserved, as is typical for an outside option. RUM over the underlying alternatives is the natural assumption on the data generating process, while an aggregated random utility model (ARUM) is the standard empirical tool. We characterize RU rationality in three frameworks and show its testable implications are substantially weaker than those of an ARUM. We provide two independent conditions for their equivalence: non-overlapping preferences within aggregates and menu-independent aggregation. Simulations show that violating either condition produces meaningful estimation bias when imposing an ARUM.\end{abstract}

\section{Introduction}

In empirical economics, it is common practice to aggregate alternatives. For instance, in meat sales data (e.g. \cite{ealesmeat1982}), all beef products, despite differing in quality and price, are sometimes grouped under a single aggregate, “beef.” 
Another particularly important form of aggregation involves {\it the outside option}. A canonical treatment in the standard IO framework  models the outside option as a single alternative that aggregates all otherwise unlisted choices (see, e.g., the textbook treatment in \citet{shum2016econometric}). For instance, in \cite{nevo2001measuring}, which studies consumer demand for popular cereal brands, the market share of the outside good is defined as  the residual: the difference between one and the total market share of the inside goods, implicitly capturing all other breakfast options, such as pancakes or omelettes.\footnote{See Appendix A: Data, page 337 of \cite{nevo2001measuring}.}

The standard modeling approach assumes a random utility model  defined directly over the aggregates---what we call the {\it aggregated random utility model} (ARUM). That is, researchers assume that each aggregate corresponds to a single alternative, and that preferences are defined over these aggregates as if they were atomic.

However, this simplification obscures an important complication. While the analyst only observes choice frequencies over aggregates, consumers evaluate the underlying alternatives that compose each aggregate—alternatives that may differ substantially in desirability. In the meat example, ground beef and table cuts vary significantly in both features. Likewise, the outside option may encompass alternatives that differ greatly in quality and price. Thus, a random utility model (RUM) should in principle be defined over the underlying alternatives, {\it not} over analyst-defined aggregates. The difficulty is that the exact composition of an aggregate may be heterogeneous across consumers and unobservable to the analyst. For instance, some consumers may not consider omelettes a viable breakfast option due to egg shortages from bird flu, while others may rule out pancakes because essential ingredients are unavailable.

In this paper, we formalize a RUM that accommodates such heterogeneous and unknown composition of aggregates, and we study its relationship to the ARUM. This relationship is important because the two models play distinct roles. A RUM over the underlying alternatives is the natural assumption on the data generating process, since consumers evaluate and choose among the actual underlying goods available to them. An ARUM, by contrast, is an empirical tool that enables tractable estimation by treating each aggregate as a primitive object of choice. Whether this simplification is valid depends on whether the choice frequencies over aggregates, generated by a RUM on the underlying alternatives, are also consistent with an ARUM.

This paper is the first to characterize the testable implications of a RUM when aggregated categories have heterogeneous and unknown composition. Using three different frameworks, we show precisely how these implications are weaker than those of an ARUM. Through simulations, we demonstrate that this gap can lead to significant estimation biases when an ARUM is incorrectly imposed. Building on these findings, we provide theoretical guidance on how to define aggregated categories in ways that mitigate such biases.

\vspace{-0.2cm}

\subsection{Preview of theoretical results}

\vspace{-0.2cm}

To formalize our framework, consider a setup involving popular cereal brands along with an outside good that contains all other breakfast options. In this setup, the outside option is an aggregated alternative and all other alternatives are not aggregated. All definitions and results below are naturally generalized into the case of multiple aggregated alternatives. 

The composition of the outside good may differ across individuals in the population due to local availability. We formalize this using two key concepts. An {\it aggregation correspondence} $X$ specifies the possible underlying goods that an aggregate may represent--for example, the outside option may include omelettes, pancakes, or other breakfast items. A {\it composition distribution} $\lambda$ then assigns, for each menu, probabilities to these possible realizations of the outside option. Thus, when the outside option is present, $\lambda$ determines whether a consumer faces only omelettes, only pancakes, both, or some other combination of breakfast goods. Crucially, the analyst does not observe the aggregation correspondence or the composition distribution.

We say that the choice frequencies $\rhoa$  is consistent with a RUM if there exists (i) a composition distribution $\lambda$ with an aggregation correspondence $X$ and (ii) a probability distribution~$\muX$ over rankings (i.e., linear orders)~$\succ$ on the set $\X$ of the underlying alternatives (in this case, the popular cereal brands along with other breakfast options)\footnote{We write $\muX$ to clarify that the measure is defined on rankings over the underlying alternatives $\X$, not the aggregates.}, such that for all subsets $D$ of cereals and $x \in D$, the market share $\rho(D \cup \text{out}, x)$ of cereal $x$ in a market  where the set $D$ of cereals is available equals 
\begin{equation}\label{eq:examp}
\begin{aligned}
\sum_{S \subseteq X(out)}\lambda_{D \cup {\text{out}}}(S) \, \muX(\succ \mid x \succ y \text{ for all } y \in (D\setminus \{x\}) \cup S),
\end{aligned}
\end{equation}
where $\lambda_{D \cup {\text{out}}}(S)$ is the probability that the outside option ``out" is composed of the set $S \subseteq X(out)$ in the context of choice set $D \cup {\text{out}}$.  This formulation is a natural extension of the standard RUM: Equation~\eqref{eq:examp} shows that the observed choice probability is a weighted average of standard RUM probabilities, with weights $\lambda$ reflecting the distribution of outside-option compositions.

We theoretically study the implications of RU-rationality and compare them with those of ARUM along three dimensions. First, we characterize the implications of RU-rationality in terms of choice frequencies. Theorem \ref{thm:1} show that RU-rationality is equivalent to two conditions: (i) {\it Limited Monotonicity}—choice frequencies over non-aggregated alternatives weakly decrease when the aggregated alternatives are  added to the menu; and (ii) {\it Partial RU-rationality}—on menus containing only non-aggregated alternatives, choices frequencies satisfy standard RUM restrictions. These conditions are substantially weaker than ARUM. Most notably, Limited Monotonicity does not require that choice frequencies decrease monotonically on menus already containing aggregated alternatives. This can occur because adding a new alternative may signal information about the composition of the aggregate alternatives—for instance, the presence of a premium cereal brand may indicate a high-income market where the outside option contains more attractive alternatives.

Second, we characterize the admissible deterministic behaviors under RU-rationality and compare them with those under ARUM. The deterministic behaviors completely characterize the RU-rational (respectively ARU-rational) stochastic choice functions  in the sense that every stochastic choice function of either class  can be written as a convex combination of deterministic behaviors.  Under ARUM, admissible behaviors are precisely the rational choice functions induced by linear orders. Theorem~\ref{thm:vertex} shows that RU-rationality admits a strictly richer class: an agent may maximize a linear order on some menus while defaulting to an aggregate such as the outside option on other menus. This can be interpreted as menu-dependent consideration—when faced with unfamiliar or complex choice sets, consumers may revert to the default rather than evaluate all options. Formally, the set of RU-rational stochastic choice functions forms a polytope whose vertices are these ``menu-effect'' choice functions, and the ARUM polytope is a strict subpolytope with double-exponentially fewer vertices.

Third, we characterize RU-rationality when the aggregation correspondence is known to the analyst but the composition distribution remains unknown. Theorem \ref{thm:nests} establishes that as the number of underlying alternatives in each aggregated alternative grows, the set of RU-rationalizable choice functions expands and stabilizes once it exceeds a threshold. Beyond the threshold, RU-rationality is fully characterized by Limited Monotonicity and Partial RU-rationality. Theorem  \ref{thm:errorbound} shows that even below this threshold, the RU-rationalizable set may closely approximate the unconstrained case. These results suggest that in practice, Limited Monotonicity and Partial RU-rationality capture the empirically relevant implications of RUM; thus the implication of RU-rationality is significantly weaker than that of ARU-rationality.

Finally, given the gap between RU-rationality and ARU-rationality, we ask: under what conditions do they coincide? We identify two independent necessary and sufficient conditions. The first is {\it non-overlapping } preferences in the support of $\mu_{\X}$ (Proposition \ref{thm:non-overlapping}): if all underlying alternatives  that each aggregate represents occupy adjacent positions in every consumer's preferences, then RU-rationality implies ARU-rationality regardless of the composition distribution. The second is menu-independent composition of $\lambda$ (Proposition \ref{thm:independent}): if the composition distribution does not vary across menus, then RU-rationality implies ARU-rationality regardless of the preference structure. These results provide practical guidance for constructing aggregates--alternatives should be grouped  if they are close substitutes and if their availability does not vary systematically across markets.

\vspace{-0.2cm}

\subsection{Simulation results}

\vspace{-0.2cm}

Our results imply that when either condition—non-overlapping preferences or menu independence—fails, the observed choice frequencies diverge from those predicted by an ARUM. Estimating an ARUM (such as logit model over aggregates)  can therefore generate substantial bias. Assessing this bias empirically is crucial for understanding the practical relevance of our theory.

We examine this prediction through simulations.  We assume that the true choice dataset defined on the underlying alternatives $\X$ is a logit model with fixed utility levels.  We then fix values of a composition distribution $\lambda$ and an aggregation correspondence $X$ and construct the reduced dataset $\rho$ over aggregates. Finally, we re-estimate utilities from $\rho$ as if it followed a logit model over aggregates. When $\lambda$ is menu-independent or preferences are non-overlapping, $\rho$ is consistent with an ARUM and bias must be low, otherwise the bias should become large as $\rho$ cannot be represented by an ARUM.

As our measure of bias, we examined the difference in estimated utilities between two atomic aggregates, and we also computed the geometric distance from $\rho$ to the set of ARUMs. Both measures increase as $\lambda$ becomes more menu-dependent or preferences more overlapping. In some cases, the bias was substantial enough to reverse the ranking of alternatives—for example, even when $u(x) > u(y)$ in the true model, the estimates produced $\hat{u}(y) > \hat{u}(x)$. These findings align with our theoretical predictions and underscore the empirical importance of potential biases. Further details are provided in Section~\ref{sec:simulation}.

\vspace{-0.2cm}

\subsection{Related Literature}

\vspace{-0.2cm}

There is no paper in the stochastic-choice literature that directly studies aggregation when the composition of aggregates is unknown to the analyst and may vary across individuals and menus. There are, however, several related papers that share aspects of our motivation or mathematical setup.

A first related line of research is the recent work on random attention (or limited consideration) models, such as \citealp{masatlioglu2020ram,aguiar2022attention,rehbeck2016attention,manzinimariotti2014ecma,horan2019default}. See \citealp{tomasztextbook} for a survey. These papers study environments in which agents do not evaluate the full menu but instead randomly consider a subset. In particular, \citet{horan2019default} analyzes a setting in which the choice frequency of the outside option is not observed.

The resemblance, however, is limited. Random-attention models are driven by unobserved consideration sets, whereas our setting is driven by aggregation-induced unobservability: the analyst observes only aggregates, whose internal composition can differ across consumers and menus. This difference has sharp implications. In our framework, when an aggregate is present, at least one of its underlying components is available; random-attention models impose no such requirement. Moreover, in the random attention model, it is  assumed that the analyst observes an agent's choices over underlying alternatives; and they  violate RU rationality because some options are not considered. In contrast, we maintain random utility at the underlying level;  violations arise only at the aggregated level.

\cite{rehbeck2022marginal} studies a mathematically related model but with a different focus. They characterize marginal stochastic choice—frequencies not conditioned on particular menus—assuming the analyst observes the distribution of menus but not within-menu choices. By contrast, in our framework the distribution over available alternatives is unobserved and captured by $\lambda$, while we observe reduced stochastic choice over aggregates. Although \cite{rehbeck2022marginal} also allows for endogenous menu availability, their analysis centers on models such as \cite{gul2001temptation} and \cite{kreps1979flexibility}, and addresses rationalizability rather than estimation bias. \cite{apesteguia2016stochastic} also study aggregation of stochastic choice data, but with a different motivation. They examine whether aggregated choice behavior, obtained by combining individual choices, remains in the same class as the underlying individual behaviors.

 \cite{kono2025unobservable} also studies the RUM with an outside option, but in a different setup. The paper focuses on the outside option case. Moreover, from the perspective of the observed dataset, the key distinction between the two papers is that the present study takes the reduced dataset~$\rho$ as the observed object, whereas  \cite{kono2025unobservable} treats the full dataset over the underlying alternatives (with missing information) as given. In addition,  \cite{kono2025unobservable} assumes the outside option always corresponds to the full set of underlying alternatives—that is, $\lambda$ is menu-independent and degenerate. Even in this idealized case, much of the implications of the RUM are lost. In the present paper, we extend that framework by allowing $\lambda$ to be non-degenerate or menu-dependent, and we characterize the observable implications of the RUM under this more general specification.

A few papers in the empirical IO literature such as \cite{brownstone2017broad}, \cite{stafford2018fishing}, and \cite{huangoutside} have examined potential biases arising from aggregation and misspecification of the outside option. More recently, \cite{Zhang2023identification} analyzed biases resulting from misspecifying market size, and hence the market share of the outside good. However, all of these papers are purely empirical in nature and do not address the issue of unknown composition or outside options theoretically.

Although less directly related, classical statistical work has developed a comprehensive framework for making inferences from data that are partially observed or categorized. For example, see \cite{heitjan1989} for a review of grouped data; \cite{blumenthal1968} and \cite{nordheim1984} for categorized data; and \cite{heitjan1990, heitjan1991} for coarse data. In this framework, subsets of the sample space are observed rather than exact realizations.\footnote{For example, if the sample space is the real line (e.g., blood pressure), we may only observe a random variable taking interval values (e.g., an interval of blood pressures containing the true value).} A key distinction from our setting is that, in our model, uncertainty arises from the actual underlying alternatives that an agent faces, rather than from the data generating process. We model decision making of consumers under the random utility model, which does not appear in this statistical literature.




\vspace{-0.2cm}

\section{Setup}

\vspace{-0.2cm}

We introduce two notations that we use throughout the paper: For any finite set $Z$, let $\L(Z)$ be the set of rankings (i.e., linear orders) on $Z$ and let $\Delta(Z)$ be the set of probability distributions over $Z$. 

Let ${\A}=\{a_0, .... , a_{N}\}$ be the finite set of {\it aggregates}. There are two types of aggregates; one is {\it atomic} aggregate, which has only one underlying alternative. The other one is a {\it non-atomic} aggregate such as the outside option, which has multiple different underlying alternatives.  (See Example \ref{ex:nevo} below.)
When there is no danger of confusion, we simply say aggregates rather than non-atomic aggregates. We denote the set of atomic aggregates by $\A_A$; denote the set of non-atomic aggregates by $\A_N$. Thus,  $\{\A_A, \A_N\}$ forms a partition of $\A$.\footnote{That is, $\A_A \cup \A_N = \A$ and $\A_A \cap \A_N = \emptyset$.}

Let $\X$ be the set of underlying alternatives. We assume that the underlying set $\X$ is finite.\footnote{We will make an assumption on the richness of $\X$ later for some results.} In this paper, the partition $\{\A_A, \A_N\}$ and $\X$ are exogenously given.

We now define an {\it aggregation correspondence}, which connects $\A$ and $\X$.

\begin{definition}
A function $X: \A \to 2^\X \setminus \{\emptyset\}$ is called an {\it aggregation} correspondence  if (i) $X(a_i) \cap X(a_j)= \emptyset$ for any distinct $i, j$; (ii) $X(a)$ is singleton if $a \in \A_A$; (iii) $X(a)$ is non-singleton if $a \in \A_N$. 
\end{definition}

We call $X$ an aggregation correspondence because the correspondence (or more precisely, its inverse) specifies how each underlying alternative is assigned to an aggregate. Condition (i) requires that every underlying alternative be assigned to exactly one aggregate. Conditions (ii) and (iii) justify the interpretation that $\A_A$ represents the set of {\it atomic aggregates}, while $\A_N$ represents the set of {\it non-atomic aggregates}.




In many parts of the paper, we consider the setup in which there is only one non-atomic aggregate $a_0$ called {\it outside option}.  In the setup, we assume $\A_N=\{a_0\}$. 

\begin{example}(\textbf{Cereal Choice}) \label{ex:nevo}
Consider a similar setup to \cite{nevo2001measuring}, who estimates demand over the popular cereal brands. As mentioned in the introduction, a key feature of the analysis is the inclusion of an outside good, representing all breakfast options not explicitly listed among the popular cereal brands in the dataset. In the setup, all available alternatives are 
\[
{\mathscr A}= \{\text{Kellogg Corn Flakes, Kellogg Crispix, $\cdots$, Quaker Life, Outside goods}
\}.
\]
Since  the cereal brands are finely specified in $\A$, all cereals are non-aggregated; thus we have $\A_A= \{\text{Kellogg Corn Flakes,  $\cdots$, Quaker Life}\}$; and $\A_N= \{\text{Outside goods}\}$. For example, $
X(\text{Kellogg Corn Flakes})= \{\text{Kellogg Corn Flakes}\}$.  Under the assumption that the outside goods contain only two alternatives, we have  $X(\text{Outside goods})= \{\text{Omelette, Pancakes}\}$.\footnote{In our example, we assume omelettes are unavailable when eggs are in short supply. An alternative setup is that eggs remain available but at higher cost, so omelettes are still feasible though less attractive. We can represent this case as “Omelette with high price,” reflecting the increased cost while keeping it as an available alternative.}
\end{example}

\vspace{-0.2cm}

\subsection{Datasets and RU-rationalities}

\vspace{-0.2cm}

Let $\D \subseteq 2^\A \setminus \{\emptyset\}$ be the set of choice sets. The dataset in this paper is a stochastic choice function $\rhoa$ on $\D$. That is, the dataset is a function $\rho: \D \times \A \rightarrow \mathbf{R}_+$ such that for any $ A \in \D$, $\sum_{a \in A}\rhoa(A,a) =1$,  $\rhoa(A,a) \ge 0$ for all $a \in A$ and $\rho(A,a) = 0$ for all $a \notin A$.  Note that the stochastic choice function is defined over menus of aggregates in ${\mathscr A}$, not over  the underlying alternatives $\X$. Thus, we sometimes call $\rhoa$ {\it reduced} dataset. We define a distance function between two  stochastic choice functions $\rho$ and $\rho'$ by $\|\rho-\rho'\|=\sqrt{\sum_{x \in D \in \D} [\rho(D,x) - \rho'(D,x)]^2}$ by using the standard Euclidean norm.

\begin{remark}
In most parts of this paper, we assume that $\rhoa$ is defined on all subsets of $\A$ (i.e., $ \D=2^{\A} \setminus \{\emptyset\}$) . This domain assumptions are technical, as they only serve to facilitate our expositions.  For each main reuslt, the assumption can be weakened. For details, see Remark \ref{rem:domain_assumption} for Theorem \ref{thm:1};   Remark \ref{rem:dom_vertex} for Theorem \ref{thm:vertex}; and Remark \ref{rem:dom_nests} for Theorem \ref{thm:nests} in the online appendix.
\end{remark}

As mentioned in the introduction, in discrete choice analysis, the standard method in the presence of aggregation is to assume {\it aggregated random utility models} (ARUMs), which is a probability distribution over rankings over $\A$.

\begin{definition}
A stochastic choice function $\rhoa$ is aggregated random utility (ARU)-rational if there exists $\muA \in \Delta(\L(\A))$ such that  for any $a \in A \in \D$,\footnote{We write $\muA$ to clarify that the measure is defined over rankings over $\A$, not on the set $\X$ of underlying alternatives.} 
\begin{align}\label{eq:agg_rationalization}
\rhoa(A,a) =\muA(\ \succ\in \L(\A) \  | a \succ b \text{ for all } b \in A \setminus a).
\end{align}
\end{definition}

As discussed in the introduction, consumers evaluate and choose among the actual goods available to them, not among analyst-defined aggregates. Thus, RUM over the underlying alternatives $\X$ is the natural assumption on the data generating process. The difficulty is that the components of a non-atomic aggregate may differ substantially—for example, across beef products or within a heterogeneous outside option—and that the exact composition may vary across consumers and be unobservable to the analyst. To capture this, we introduce a \emph{composition distribution} which links the aggregates to the distribution over underlying alternatives in their composition. Given an aggregation correspondence, this mapping $\lambda$ specifies the distribution of available underlying alternatives faced by agents in the population.

\begin{definition}Given an aggregation correspondence $X$, a set of functions $(\lambda_A)_{A \in \D}$ is called a {\it composition distribution} if for any $A \in \D$,
\begin{itemize}
    \item $\lambda_{A}((S_a)_{a \in A})\ge 0$ for any $(S_a)_{a \in A} \in \prod_{a \in A}(2^{X(a)}\setminus \{\emptyset\})$,
    \item $\sum_{(S_a)_{a \in A}\in \prod_{a \in A} (2^{X(a) }\setminus \{\emptyset\})}\lambda_{A}((S_a)_{a \in A})=1$.
\end{itemize} 
When there is no danger of confusion, we simply write $\lambda$ rather than $(\lambda_A)_{A \in \D}$.
\end{definition}

The value of the composition distribution $\lambda_{A}((S_a)_{a \in A})$ represents the frequency with which each aggregate $a$ is composed of the subset $S_a \subseteq X(a)$ within the choice set $A \subseteq \A$. We allow $\lambda$ to depend on the choice set $A$, since the availability of underlying alternatives may vary across markets---for example, it may depend on which stores consumers can access. In Example \ref{ex:nevo}, eggs may be unavailable in one market (e.g., due to a bird flu outbreak), making omelettes infeasible, while pancake ingredients remain widely available. In this case, $\lambda_{A}$ would assign higher probability to $S_{\text{out}}$ containing pancakes than to sets including omelettes.

We now define RU-rationality in the presence of aggregates whose compositions are unknown to the analyst. To motivate the definition, return to the cereal example and consider a distribution $\muX$ over preferences among the underlying alternatives (the cereals and other breakfast options). Let $D$ be the set of cereals available in a store and let $x \in D$ denote one such cereal. For a subpopulation of consumers whose outside option consists only of omelettes and pancakes, the market share of $x$ is $\muX(\succ \mid x \succ y \text{ for all } y \in D \cup \{\text{omelette}, \text{pancakes}\})$. Analogous expressions arise for other possible compositions of the outside option. Aggregating across these compositions using $\lambda_A$ yields equation~\eqref{eq:examp} in the introduction. The following extends this construction to the case with multiple non-atomic aggregates:

\begin{definition}\label{def:rationalization}
A stochastic choice function $\rhoa$ is RU-rational if there exist $\muX \in \Delta(\L (\X))$, an aggregation correspondence $X$, and a composition distribution $\lambda$ such that for any $a \in A \in \D$, 
\begin{align}\label{eq:rationalization}
\rhoa(A,a) = \sum_{(S_b)_{b \in A}\in \underset{b\in A}{\prod} (2^{X(b)}\setminus \{\emptyset\})} \lambda_{A}\big((S_b)_{b \in A}\big)\muX\big(\succ\in \L (\X) | \exists x \in S_a, \ \forall y \in \underset{b \in A \setminus a}{\cup} S_b, \ x \succ y\big).
\end{align}

In this case, we say that the pair $(\muX, \lambda)$ \emph{rationalizes} $\rho$.\footnote{Strictly speaking, we should write $(\muX, X, \lambda)$ rationalizes $\rho$. However, since the role of $X$ is implicit in $\lambda$, we will simply write $(\muX, \lambda)$.}  
\end{definition}

Note that the term $\muX(\cdots)$ on the right-hand side represents the total market share of aggregate~$a$ from menu~$A$ within the subpopulation of consumers for whom each aggregate~$b$ corresponds to the subset $S_b \subseteq X(b)$ of underlying alternatives. The term $\lambda_A((S_b)_{b \in A})$ denotes the proportion of such consumers in the overall population. Thus, the right-hand side of equation~\eqref{eq:rationalization} is a natural extension of the standard RUM: the observed choice probability is a weighted average of standard RUM probabilities, with weights given by $\lambda$ reflecting the distribution of compositions across consumers.

\begin{remark}
Definition~\ref{def:rationalization} assumes independence between the distribution~$\muX$ over preference rankings and the composition distribution~$\lambda$ over underlying alternatives. While one could consider a more general framework allowing dependence between these components, such a model may lack testable implications, since the joint distribution could vary with the menu as $\lambda$ does. In particular, if the agent's preferences may depend on composition, such dependence would allow the distribution over rankings to vary with the menu, rendering the model unfalsifiable. If, however, we only allow the composition distribution $\lambda$ to depend on the preferences, the characterization remains the same. See Remark \ref{rem:generalru} in the online appendix for more details.
\end{remark}

\begin{remark}\label{rem:aggregated_RU_RU} ARU-rationality implies RU-rationality.  To see why suppose $\rhoa$ is rationalized by $\muA$ over aggregates $\A$. Choose an aggregation correspondence $X$ so that $X(a)$ is a singleton for all $a \in \A$. Label the sole element of $X(a)$ by $x_a$. Then for each $\succ$ in the support of $\muA$ define $\succ'$ by $x_a \succ' x_b$ whenever $a \succ b$ and define all other comparisons arbitrary. By setting $\muX(\succ') = \muA (\succ)$ for all $\succ$ in the support of $\muA$, it is apparent that $(\muX,\lambda)$ rationalizes $\rhoa$.
\end{remark}

We define one more concept of rationality.

\begin{definition} A stochastic choice function $\rhoa$ is {\it partially RU-rational} if there exists $\mu \in \Delta(\L(\A_A))$ such that for any $a \in A  \subseteq \A_A$ such that $A \in \D$, 
\begin{align}\label{eq:partial_ru}
\rhoa(A,a)= \mu(\succ\in \L(\A_A)| a \succ b \text{ for all }b \in A \setminus a).
\end{align}
\end{definition}

\begin{remark} 
RU-rationality implies partial RU-rationality. When the menu contains only atomic aggregates, each aggregate maps to a single underlying alternative, so the RU-representation collapses to a standard RUM on the set of atomic aggregates, which corresponds to  partial RU rationality; $\lambda$ terms disappear because there is no compositional uncertainty on such menus.
\end{remark}

\begin{remark}Summarizing the relationship between the three rationality concepts, we have
\[
\text{ ARU }(\ref{eq:agg_rationalization}) \implies \text{ RU } (\ref{eq:rationalization}) \implies
\text{Partial RU } (\ref{eq:partial_ru}).
\]
As discussed in the introduction, RUM over the underlying alternatives is the natural assumption on the data generating process, while ARUM is the standard empirical tool. The central question is therefore whether choice frequencies generated by a RUM are also consistent with an ARUM. 
\end{remark}

\vspace{-0.2cm}

\section{Characterization}\label{sec:characterization}

\vspace{-0.2cm}

In this section, we characterize RU-rationality in three different frameworks and show precisely how its testable implications are weaker than those of ARU-rationality. 
Throughout this section except Section \ref{sec:multiple_cat}, we consider {\it the outside option setup} in which there exists only one non-atomic aggregate $a_0$. Formally, we set $\A_N = \{a_0\}$ and $\A_A = \A \setminus \{a_0\}$. We call $a_0$ \textit{the outside option}. This case is useful for two reasons. First, it reflects the common modeling practice of including an outside option---an aggregate that captures unobserved or unlisted choices---frequently used in empirical applications.  Second, this simplified setting serves as a foundation for analyzing the more general case in Theorem \ref{thm:1multiple} with multiple non-atomic aggregates in Section \ref{sec:multiple_cat}.

We will make an assumption for the following theorems (Theorem \ref{thm:1} and Theorem \ref{thm:vertex}) and its generalization (Theorem \ref{thm:1multiple}) for multiple non-atomic aggregates.  

\smallskip

\noindent\textbf{Assumption (Richness): } 
 {\it The set $\X$ contains at least $|\A|^2$ distinct elements, so that each aggregate  $a\in\A$ may be associated with an underlying set $X(a)$ containing at least $|\A|$ elements.\footnote{ This assumption is required for Theorems \ref{thm:1} and \ref{thm:1multiple}. Since the if direction of Theorem \ref{thm:vertex} relies on Theorem \ref{thm:1}, it is also required there. For the outside option setup where $\A_N = \{a_0\}$, the bound can be weakened to $\X$ contains at least $2|\A_A|+1$ elements; see Theorem \ref{thm:nests}. The stronger $|\A|^2$ assumption is only needed for multiple non-atomic aggregates. }}

\smallskip

This richness assumption is natural in applications where one aggregate is the outside option: the outside option typically aggregates a large and heterogeneous collection of underlying alternatives. We study the case in which the aggregation correspondence is exogenously fixed and $\X$ is not sufficiently rich in Section~\ref{sec:howmanyx}.

\begin{theorem}\label{thm:1} Assume Richness of $\X$.
        Consider the outside option setup (i.e., $\A_N=\{a_0\}$).  A stochastic choice function $\rho$ is RU-rational if and only if
$\rho$ satisfies the following two conditions:
    \begin{itemize}
\item[(i)] (Limited Monotonicity): if $b \in D \subseteq \A_A$ then $\rhoa(D, b) \ge \rhoa(D \cup \{a_0\},b)$,
\item[(ii)] (Partial RU-rationality) $\rhoa$ is partially RU rational.
    \end{itemize}
\end{theorem}

The necessity of (i) and (ii) is immediate. The sufficiency part of the proof is constructive and nontrivial: given a stochastic choice function $\rho$ that satisfies the two conditions (i) and (ii), we need to construct an aggregation correspondence $X$, a composition distribution $\lambda$  and a distribution $\mu_{\X}$ over linear orders on $\X$ that satisfies equation (\ref{eq:rationalization}).   The proof is in the appendix.

The theorem makes precise the weakness of RU-rationality relative to ARU-rationality: RU-rationality is substantially weaker. To see this, note that Limited Monotonicity requires only that adding the outside option does not increase the choice frequencies of {\it atomic aggregates}.  In particular, it does {\it not} require the monotonicity condition with respect to menus containing {\it non-atomic aggregates}. To see this, consider a choice set $D$ such that $a_0 \in D$ and $b \notin D$. The theorem shows that under RU-rationality, it is possible to observe
\[
\rhoa(D,a_0)<\rhoa(D\cup\{b\},a_0),
\]
which violates the standard monotonicity condition. This may arise because, in the expanded set $D \cup {b}$, the outside option $a_0$ may become more desirable than it was in $D$. 
 For example, in the context of cereal choice, the presence of a fancy, expensive cereal brand may signal that the market corresponds to a high-income area leading to a change in the interpretation or attractiveness of the outside good: the outside goods may contain more appealing options  such as smoked salmon or quiche.  Another example would be the newly added alternative $b$ is only available locally in a single area. Thus, $D\cup b$ corresponds to a special area which may contain a different composition of outside options. 
 By contrast, the meanings of atomic aggregates remain the same when $a_0$ is added to a choice set; thus the limited monotonicity condition holds.

ARU-rationality, on the other hand, imposes much stronger restrictions including full monotonicity and beyond. When $\rhoa$ is defined on all subsets of $\A$ (i.e., $\D= 2^{\A}\setminus \{\emptyset\}$), ARU-rationality is equivalent to the non-negativity of all Block-Marschak (BM) polynomials. Among these, monotonicity is one of the simplest and most intuitive implications. However, as Theorem~\ref{thm:1} shows, RU-rationality only implies monotonicity on limited choice sets. This result highlights the substantial gap between RU-rationality and ARU-rationality. 

A few more technical remarks are in order:

\begin{remark}
In our discussion, we focus on Limited Monotonicity rather than on partial RU-rationality. The reason is straightforward: RU-rationality on atomic alternatives is well studied.\footnote{For example, under the full domain assumption (i.e., $\D = 2^{\A}\setminus \{\emptyset\}$), it can be characterized by the non-negativity of Block–Marschak (BM) polynomials defined over sets composed solely of atomic aggregates.}
\end{remark}

\begin{remark}\label{rem:generalru}
    In our notion of RU rationality, we assume that the composition of aggregates and the agent's preferences are independent. A natural extension is to allow the composition distribution to depend on the agent's preferences. In Remark \ref{rem:generalru} of the online appendix, we show such a representation (\ref{eq:rationalization}) is also characterized by  Limited Monotonicity and Partial-RU rationality. 
\end{remark}

\begin{remark}\label{rem:nested}
Readers may wonder about the relationship between RU-rationality and the nested logit model. First, nested logit does not address the fundamental problem that the economic meaning of the outside option may vary across choice sets. Second, even when the composition $X(a_0)$ is known, nested logit is most effective precisely in environments where the underlying alternatives within the outside option share highly correlated utility shocks---as in the canonical red-bus/blue-bus example (see Section 4.2.2 of \citealp{train2003discrete}). In Section~\ref{sec:mu}, we show that this case corresponds exactly to a condition called {\it non-overlappingness of  preferences} under which RU-rationality implies ARU-rationality. In such settings, ARUM may entail little additional bias, so nested logit offers no real advantage.
\end{remark}

\vspace{-0.2cm}

\subsection{Vertex representation}\label{sec:vertex}

\vspace{-0.2cm}

In this section, we provide a vertex characterization of admissible behavior under
RU-rationality and show that many such behaviors are not ARU-rational. This analysis
does not require any assumption on the choice domain $\D$. To streamline the exposition,
however, we impose that every observable menu contains the outside option $a_0$
(i.e., $a_0\in D$ for all $D\in\D$), consistent with the common view that a default (status
quo) alternative is always available.\footnote{Even without this
assumption, the main message of this section is unchanged; see Remark~\ref{rem:dom_vertex} in the online appendix for details.}

Recall that an ARU-rational stochastic choice function can be written as a convex
combination of deterministic choice functions induced by linear orders over aggregates
(see Remark~\ref{rem:ARUM_polytope} below). This means that ARU-admissible deterministic
behaviors are precisely the rational choice functions generated by linear orders on the
set of aggregates.

To characterize admissible consumer behaviors under RU-rationality, we first introduce
the following notation: for each $\succ \in \L(\A)$ and $\E \subseteq \D$, define a function denoted by $c^{\succ}_{\E}$ over $\D$ as follows: for all $D \in \D$,
\[
c^{\succ}_{\E}(D)= \begin{cases}
\max_D \succ \ &\text{ if } D \in \E, \\
a_0\ &\text { if } D \not \in \E.
\end{cases}
\]
The function $c^{\succ}_{\E}$ is a choice function that maximizes $\succ$ when $D \in \E$ and otherwise chooses $a_0$. Now define a degenerate stochastic choice function $\rho^{\succ}_{\E}$ that corresponds to $c^{\succ}_{\E}$. That is, for all $D \in \D$ and $a \in D$,
\begin{equation}\label{rho_succ_e}
\rho^{\succ}_{\E}(D,a)=1(c^{\succ}_{\E}(D)=a).
\end{equation}

\begin{remark}\label{rem:ru_rational} $\rho^\succ_\E$ is RU-rational for all $\succ \in \L(\A)$ and $\E \subseteq \D$.\footnote{See  Section \ref{pf:ru_vertex} in the online appendix for the proof.}
 \end{remark}

We can interpret $\rho^{\succ}_{\E}$ as follows. The set $\E$ represents the collection of choice sets in which the agent makes decisions according to the linear order $\succ$. Outside this collection (i.e., on $\E^c$), the agent always selects the outside option $a_0$. This behavior is natural: $\E$ can be viewed as the domain on which the consumer behaves rationally according to $\succ$, while outside $\E$, {\it menu effects} lead the agent to default to $a_0$. For example, in marketing experiments consumers may follow a stable ranking when choosing among familiar products, but when faced with an unfamiliar or unusually large assortment, they often revert to the default option $a_0$.



\begin{theorem} 
\label{thm:vertex}  Assume Richness of $\X$.
Consider the outside option setup (i.e., $\A_N=\{a_0\}$).  A stochastic choice function $\rho$ is RU-rational if and only if   
\[
\rho \in \co\{\rho^{\succ}_{\E} \mid \succ \in \L(\A) \text{ and } \E \subseteq  \D \}. 
\]
The set is called the {\it RU polytope}. Moreover, each $\rho^{\succ}_{\E}$ is a vertex of the RU polytope. 
\end{theorem}

The proof of the theorem is in Section \ref{sec:proof_vertex} of the appendix. Mathematically, the theorem provides a vertex characterization of RU-rational stochastic choice functions, whereas Theorem~\ref{thm:1} offers a hyperplane characterization.\footnote{Partial-RU rationality can be characterized by linear inequalities--the non-negativity of BM polynomials when $\D=2^{\A}\setminus \{\emptyset\}$. } The vertex characterization shows that RU-rationality permits a type of behavior excluded under ARU-rationality: the agent may sometimes choose the outside option (the default option) $a_0$, possibly due to menu effects, while in other cases adhering to a rational ranking.

To explicitly compare the RU polytope and the set of stochastic choice functions that are ARU-rational, we first define the latter set as follows denoted by $\mathcal{P}_{\A}$:


\begin{definition}\label{def:ARUM_polytope}
$\mathcal{P}_{\A} = \{ \rhoa : \rhoa= \sum_{\succ \in \L(\A)} \muA(\succ) \rho^{\succ}_{\D} \text{ for some } \muA
\in \Delta(\mathcal{L}(\A)) \}$, where $\Delta(\mathcal{L}(\A))$ is the set of probability distributions over $\mathcal{L}(\A)$ and $\rho^{\succ}_{\D}$ is defined by (\ref{rho_succ_e}). 
\end{definition}

\begin{remark}\label{rem:ARUM_polytope}
The set $\mathcal{P}_{\A}$ is a polytope with vertices $\rhoa^{\succ}_{\D}$, or $
\mathcal{P}_{\A}= \co\{\rho^{\succ}_{\D}| \succ \in \L (\A)\}$. We refer to this set as the {\it ARU polytope}.
\end{remark}

\begin{corollary}\label{coro:vertex}
The RU polytope contains the ARU polytope ${\cal P}_{\A}$ as a subpolytope.\footnote{A subpolytope is a polytope obtained by taking a subset of the vertices of the original polytope and forming their convex hull. 
}
\end{corollary}

Theorem~\ref{thm:vertex} and Corollary~\ref{coro:vertex} together show that the ARU polytope is strictly smaller than the RU polytope: it is a subpolytope whose vertices form only a small subset of the much larger collection of RU vertices. Specifically, the vertices of the RU polytope are given by $\rho^{\succ}_{\E}$ for all $\succ \in \L(\A)$ and $\E \subseteq \D$, while the vertices of the ARU polytope are given by $\rho^{\succ}_{\D}$ for all $\succ \in \L(\A)
$ (i.e., $\E=\D$). Thus, ARU vertices vary only with the ranking $\succ$, whereas RU vertices vary with both $\succ$ and the choice of $\E$. See Figure \ref{fig:ARU_polytope_approx} in the next section for an illustration. The next remark quantifies this gap:

\begin{remark}\label{rem:vertex}
The number of vertices (i.e., the number of $\rho^{\succ}_{\E}$) of the RU polytope is double-exponentially larger than the number of vertices of the ARU polytope (i.e., the number of $\rho^{\succ}_{2^{\A}}$), provided that $\D$ is sufficiently rich. For example, if $\D$ contains all sets that include $a_0$, then the ratio of the number of vertices in the RU polytope to the number in the ARU polytope is at least $2^{(2^{n} - {n \choose 2} -1)}/(n+1)$, where $|\A_{A}|=n$. Since the dominant term in the numerator is double exponential, the ratio grows extremely rapidly. In fact, when $n \ge 6$, the ratio already exceeds $2^{2^{\,n-1}}$.
\end{remark}

Theorem \ref{thm:vertex} and Remark \ref {rem:ARUM_polytope} and  \ref{rem:vertex} together show that  RU-rationality imposes far weaker restrictions than ARU-rationality, which is consistent with the main message of Theorem \ref{thm:1}.

\subsection{Exogenously given aggregation correspondences}\label{sec:howmanyx}
\vspace{-0.2cm}

In the previous sections, we impose no restrictions on the aggregation correspondence $X(\cdot)$, reflecting the practical difficulty of enumerating all underlying alternatives represented by a non-atomic aggregate $a$. The framework, however, extends directly to the case in which the aggregation correspondence is fixed exogenously. This is relevant when the analyst knows the possible underling alternatives that each aggregates  may represent but does not observe which of them are available to consumers in each market. For example, in a study of demand for specific automobile models, the outside option---not purchasing a car---may encompass known transportation alternatives such as public transit, cycling, and ride-sharing. The analyst can enumerate these alternatives, yet their availability varies across markets: some cities lack a subway system, others have no bike-sharing program.

For simplicity, we focus on the case in which there is only one non-atomic aggregate (i.e., $\A_N=\{a_0\}$). Fix an aggregation correspondence $X(a_0)$. We study the set of stochastic choice functions that are RU-rational under this restriction.  Formally, a stochastic choice function $\rho$ is \emph{RU-rational given $X$} if there exist a preference distribution $\mu_{\X}$ and a composition distribution $\lambda$ that are consistent with the fixed aggregation correspondence $X$ such that the pair $(\mu_{\X},\lambda)$ rationalizes $\rho$.

As we do not make any assumptions on preferences over the underlying alternatives, the RU-rational stochastic choice functions depend only on the size of $X(a_0)$; thus we denote the set by $RU(n)$, where $n=|X(a_0)|$. Thus, the RU-polytope can be written as $\bigcup_{n=2}^\infty RU(n)$ since in the previous section  we do not fix $X$. 

\begin{theorem}\label{thm:nests} Consider the outside setup (i.e., $\A_N = \{a_0\}$) and suppose that $\D = 2^{\A}\setminus \{\emptyset\}$.
    Then,
\begin{equation}
ARU \subsetneq RU(2) \subsetneq \cdots \subsetneq RU(|\A_A|) \subsetneq \underbrace{RU(|\A_A|+1) = RU(|\A_A|+2) = \cdots}_{\text{Limited\ Monotonicity + Partial RU-rationality}}.  \end{equation}
\end{theorem}

Theorem \ref{thm:nests} shows that $RU(n)$ becomes larger as $n$ increases and when $n$ becomes larger than or equal to $|\A_A|+1$, $RU(n)$ coincides with the set of the stochastic choice functions that are characterized by Limited Monotonicity and Partial-RU rationalizability.\footnote{For intuition on why the threshold is $|\A_A|+1$, see Remark \ref{remark:threshold} in Section \ref{sec:remarks} in the online appendix.} As $n$ increases, the set $RU(n)$ becomes larger as the set of preferences over $X(a_0)$ as well as the set of composition distributions become richer, which in turn generate richer  choice behavior over $\A$. However, once $n$ exceeds $|\A_A|+1$, the set stabilizes and is fully characterized by the two conditions stated in Theorem \ref{thm:1}. Thus
 \[
 \bigcup_{n=2}^\infty RU(n) = RU(|\A_A|+1)=\{\rho \text{ satisfies Limited Monotonicity \& Partial RU-rationality}\}.
 \]
In practice, the outside option usually represents a broad heterogeneous collection of underlying alternatives—such as staying home, consuming other products, or allocating resources elsewhere—so it is reasonable to expect that the set $X(a_0)$ is large relative to the number of explicitly modeled alternatives (i.e., $|\A_A|$). Thus, in practice, the implication of RUM can be considered to be Limited  Monotonicity and Partial RU-rationality.

On the other hand, when the analyst knows $n\equiv |X(a_0)|$ is less than $|\A_A|+1$, the sets of feasible preferences over $X(a_0)$ and composition distributions are both smaller, rendering some behaviors not RU-rationalizable. The next remark shows that $RU(n)$ becomes non-convex: although the vertices $\rho^{\succ}_{\E}$ remain RU-rationalizable, some of their convex combinations do not.

\begin{remark}\label{rem:non-convex}
For any integer $n$ such that $2 \le n <|\A_A|+1$,
\begin{itemize}
\item[(i)] (deterministic menu-effect) $\rho^{\succ}_{\E} \in RU(n)$ for all $\succ \in \L$ and $\E \subseteq \D$,
\item[(ii)] (non-convexity) $RU(n)$ is not convex.
\end{itemize}
\end{remark}

\begin{figure}[h]
\begin{center}
\begin{tikzpicture}[scale=0.8,transform shape, mystyle/.style={draw,shape=circle,fill=black, inner sep=0pt, minimum size=4pt}]

\def\ngon{20}

\node[draw, regular polygon,regular polygon sides=\ngon,minimum size=6cm, fill=blue!20, opacity=0.7] (p) {};

\foreach\x in {1,...,\ngon}{
    \node[mystyle] (p\x) at (p.corner \x){};
}

\foreach \i in {1,...,\ngon}{
  \foreach \j in {1,...,\ngon}{
    \ifnum\i<\j
      \draw[very thin, gray!60] (p\i) -- (p\j);
    \fi
  }
}

\draw[red, thick, fill=red!30, opacity=0.5] (p.corner 1) -- (p.corner 11) -- (p.corner 12) -- cycle;

\foreach \x in {2,3,4,5,6,7,8,9,10,13,14,15,16,17,18,19,20} {
    \pgfmathsetmacro\angle{90+360/\ngon*(\x-1)}
    \node at ($(p.corner \x)+(\angle:0.4cm)$) {$\rho^{\succ_{\x}}_{\mathcal E_{\x}}$};
}

\pgfmathsetmacro\angleA{90+360/\ngon*(1-1)}
\pgfmathsetmacro\angleB{90+360/\ngon*(11-1)}
\pgfmathsetmacro\angleC{90+360/\ngon*(12-1)}

\node at ($(p.corner 1)+(\angleA:0.4cm)$) {$\rho^{\succ_1}_{\D}$};
\node at ($(p.corner 11)+(\angleB:0.4cm)$) {$\rho^{\succ_2}_{\D}$};
\node at ($(p.corner 12)+(\angleC:0.4cm)$) {$\rho^{\succ_3}_{\D}$};

\coordinate (q) at ($(p.corner 6)!0.2!(p.corner 16)$);

\path (p.corner 7);
\coordinate (proj) at ($(p.corner 6)!0.5!(p.corner 16)$);

\node[circle,fill=black,inner sep=1pt,label=left:$\rho$] at (q) {};

\end{tikzpicture}
\caption{RU polytope (blue), ARU polytope ${\mathcal P}_{\A}$ (red), and RU(3). RU(3) contains all vertices (i.e., $\rho^{\succ_i}_{\E_i}$) and all line segments connecting two vertices.}\label{fig:ARU_polytope_approx}
\end{center}
\end{figure}

See Section~\ref{sec:non-convexproof} in the online appendix for the proof. Remark~\ref{rem:non-convex}~(ii) implies that $RU(n)$ admits no half-space representation; in particular, it cannot be characterized by axioms consisting solely of linear inequalities. Despite this, we can establish the following result:

\begin{theorem}\label{thm:errorbound}
Fix any integer $n$ such that $2 < n <|\A_A|+1$. For all $\rho$ satisfying Limited Monotonicity and Partial-RU rationality, there exists $\rho' \in RU(n)$ such that,
\[
\frac{\|\rho- \rho'\|^2}{|\D|} \le \frac{1}{n-1}.
\]
\end{theorem}

The theorem implies that once $n$ is sufficiently large (e.g., $n\ge 11$), the class $RU(n)$ is highly rich and {\it close} to the set of stochastic choice functions characterized by Limited Monotonicity and Partial-RU rationality: for any $\rho$ satisfying the two axioms, there exists $\rho'\in RU(n)$ whose {\it average squared distance} from $\rho$ is below ten percent. Therefore, when $n$ is sufficiently large, we can make the same conclusion as in the previous sections that the RU-rationality imposes much weaker restrictions than ARU-rationality. 

The theorem follows by applying an approximate version of Carath\'eodory's theorem. (See Theorem \ref{thm:approxcaratheodory} in the Appendix.) The key step is the following lemma:

\begin{lemma}\label{lem:subpolys}
For any sequence $(\succ_i,\E_i)_{i=1}^{n-1}$ with $n\ge 2$, any convex combination among $\{\rho^{\succ_i}_{\E_i}\}^{n-1}_{i=1}$ belongs to  $RU(n)$.
\end{lemma}

See Section \ref{sec:caratheodoryproof} in the appendix for the proof of the lemma and the theorem.\footnote{In Remark \ref{rem:simplexintuition} of Section \ref{sec:remarks} of 
the online appendix, we provide an intuitive explanation why  the lemma holds. }
Figure~\ref{fig:ARU_polytope_approx} illustrates the intuition for Theorem~\ref{thm:errorbound}. The set $RU(3)$ contains all vertices (i.e., the deterministic menu-effect choice functions $\rho^{\succ_i}_{\E_i}$) and, moreover, every line segment connecting any two vertices. Consequently, for any RU-rational stochastic choice function $\rho$, there exists a point in $RU(3)$ arbitrarily close to $\rho$.

\vspace{-0.2cm}

\subsection{Multiple non-atomic aggregates}\label{sec:multiple_cat}

\vspace{-0.2cm}

In this section, we consider the case in which there exist multiple non-atomic aggregates. The main results from the single non-atomic aggregate case continue to hold in essentially the same form, as we show below:

\begin{theorem}\label{thm:1multiple} Assume Richness of $\X$. A stochastic choice function $\rho$ is RU-rational if and only if $\rho$ satisfies the following two conditions:
     \begin{itemize}
\item[(i)] (Limited Monotonicity): if $b \in D \subseteq \A_A$ and $E \subseteq \A_N$ then $\rhoa(D, b) \ge \rhoa(D \cup E,b)$,
\item[(ii)] (Partial RU-rationality) $\rhoa$ is partially RU rational.
    \end{itemize}
\end{theorem}

The only substantive difference from Theorem~\ref{thm:1} is that the condition (i) becomes stronger. In the earlier theorem, the condition (i) required monotonicity with respect to the addition of a single aggregate---namely, the outside option $a_0$. In the current theorem, where multiple non-atomic aggregates are present, the condition (i) instead requires monotonicity with respect to the addition of arbitrary subsets $E$ of the non-atomic aggregates. This is a natural extension of Theorem~\ref{thm:1}: Theorem~\ref{thm:1} corresponds to the case in which $\A_N$ is singleton $\{a_0\}$. Therefore, all of the implications of Theorem~\ref{thm:1}, as explained in the previous section hold with Theorem \ref{thm:1multiple}. We provide  the formal proof in the appendix.

In addition to Theorem \ref{thm:1}, all of the theorems (Theorem \ref{thm:vertex}, \ref{thm:nests}, \ref{thm:errorbound}) in this section generalize to the multiple aggregated alternative case, with appropriate modifications. In Section \ref{sec:multiplegeneralizations} of the online appendix, we provide all statements and proofs of the generalized results. The vertex characterization is formalized in Theorem \ref{thm:vertexmultiple}. In the general case, the vertices of the RU-polytope are deterministic choice functions that are rationalized by a linear order $\succ$ over aggregates, with the additional behavior that the agent may arbitrarily deviate to any non-atomic aggregate. In Theorem \ref{thm:multiplenests}, we show that the results on exogenous aggregation correspondences also generalize in a natural way. As the number of underlying alternatives for each aggregate increase, the RU-rationalizable set expands, up until a threshold at which point Limited Monotonicity and Partial-RU characterize the set. Finally,  in Theorem \ref{thm:errorboundmultiple},  we show that as long as the number of underlying alternatives is relatively large, Limited Monotonicity and Partial-RU approximate the RU-rationalizable set even below the threshold. Therefore, even with multiple non-atomic aggregates, RU-rationality alone does not justify the use of ARMs.

\vspace{-0.2cm}

\section{
Conditions restoring ARUM 
}\label{sec:restrictions}

\vspace{-0.2cm}

This section investigates the conditions under which RU-rationality implies ARU-rationality.  
These conditions are especially relevant given the widespread use of ARUMs in empirical IO and the common assumption that the data-generating process satisfies RU-rationality. In particular, we provide restrictions on the underlying preference distribution~$\muX$ and the composition distribution~$\lambda$, each of which ensures RU-rationality and ARU-rationality coincide on the observed choice domain~$\A$. The two conditions we introduce are independent: one applies to~$\muX$ (Section~\ref{sec:mu}) and the other to~$\lambda$ (Section~\ref{sec:la}).

\vspace{-0.2cm}
 
\subsection{Condition on $\muX$:  non-overlappingness}\label{sec:mu}

\vspace{-0.2cm}

 In this section we provide a restriction on the distribution of preferences that recovers the implications of ARUM. 

\begin{definition}
Given an aggregation correspondence $X$, a distribution $\muX\in \Delta (\L(\X))$ is non-overlapping  if for any $\succ\in supp (\muX)$, $a \in \A_N$, and $y,z \in X(a)$,
\[
y \succ x \succ z \implies x \in X(a).
\]
\end{definition}

\noindent The condition requires that for each aggregate $a$, the underlying alternatives in $X(a)$ occupy adjacent positions in every preference ranking in the support of $\mu_{\X}$.
\footnote{It  is possible to weaken the definition  as follows but we keep this condition for simplicity: A distribution $\muX \in \Delta (\L(\X))$ is almost non-overlapping if for any $\succ \in supp (\muX)$, there exists $b,c \in \A$ such that (i) $y \succ x$ for all $x \in X(b) \cup X(c)$ and $y \in \X \setminus (X(b) \cup X(c))$ and (ii) for all  $a \in \A_N \setminus \{b,c\}$, and $y,z \in X(a)$, $y \succ x \succ z \implies x \in X(a)$.}  


\begin{proposition}
\label{thm:non-overlapping}
A stochastic choice function $\rhoa$ is RU rational with a non-overlapping $\muX$ if and only if $\rhoa$ is ARU rational.
\end{proposition}

Proposition \ref{thm:non-overlapping} shows that, without any assumptions on $\lambda$,  the use of ARUM is valid  when the underlying alternatives  represented by each aggregate occupy adjacent positions in agents’ preference rankings. Consequently, when constructing an aggregated dataset, alternatives with closely related characteristics should not be assigned to different aggregates. A similar observation has been made empirically by \citet{stafford2018fishing} using a dataset of lobster fishing. In particular, she demonstrated  that classifying crab-fishing sites as part of the outside option leads to biased coefficient estimates. Because lobster and crab fishing are close substitutes, an excellent lobster site is preferred to a poor crab site and vice-versa, their rankings overlap, thereby violating the non-overlapping condition specified in Proposition \ref{thm:non-overlapping}.


\vspace{-0.2cm}

\subsection{Condition on $\lambda$: menu-independence}\label{sec:la}

\vspace{-0.2cm}

The second condition is on the composition distribution:

\begin{definition}\label{def:independent} Given an aggregation correspondence $X$,
a composition distribution $(\lambda_B)_{B \in \D}$ is said to be {\it  menu independent} if there exists an unconditional distribution $\lambda\in \Delta\Big( \prod_{a \in \A} (2^{X(a)}\setminus \{\emptyset\})\Big)$ that the marginal distribution of $\lambda$ on $\prod_{a \in B} (
2^{X(a)} \setminus \{\emptyset\}$ coincides with $\lambda_B$ for all $B \in \A$.\footnote{When $(\lambda_B)_{B \in \D}$ is menu-independent and when $\A \in \D$, the unconditional distribution must coincides with $\lambda_\A$. Thus the menu-independence condition can be written equivalently as the marginal distribution of $\lambda_{\A}$ on $\prod_{a \in B} (
2^{X(a)} \setminus \{\emptyset\})$ coincides with $\lambda_B$ for all $B \in \A$.} \end{definition}



\begin{remark}
The menu-independence requires that the conditional distribution must be consistent across menus. If $(\lambda_B)_{B \in \D}$ is menu-independent, then for any $A, B \in \D$, the marginal distributions $\lambda_A$ and $\lambda_B$  on $\prod_{a \in A\cap B} (
2^{X(a)} \setminus \{\emptyset\})$ must coincide.  In particular,  $\lambda_A = \lambda_B$ whenever $A \cap \A_N = B \cap \A_N$.

On the other hand, the menu-independence allows correlation across aggregates. To see this, consider a stronger condition that does not allow the correlation:   there exists  $\lambda^{a} \in \Delta(2^{X(a)}\setminus \{\emptyset\})$ for each aggregate $a \in {\mathscr A}$ such that for any $B\subseteq {\mathscr A}$ and $(S_a)_{a \in B} \in \prod_{a \in B} (
2^{X(a)} \setminus \{\emptyset\})$, $\lambda_{B}((S_a)_{a \in B})= \prod_{a \in B} \lambda^{a}(S_a)$.\footnote{For example, consider $a, b \in \A_N$ and let $X(a)=\{x_1, y_1\}$ and $X(b)=\{x_2, y_2\}$. Definition \ref{def:independent}  allows the situation that $a$ means alternative $x_1$ ($y_1$) if and only if $b$ means alternative $x_2$  ($y_2$, respectively); and  $a$ means $x_1$ and $y_1$ equally likely.  In this example, the probability that $a$ means $x_1$ and $b$ means $y_2$ is zero, which is not allowed in the stronger menu-independence condition.}
 The two conditions coincide when there is only one non-atomic  aggregate. In that case, both reduce to the following: $\lambda_D=\lambda_E$ for all $D, E \subseteq \A$.

\end{remark}

\begin{proposition}  \label{thm:independent}

A stochastic choice function $\rhoa$ is RU rational with menu independent $\lambda$ if and only if $\rhoa$ is ARU rational.
\end{proposition}


Proposition \ref{thm:independent} shows that, without any assumption on $\mu_{\X}$, if the composition distribution~$\lambda$ is menu independent, then each aggregate can be treated as a single good. Whether this assumption holds depends on the setting. In some markets, $\lambda$ may indeed be menu independent. A natural example is automobiles: when alternatives are aggregated by brand, the distribution of models within a brand is relatively stable, conditional on that brand being available, and thus invariant to which other brands are offered. In other markets, however, $\lambda$ is menu dependent. As discussed, the outside option often varies across markets (e.g., with income, store access, or local amenities), so its composition—and hence $\lambda_A$—changes with the menu $A$.

\vspace{-0.2cm}

\subsection{How to define outside option?}

\vspace{-0.2cm}

The previous two propositions provide guidance on how to define the outside option. Specifically, an underlying alternative $x$ should be excluded from the outside option and instead modeled as a separate aggregate if either:\\
(1) the desirability of $x$ is very different from the other underlying alternatives that the outside option $a_0$ represents, since this would violate non-overlappingness (Proposition~\ref{thm:non-overlapping}); or\\
(2) the availability of $x$ varies across markets, since including it in the outside option would violate menu-independence (Proposition~\ref{thm:independent}).


\vspace{-0.2cm}

\section{Simulation}\label{sec:simulation}

\vspace{-0.2cm}

The theoretical results from the previous section show that, in general, the implications of RU-rationality for the observed dataset~$\rhoa$ are too weak to guarantee ARU-rationality except under two special conditions: when the composition distribution~$\lambda$ is menu independent, or when the underlying preference structure is non-overlapping.  


Our theoretical results suggest that if $\lambda$ is sufficiently menu dependent and $\mu_\X$ is sufficiently overlapping, $\rhoa$ is unlikely to be rationalized by an ARUM. In this section, we identify two empirical observations based on the implication of the  theoretical results: {\it if $\lambda$ is sufficiently menu dependent and $\mu_\X$ is sufficiently overlapping, then (a) estimation bias becomes large, and (b) the distance from the RU polytope becomes large.}


We also examine {{\it how large such biases can be}}, an important practical question.  In particular, we demonstrate that {\it the resulting biases are often substantial enough not only to distort estimated utility levels but also to overturn inferred preference orderings.}


\vspace{-0.2cm}

\subsection{Setup}\label{subsec:logitsetup}

\vspace{-0.2cm}

Consider three alternatives $x,y,a_0$, where $x,y$ are atomic aggregates \footnote{For simplicity, we abuse the notation $x$ and $y$ to denote both the aggregates and their respective underlying alternatives.}; $a_0$ is a non-atomic aggregate (the outside option). The outside option can mean good $z$ or $w$. That is, $X(a_0)=\{z, w\}$.\footnote{This means that the set of RU-rational stochastic choice functions is 
$RU(2)$. We adopt this setup for two reasons. First, it permits a two-dimensional graphical representation of our bias and distance results. Second, as shown in Theorem \ref{thm:nests}, the set of RU-rational stochastic choice functions is smallest in $RU(2)$, although it remains strictly larger than the ARU polytope. We show that even in this setting, the biases and distances are substantial. With a larger number of underlying alternatives, both the biases and the distances would increase.}
Given fixed utility values $u(x), u(y), u(z)$, and $u(w)$, we generate the true dataset $\rhoX$ over the set $\X$ of underlying alternatives by assuming the logit model: for all $D \subseteq \{x,y,z,w\}$ and $i \in D$, $\rhoX(D,i)= \exp (u(i))/\sum_{j \in D}\exp (u(j))$.

We examine three distinct markets, each represented by one of the following choice sets: $\{x, a_0\}$, $\{y, a_0\}$, and $\{x, y, a_0\}$. Importantly, the composition of the aggregate $a_0$ varies across these markets; that is, $\lambda_{\{x, a_0\}}$, $\lambda_{\{y, a_0\}}$, and $\lambda_{\{x, y, a_0\}}$ over $\{\{w\}, \{z\}, \{w,z\}\}$ are allowed to differ.  With $\rhoX$ and $\lambda$, we construct the observable reduced dataset $\rhoa$  over the set $\A$ of aggregates.  That is  for any $D \subseteq \{x,y\}$ and $i \in D$, we define $\rhoa(D \cup a_0, i)$ as 
\[
\lambda_{D \cup a_0} (\{w\}) \rhoX(D\cup \{w\}, i) +\lambda_{D \cup a_0} (\{z\}) \rhoX(D\cup \{z\}, i)+\lambda_{D \cup a_0} (\{z,w\}) \rhoX(D\cup \{z,w\}, i).
\]

Given the observed aggregated dataset~$\rhoa$, we estimate the utility parameters~$\hat{u}(x)$ and~$\hat{u}(y)$ under the assumption that~$\rhoa$ follows a multinomial logit model defined over $\A$ and that the utility of the outside option is zero, i.e.,~$\exp(u(a_0)) = 1$. Note that in this setup, (i) {\it the composition distribution $\lambda$ becomes more {\it menu dependent} as $\lambda_{\{x, a_0\}}$, $\lambda_{\{y, a_0\}}$, and $\lambda_{\{x, y, a_0\}}$ become more distinct;
(ii) Fixing the true utilities of $u(x)$ and $u(y)$---say,  $u(x) > u(y)$---the preference structure becomes \emph{more overlapping} as $\max\{u(z), u(w)\}$ exceeds $u(x)$ and $\min\{u(z), u(w)\}$ falls below $u(y)$.}

In the following sections, we fix certain parameter values that are not the main focus of the analysis. Section~\ref{sec:additional_fig} of the online appendix reports simulations under alternative parameter values, with qualitatively similar results throughout.

\vspace{-0.2cm}

\subsection{Estimation Bias from ARUM Misspecification}\label{sec:bias}

\vspace{-0.2cm}

We quantify the estimation bias using the following measure:
\begin{equation}\label{eq:bias}
(\hat{u}(x) - \hat{u}(y)) - (u(x) - u(y)),
\end{equation}
where $\hat{u}(x)$ and $\hat{u}(y)$ denote the maximum-likelihood estimates of $u(x)$ and $u(y)$, respectively. This measure captures the distortion in relative utilities between alternatives $x$ and $y$. Remember that in our setup the values of $u(x)$ and $u(y)$ are fixed at $u(x)=2$ and $u(y)=1$. We are interested in the values of the measure as well as the sign of the measure.  When the bias (\ref{eq:bias}) is less than $-1$ it means that in the estimate we have $\hat{u}(y)> \hat{u}(x)$, even though in the true data we have $u(x)> u(y)$. This means that bias affects not only the estimated utility levels but also the inferred preference orderings.

The values of the measure are also meaningful. To see this, consider the case in which $\hat{u}(x) - \hat{u}(y) > u(x) - u(y)$. Then, by taking the exponential of the measure, we have $\frac{\exp(\hat{u}(x))}{\exp(\hat{u}(y))} \Big/ \frac{\exp(u(x))}{\exp(u(y))}$, which represents the ratio of the estimated odds ratio to the true odds ratio of choosing $x$ over $y$. A deviation in this ratio reflects how the estimated utilities---and thus predicted behavior---diverges from the true underlying model.

\vspace{-0.2cm}
\subsubsection{Effect of composition distribution $\lambda$ on bias}\label{subsec:effect_lambda_bias}
\vspace{-0.2cm}

In the first simulation,  focusing on the role of $\lambda$: {\it the bias defined in (\ref{eq:bias}) becomes larger as $\lambda$ becomes more menu dependent (i.e., as $\lambda$ changes across choice sets $\{x,a_0\},\{y,a_0\}$, $\{x,y,a_0\}$).}  As the preference structure is not the focus of the analysis, we fixed the values $u(x)=2$, $u(y)=1$, $u(z)=3, u(w)=0$.

\begin{figure}[ht]
    \centering
    
    \begin{subfigure}[t]{0.48\textwidth}
        \centering
\includegraphics[scale = 0.47]{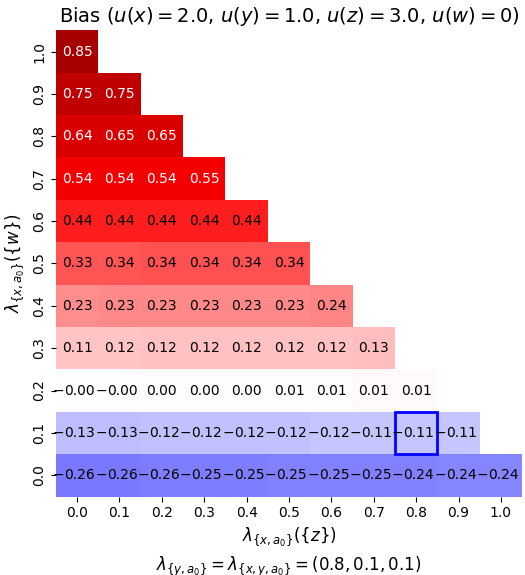}
    \caption{Heatmap of biases across values of $\lambda_{\{x,a_0\}}$, with $\lambda_{\{y,a_0\}}$ and $\lambda_{\{x,y,a_0\}}$ fixed at $(0.8,0.1,0.1)$ and utilities fixed at $u(x)=2$, $u(y)=1$, $u(z)=3$, $u(w)=0$. The horizontal axis reports $\lambda_{\{x,a_0\}}(\{z\})$, the vertical axis reports $\lambda_{\{x,a_0\}}(\{w\})$, and color intensity represents the magnitude of the bias defined in (\ref{eq:bias}). The blue-outlined cell marks the independent case, where $\lambda_{\{x,a_0\}}=\lambda_{\{y,a_0\}}=\lambda_{\{x,y,a_0\}}$; cells farther from the blue cell correspond to more menu-dependent cases, where the bias is larger.}
   
    \label{fig:bias_lam}
    \end{subfigure}
    \hfill
    \begin{subfigure}[t]{0.48\textwidth}
        \centering
    \includegraphics[width=1\linewidth]{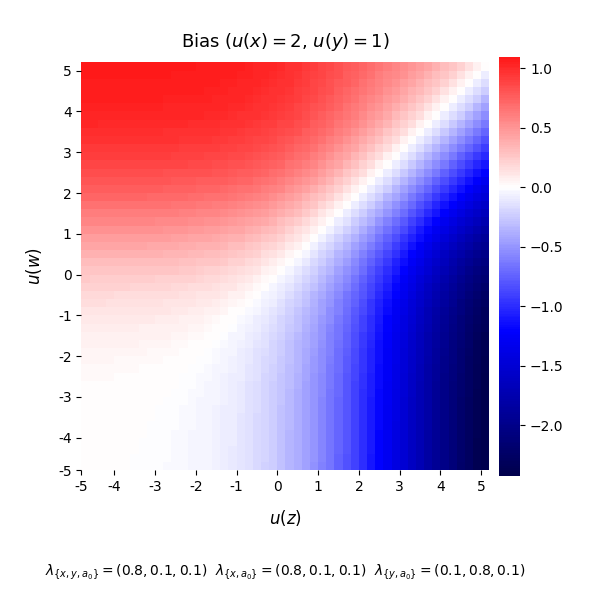}

    \caption{Heatmap of biases across values of $u(z)$ and $u(w)$, with $u(x)=2$, $u(y)=1$, and composition distributions fixed at $\lambda_{\{x,y,a_0\}}=(0.8,0.1,0.1)$, $\lambda_{\{x,a_0\}}=(0.8,0.1,0.1)$, and $\lambda_{\{y,a_0\}}=(0.1,0.8,0.1)$. The horizontal axis reports $u(z)$, the vertical axis reports $u(w)$, and color intensity indicates the magnitude of the bias defined in (\ref{eq:bias}) (red for positive bias, blue for negative bias). Biases are largest when preferences overlap strongly, i.e., when $u(z)$ and $u(w)$ straddle $u(x)=2$ and $u(y)=1$.}
    
    \label{fig:bias_pref_struc}
    \end{subfigure}
    \end{figure}

First, we construct a heatmap to visualize the sensitivity of bias to deviations from the independence case. In Figure \ref{fig:bias_lam}, we vary the values of $\lambda_{\{x,a_0\}}$ while equating and holding $\lambda_{\{x,y,a_0\}}$ and $\lambda_{\{y,a_0\}}$ fixed at $(0.8,\ 0.1,\ 0.1)$, where the first, second, and third entries represent the probabilities assigned to $\{z\}$, $\{w\}$, and $\{z, w\}$, respectively. Note that the horizontal axis measures $\lambda_{\{x,a_0\}}(\{z\})$, while the vertical axis measures $\lambda_{\{x,a_0\}}(\{w\})$. (Note that the numbers $\lambda_{\{x,a_0\}}(\{z\})$ and $\lambda_{\{x,a_0\}}(\{w\})$ determine $\lambda_{\{x,a_0\}}$ uniquely  as $\lambda_{\{x,a_0\}}(\{z,w\})=1-\lambda_{\{x,a_0\}}(\{w\})-\lambda_{\{x,a_0\}}(\{z\})$.) Each grid point in the heatmap represents a specific assignment of $\lambda_{\{x,a_0\}}$; the bias is computed at each grid point. The blue-outlined cell corresponds to the menu-independent case, in which $\lambda_{\{x,a_0\}}(\{z\})=0.8$ and $\lambda_{\{x,a_0\}}(\{w\})=0.1$, implying that $\lambda_{\{x,a_0\}}=\lambda_{\{y,a_0\}}=\lambda_{\{x,y,a_0\}}$.

The heatmap  shows that  bias grows as $\lambda$ becomes more menu dependent—that is, as the cell moves further away from the menu-independent case. The mechanism is straightforward: positive bias arises when $x$ is overvalued relative to $a_0$ in the menu $\{x,a_0\}$, while negative bias arises when $x$ is undervalued relative to $a_0$ in the menu $\{x,a_0\}$. When $a_0$ consists of $\{w\}$, it is relatively unattractive, whereas when it consists of $\{z\}$ or $\{z,w\}$, it is more attractive, since in our setup $u(z)$ is high and $u(w)$ is low. Hence, the large positive bias in the figure occurs when $a_0$ is composed of ${w}$ more often in the menu $\{x,a_0\}$ (e.g., when $\lambda_{\{x,a_0\}}(\{w\})=1$). Conversely, the negative bias arises when $a_0$ is composed of $\{z\}$ or $\{z,w\}$ (e.g., when $\lambda_{\{x,a_0\}}(\{w\})=0$).



\vspace{-0.2cm}

 \subsubsection{Effect of preference structure on bias}\label{subsec:effect_pref_bias}
\vspace{-0.2cm}

In the last simulation on biases, we vary the preference structure to evaluate how overlapping preferences affect bias. We assume that the true unobservable dataset $\rhoX$ is generated by logit models with $u(x)=2$, $u(y)=1$, and various values of $u(z)$ and $u(w)$ on a uniform grid over $[-5,5]$. The composition distributions $\lambda$ are fixed to be menu-dependent, since $\lambda$ is not the focus of this analysis. Specifically, we set
$\lambda_{\{x,y,a_0\}} = (0.8,0.1,0.1)$,
$\lambda_{\{x,a_0\}} = (0.8,0.1,0.1)$, and
$\lambda_{\{y,a_0\}} = (0.1,0.8,0.1)$,
where the three coordinates represent the probabilities of $\{z\}$, $\{w\}$, and $\{z,w\}$, respectively. 

For each pair $(u(z),u(w))$, we compute the bias defined in \eqref{eq:bias}. Since $u(x)=2$ and $u(y)=1$, violations of the non-overlapping condition are more likely when $u(z)$ and $u(w)$ are far apart, with $\max\{u(z),u(w)\}>2$ and $\min\{u(z),u(w)\}<1$.

Figure~\ref{fig:bias_pref_struc} displays the resulting heatmap, with $u(z)$ and $u(w)$ on the axes. Red indicates positive bias at the top left corner and blue indicates negative bias at the bottom right corner, with color intensity reflecting the magnitude.  The figure again confirms that  bias is largest when the preference structure is highly overlapping. In particular, as $u(z)$ and $u(w)$ move farther apart, cells lie farther from the diagonal, and the magnitude of bias increases. Conversely, when both $u(z)$ and $u(w)$ are either very high or very low, the bias decreases.  We provide further explanation in Section \ref{sec:sim_intuition} of the online appendix.

\subsubsection{Maximum and minimum biases and the independent case}\label{subsection:max_min_biases}
\vspace{-0.2cm}

In the previous sections, we fixed values of some $\lambda$s.\footnote{ 
To construct Figure \ref{fig:bias_lam}, we assumed $\lambda_{\{y,a_0\}}=\lambda_{\{x,y,a_0\}}$.} In this section, we vary the all  $\lambda
$s and report the resulting maximum and minimum biases in Figure \ref{fig:max_min_ind_bias}. Each point in the figure corresponds to four possible bias measures, evaluated at  values of $\lambda_{\{x,y,a_0\}}(\{w\})$ and $\lambda_{\{x,y,a_0\}}(\{z\})$ satisfying $\lambda_{\{x,y,a_0\}}(\{w\}) + \lambda_{\{x,y,a_0\}}(\{z\}) \le 1$:\\
 (1) Maximum positive bias, obtained by varying $\lambda_{\{x,a_0\}}$ and $\lambda_{\{y,a_0\}}$ \textcolor{red}{(red)}  \\
(2) Maximum negative bias, obtained by varying $\lambda_{\{x,a_0\}}$ and $\lambda_{\{y,a_0\}}$ \textcolor{violet}{(purple)}\\
(3) Minimum absolute bias, obtained by varying $\lambda_{\{x,a_0\}}$ and $\lambda_{\{y,a_0\}}$ \textcolor{blue}{(blue)}\\  
(4) Independent case, where $\lambda_{\{x,y,a_0\}}=\lambda_{\{x,a_0\}}=\lambda_{\{y,a_0\}}$ \textcolor{teal}{(green)}  

\vspace{-.2em}
\begin{figure}[H]
    \centering   
    \includegraphics[width=0.5\linewidth]{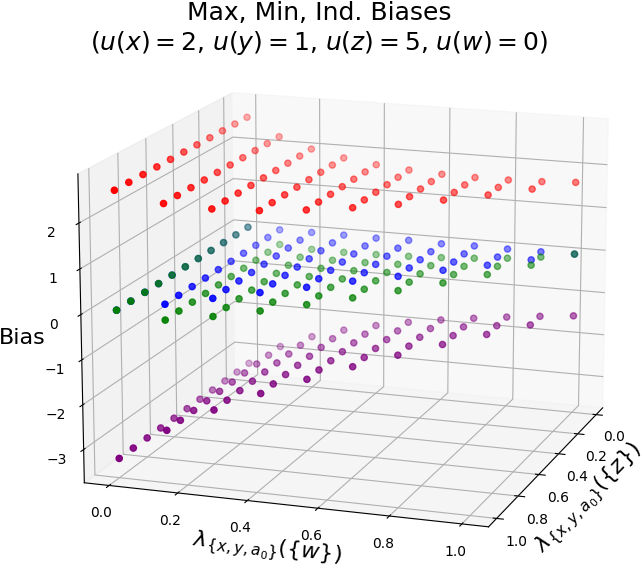}
    \vspace{0em}
    
    \caption{Maximum positive bias (red), maximum negative bias (purple), minimum absolute bias (blue), and bias under menu-independent $\lambda$ (green), plotted across values of $\lambda_{\{x,y,a_0\}}$. The horizontal axis reports $\lambda_{\{x,y,a_0\}}(\{w\})$ and the vertical axis reports $\lambda_{\{x,y,a_0\}}(\{z\})$. The maximum and minimum biases are obtained by optimizing over all admissible values of $\lambda_{\{x,a_0\}}$ and $\lambda_{\{y,a_0\}}$.}
\label{fig:max_min_ind_bias}
\end{figure}

The figure illustrates that, in the menu-independent case, the biases are significantly smaller than the maximum biases and are often close to the minimum biases.\footnote{In fact, when $\lambda_{\{x,y,a_0\}}(\{w\}) = 1$ or $\lambda_{\{x,y,a_0\}}(\{w\}) = 0$, the blue points and the green points coincide. } Notably, both the maximum positive and negative biases can be substantial in magnitude. Specifically, the bias can exceed $2$ or fall below $-3$. As mentioned, when the bias $(\hat{u}(x)-\hat{u}(y)) - (u(x)-u(y))$ is less than $-1$, it implies that the estimated utilities satisfy $\hat{u}(x) < \hat{u}(y)$, despite the true ordering being $u(x) > u(y)$. Moreover, the difference is large enough that $y$ appears substantially better than $x$ in the estimates. This demonstrates that {\it bias distorts not only the estimated utility levels but also the inferred preference orderings, potentially leading to incorrect conclusions and decisions.}

Conversely, when the bias exceeds $2$, we have $\frac{\exp(\hat{u}(x))}{\exp(\hat{u}(y))} \;>\; e^2 \cdot \frac{\exp(u(x))}{\exp(u(y))}$. This means the estimated odds ratio overstates the true odds ratio by a factor greater than $e^2 \approx 7$. Such distortions in relative odds can result in highly inaccurate predictions of choice probabilities.

\vspace{-0.2cm}

\vspace{-0.2cm}
\subsubsection{Biases in parametric setup}
\vspace{-0.2cm}

In the previous sections, we treat each alternative such as $x$ as a label. Some readers may wonder how the results change if we consider covariates vector $x$, where utility is given by $\beta \cdot x$. In general, the results would be similar and large biases would be observed. This is because  any bias on the utility level would be reflected in the estimate of $\beta$. For example, suppose each alternative is a vector containing price and quality. Assume that  $x$ is higher quality and more expensive, and $y$ is lower quality and cheaper. As in the previous sections, assume that the true utility of $x$ is higher than that of $y$, this indicates that  the coefficient $\beta_{quality}$ is larger than the coefficient $\beta_{price}$ in the true utility function.

The previous sections  show that when $\lambda$ is significantly menu-dependent and $\mu_{\X}$ is significantly overlapping, then the estimated utility of $y$ may become larger than that of $x$. This implies that the estimated $\hat{\beta}_{quality}$ is smaller than $\hat{\beta}_{price}$. In this way, we expect that whenever the utility estimate is biased, we should observe similar biases in parameter estimation.

\vspace{-0.2cm}
\subsection{Distance to the ARU polytope}\label{sec:distance}
\vspace{-0.2cm}

The previous subsections  confirmed that estimation biases increase as the composition distribution~$\lambda$ becomes more menu dependent and the preference structure becomes more overlapping. The purpose of this subsection is to explain why it holds. We observe that  {\it the distance from the reduced dataset~$\rhoa$ to the ARU polytope grows as the composition distribution~$\lambda$ becomes more menu dependent or as the underlying preferences become more overlapping.}

To quantify the deviation of  $\rhoa$ from ARU-rationality, we calculate the minimum distance between $\rhoa$ and the set  $\mathcal{P}_{\A}$ of ARUMs (i.e., the ARU polytope):
\begin{equation}\label{eq:distance0}
\min_{\rho' \in {\mathcal P}_{\A}} \|\rhoa-\rho'\|^2.
\end{equation}
See Definition \ref{def:ARUM_polytope}  for the definition of the ARU polytope. Recall that  a dataset $\rhoa$ belongs to $\mathcal{P}_{\A}$ if and only if it is ARU rational. Hence, the minimal distance from $\rhoa$ to $\mathcal{P}_{\A}$ provides a natural metric for assessing the extent of violation of the ARU  rationality.

Figure \ref{fig:projection} shows a point $\rho$ that is RU rational but not ARU rational. The line from $\rho$ to the red triangle shows the minimal Euclidean distance from $\rho$ to the ARU polytope $\mathcal{P}_{\A}$ which is given by the orthogonal projection.

\begin{figure}[h]
\begin{center}
\begin{tikzpicture}
[scale=0.8,transform shape, mystyle/.style={draw,shape=circle,fill=black, inner sep=0pt, minimum size=4pt}]
\def\ngon{20}

\node[draw, regular polygon,regular polygon sides=\ngon,minimum size=6cm, fill=blue!20, opacity=0.5] (p) {};

\foreach\x in {1,...,\ngon}{
    \node[mystyle] (p\x) at (p.corner \x){};
}

\draw[red, thick, fill=red!30, opacity=0.5] (p.corner 1) -- (p.corner 11) -- (p.corner 12) -- cycle;

\foreach \x in {2,3,4,5,6,7,8,9,10,13,14,15,16,17,18,19,20} {
    \pgfmathsetmacro\angle{90+360/\ngon*(\x-1)}
    \node at ($(p.corner \x)+(\angle:0.4cm)$) {$\rho^{\succ_{\x}}_{\mathcal E_{\x}}$};
}

\pgfmathsetmacro\angleA{90+360/\ngon*(1-1)}
\pgfmathsetmacro\angleB{90+360/\ngon*(11-1)}
\pgfmathsetmacro\angleC{90+360/\ngon*(12-1)}

\node at ($(p.corner 1)+(\angleA:0.4cm)$) {$\rho^{\succ_1}_{\D}$};
\node at ($(p.corner 11)+(\angleB:0.4cm)$) {$\rho^{\succ_2}_{\D}$};
\node at ($(p.corner 12)+(\angleC:0.4cm)$) {$\rho^{\succ_3}_{\D}$};

\coordinate (q) at ($(p.corner 6)!0.2!(p.corner 16)$);

\path (p.corner 7);
\coordinate (proj) at ($(p.corner 6)!0.5!(p.corner 16)$);

\draw[thick] (q) -- (proj);

\draw ($(proj)+(-0.2,0.025)$) -- ++(0.029,0.22) -- ++(0.2,-0.029);

\node[circle,fill=black,inner sep=1pt,label=left:$\rho$] at (q) {};

\end{tikzpicture}
\caption{Random utility polytope (blue) and aggregated random utility polytope ${\mathcal P}_{\A}$ (red). The orthogonal projection illustrates the distance from $\rho$ to ${\mathcal P}_{\A}$ \label{fig:projection}}
\end{center}
\end{figure}


To see how the minimum distance can be calculated in our setup, remember that we have three alternatives $x, y, a_0$. Thus the number of linear orders is six; the number of coordinates  of $\rhoa$ (i.e., $(D,i)$ such that $i \in D\subseteq \{x,y,a_0\}$) is twelve. Then, the distance (\ref{eq:distance0}) can be calculated as follows: given $\rhoa \in \Re^{12}$, 
\begin{equation}\label{eq:distance}
\min_{\mu \in \Delta^5} \Big\| \rhoa - \sum_{i=1}^6 \mu_i \rho^{\succ_i}_{\D} \Big\|^2,
\end{equation}
where  $\rho^{\succ_i}_{\D}$ is defined by equation (\ref{rho_succ_e}). Here, $\Delta^5$ is the five dimensional simplex and an element of the simplex satisfies $\mu \in \mathbf{R}^6$, $\mu_i \geq 0$, and $\sum_i \mu_i = 1$.

\vspace{-0.2cm}
\subsubsection{Effect of composition distribution $\lambda$ on distance}\label{sec:ditance_lambda}
\vspace{-0.2cm}

\begin{figure}[ht]
    \centering
    
    \begin{subfigure}[t]{0.48\textwidth}
        \centering
\includegraphics[scale = 0.65]{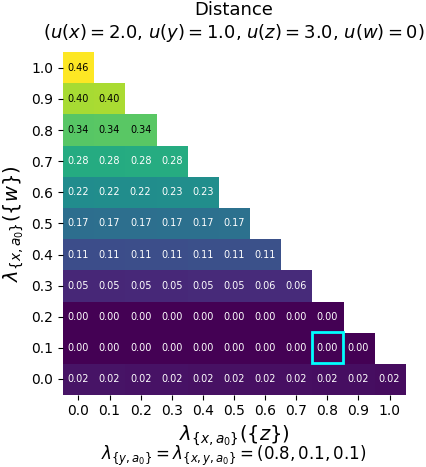}
\caption{Heatmap of the distance defined in (\ref{eq:distance}) across values of $\lambda_{\{x,a_0\}}$, with $\lambda_{\{y,a_0\}}$ and $\lambda_{\{x,y,a_0\}}$ fixed at $(0.8,0.1,0.1)$ and utilities fixed at $u(x)=2$, $u(y)=1$, $u(z)=3$, and $u(w)=0$. The horizontal axis reports $\lambda_{\{x,a_0\}}(\{z\})$, the vertical axis reports $\lambda_{\{x,a_0\}}(\{w\})$, and color intensity indicates the magnitude of the distance. The blue-outlined cell marks the independent case ($\lambda_{\{x,a_0\}}=\lambda_{\{y,a_0\}}=\lambda_{\{x,y,a_0\}}$); cells farther from the blue cell correspond to more menu-dependent cases, where the distance is larger.}\label{fig:distance_lam}
    \end{subfigure}
    \hfill
    \begin{subfigure}[t]{0.48\textwidth}
        \centering
\includegraphics[]{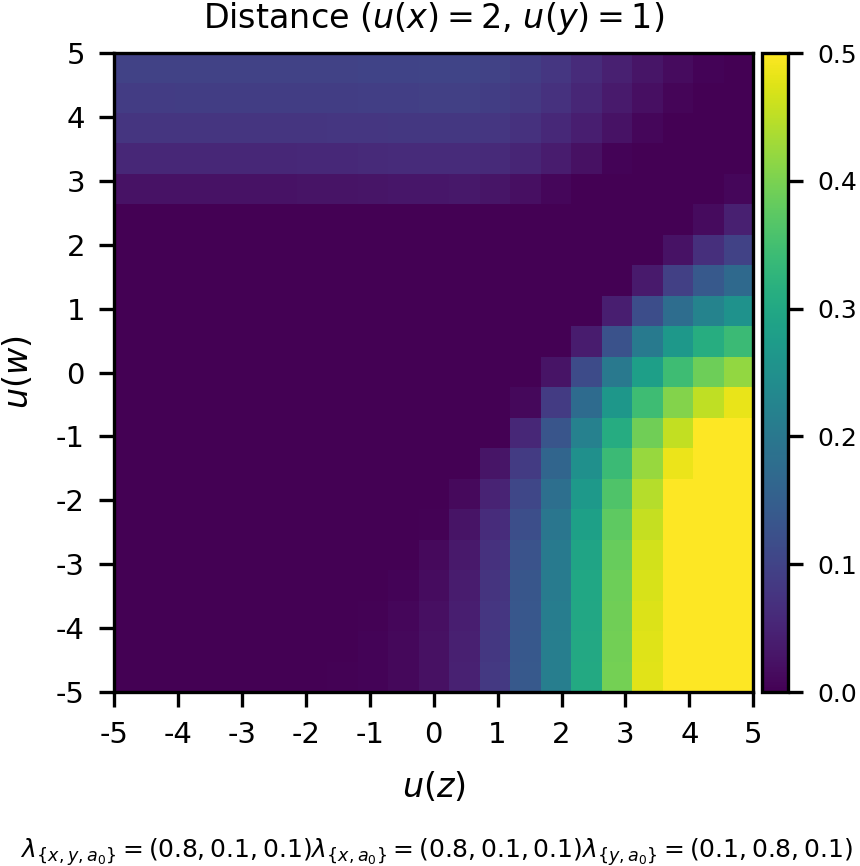}
\caption{Heatmap of the distance defined in (\ref{eq:distance}) across values of $u(z)$ and $u(w)$, with $u(x)=2$, $u(y)=1$, and composition distributions fixed at $\lambda_{\{x,y,a_0\}}=(0.8,0.1,0.1)$, $\lambda_{\{x,a_0\}}=(0.8,0.1,0.1)$, and $\lambda_{\{y,a_0\}}=(0.1,0.8,0.1)$. The horizontal axis reports $u(z)$, the vertical axis reports $u(w)$, and color intensity indicates the magnitude of the distance. The distance is largest in the lower-right corner and smaller along the diagonal.}\label{fig:distance_pref_struc}
    \end{subfigure}
    \end{figure}

In this simulation, we examine how the distance defined in (\ref{eq:distance}) varies with $\lambda$ while holding utility values fixed. We use the same utility and $\lambda$ specifications as in Subsection \ref{subsec:effect_lambda_bias}. Specifically, we vary $\lambda_{\{x,a_0\}}$ while fixing $\lambda_{\{x,y,a_0\}}$ and $\lambda_{\{y,a_0\}}$.

We compute the distance at each grid point. Figure~\ref{fig:distance_lam} reports the resulting heatmap.  Consistent with Proposition \ref{thm:independent}, the distance is zero in the menu-independent case (blue-outlined cell). As $\lambda$ moves farther from the menu independence, the distance increases.  We discuss the intuition underlying Figure~\ref{fig:distance_lam} in Section~\ref{sec:sim_intuition}.

\vspace{-0.2cm}
\subsubsection{Effect of preference structure on distance}\label{subsec:effect_pref_distance}
\vspace{-0.2cm}

Finally, we examine how the distance varies with the preference structure while keeping $\lambda$ fixed. The simulation setup follows Subsection~\ref{subsec:effect_pref_bias}: we fix $u(x)=2$ and $u(y)=1$, and vary $u(z)$ and $u(w)$ over a uniform grid.

We produced a heatmap with axes corresponding to $u(z)$ and $u(w)$ and color intensity represents the magnitude of the distance. Figure~\ref{fig:distance_pref_struc} shows  that the distance is largest in the lower-right corner, where the preference structure exhibits strong overlap, and smallest along the diagonal, where preferences are closer to non-overlapping, which is consistent with Proposition~\ref{thm:non-overlapping}. Section~\ref{sec:sim_intuition} in the online appendix provides further intuition for this figure.

\vspace{-0.2cm}
\section{Summary}
\vspace{-0.2cm}

We study stochastic choice over \emph{aggregates} when the analyst does not observe the underlying alternatives each aggregate represents and the composition can vary across menus. RUM defined over underlying alternatives describes consumer's actual choice among the alternatives available to them instead of analyst-defined aggregates. An ARUM, by contrast, is an empirical tool that enables tractable estimation by treating each aggregate as a primitive object of choice. The central question is whether choice frequencies over aggregates, generated by a RUM on the underlying alternatives, are also consistent with an ARUM. This paper addresses that question when the composition of each aggregate is heterogeneous across consumers, may vary across menus, and is unobservable to the analyst. 

We characterize the testable implications of RU-rationality across three frameworks and show that they are substantially weaker than those of ARU-rationality. In the canonical outside-option environment, Theorem~3.1 establishes that RU-rationality is equivalent to \emph{Limited Monotonicity} plus \emph{Partial RU-rationality}. Theorem~3.2 gives a vertex characterization: the RU-rational set is a polytope with ``menu-effect'' deterministic vertices, and it strictly contains the ARU polytope. When the aggregation correspondence is fixed, Theorem~3.3 shows a strict nesting $\text{ARU} \subsetneq RU(2) \subsetneq \cdots \subsetneq RU(|\A_A|)$, with stabilization at $RU(|\A_A|+1)$, and Theorem~3.4 provides an approximation bound for $RU(n)$ below this threshold. All these results extend to the multiple aggregates case and imply that  RU-rationality is far weaker than ARU-rationality. 

We then identify two independent conditions under which the use of an ARUM is justified. Propositions~4.1--4.2 show that RU-rationality implies ARU-rationality under either \emph{non-overlapping preferences} (within each aggregate, underlying alternatives are adjacent in every ranking in the support) or \emph{menu-independent composition}. These results provide practical guidance on how to construct aggregated datasets in ways that justify the use of ARUM. Finally, simulations based on an underlying logit model quantify the empirical relevance of the gap: both estimation bias and the Euclidean distance to the ARU polytope increase as composition becomes more menu-dependent or preferences become more overlapping, and the resulting biases can be large enough to reverse inferred preference orderings and substantially distort relative odds ratios.

\appendix
\vspace{-0.2cm}
\section{Proofs}
\vspace{-0.2cm}

For ease of notation, we will  assume throughout this section that each aggregate $a \in \A_A$ is composed of itself, that is $X(a)$ contains only $a$. Thus by abuse of notation, we may identify the atomic aggregates $\A_A$ and the corresponding set of underlying alternatives $\cup_{a \in \A_A} X(a)$.

\vspace{-0.2cm}
\subsection{Proof of Theorem \ref{thm:1} and Theorem \ref{thm:1multiple}}
\vspace{-0.2cm}

We provide a proof of Theorem \ref{thm:1multiple}, which implies Theorem \ref{thm:1} as a special case.\footnote{We also prove Theorem \ref{thm:1} independently in Section \ref{sec:altproofthm1} to provide additional intuition for the first step of the construction. The alternate proof uses only $|\A_N|+1$ underlying alternatives belonging to $X(a_0)$. }

The necessity is trivial. We show the sufficiency. Fix an aggregation correspondence $X$ such that $|X(a)| = |\A_A|+2$ for all $a \in \A_N$.  Such an aggregation correspondence exists by Richness of $\X$. For each $a \in \A_N$ and $y \in \A_A$, choose  an element of $X(a)$, denoted by $x_a(y)$. We also fix two alternatives $\overline{x}_a, \underline{x}_a \in X(a)$. Note that there exist such alternatives under the assumption that $|X(a)| = |\A_A| +2$ for all $a\in \A_N$. We again abuse notation and assume that each aggregate $a \in \A_A$ is composed of itself, that is $X(a) = \{a\}$.

By Partial RU-rationality, there exists a distribution $\tilde{\mu}$ over linear orders $\tilde{\succ}$ on $\A_A$ that rationalizes $\rhoa$ on subsets of $\A_A$.  Based on this distribution, we will construct a RUM $\rho$ on the entire set $\X$ of the underlying alternatives as follows. For each ranking on $\A_A$ in the support of $\tilde{\mu}$, we construct ranking on $\X$: we place $x_a(y)$ right above $y$, $\overline{x}_a$ and $\underline{x}_a$ at the bottom with $\overline{x}_a$s above $\underline{x}_a$s. The order among non-atomic aggregates is arbitrarily fixed. For example, let $\{a_0,a_1\}=\A_N$ and fix an order  by index. Also let $\A_A = \{y_0,y_1\}$ and $y_0\ \tilde{\succ}\ y_1$. Then, we define $\succ$ on $\X$ as follows: 
\[
x_0(y_0) \succ x_1(y_0) \succ y_0 \succ x_0(y_1) \succ x_1(y_1) \succ y_1 \succ\overline{x}_0\succ\overline{x}_1 \succ\underline{x}_0 \succ\underline{x}_1.
\]

Fix $D  \subseteq \A_A$ and $E \subseteq \A_N$. In the following, we will construct $\lambda \in \prod_{b \in D \cup E} (2^{X(b)} \setminus \{\emptyset\})$ to establish the following equality for all $b \in D \cup E$: 
\[
\sum_{(S_b)_{b \in D \cup E} \subseteq \prod_{b \in D \cup E} (2^{X(b)} \setminus \{\emptyset\})}\lambda((S_b)_{b \in D \cup E} ) \rhoX(\cup_{b \in D \cup E} S_b, S_b) =\rhoa(D \cup E,b).
\]
The goal is to construct a probability distribution $\lambda$ on $\prod_{b\in E} (2^{X(b)}\setminus \{\emptyset\})$ that satisfies 
\begin{equation}\label{eq:388}
\sum_{(S_b)_{b \in E} \in \prod_{b \in E} (2^{X(b)} \setminus \{\emptyset\})} \lambda((S_b)_{b \in E}) \rhoX(D \cup (\cup_{b\in E} S_b), y)=\rhoa(D\cup E,y) \text{ for all } y \in D,
\end{equation}
and
\begin{equation}\label{eq:392}
\sum_{(S_b)_{b \in E} \in \prod_{b \in E} (2^{X(b)} \setminus \{\emptyset\})} \lambda((S_b)_{b \in E}) \rhoX(D \cup (\cup_{b\in E} S_b), S_b)=\rhoa(D\cup E,b) \text{ for all } b \in E.
\end{equation}

The proof consists of two steps. At the first step, for each $a \in E$ we will construct a $\lambda^a$ that achieves (\ref{eq:388}) and puts the remaining choice probability on $a$. Given $\lambda^a$, in the second step, by taking a convex combination of $\lambda^a$, we will obtain the desired $\lambda$ that satisfy (\ref{eq:388}) and (\ref{eq:392}) for the given $D$ and $E$.

\noindent\textbf{Step 1: } We first label elements of $D$ by $y_i$ so that $\frac{\rhoa(D\cup E,y_0)}{\rhoa(D,y_0)}$ is maximum among $\frac{\rhoa(D\cup E,y_i)}{\rhoa(D,y_i)}$. We will construct a probability distribution $\lambda^a$ on $\prod_{b\in E} (2^{X(b)}\setminus \{\emptyset\})$ recursively. Initially, we  define $\lambda^a_0$ by  
\begin{equation}
\lambda^a_0(\overline{x}_a, (\underline{x})_{E\setminus a}) = \frac{\rhoa(D\cup E,y_0)}{\rhoa(D,y_0)},\quad \lambda^a_0( (X(b))_{b\in E}) = 1-\lambda^a_0(\overline{x}_a, (\underline{x})_{E\setminus a})
\end{equation}
and all other lambdas are zero. By Limited Monotonicity,  we have $\lambda^a_0(\overline{x}_a, (\underline{x})_{E\setminus a})\le 1$, thus $\lambda^a_0$ is a probability measure. Moreover, the probability that  $y_0$ is chosen equals $\rhoa(D \cup E,y_0)$ as desired:
 \begin{align*}
 &\sum_{(S_b)_{b \in E} \in \prod_{b\in E} (2^{X(b)}\setminus \{\emptyset\})} \lambda^a_0((S_b)_{b \in E}) \rhoX(D \cup (\cup_{b\in E}S_b), y_0)\\
&= \lambda^a_0(\overline{x}_a, (\underline{x})_{E\setminus a})  \rhoX(D \cup  \{\overline{x}_a, (\underline{x})_{E\setminus a}\},y_0) + \lambda^a_0((X(b))_{b \in E})   \rhoX(D\cup (\cup_{b \in E} X(b)),y_0)  \\
 &=  \lambda^a_0(\overline{x}_a, (\underline{x})_{E\setminus a})  
 \rhoX(D \cup  \{\overline{x}_a, (\underline{x})_{E\setminus a}\},y_0)
 \quad \quad \quad \quad \quad \quad \quad   (\because x_a(y_0) \succ y_0 \text{ for all }\succ\in supp (\muX)) \\
 &= \lambda^a_0(\overline{x}_a, (\underline{x}_a)_{E\setminus a})  \rhoa(D,y_0) \quad \quad \quad \quad \quad \quad \quad \quad (\because  y_0 \succ\overline{x}_a, \underline{x}_b \text{ for all }b \in E \setminus a \text{ and } \succ\in supp (\muX))\\
 &=\rhoa(D\cup E,y_0).
\end{align*}
Next we define $\lambda^a_1$ based on $\lambda^a_0$ to satisfy the equality for $\rhoa(D\cup E, y_1)$. In general at the $n$th step to obtain the desired equality for $\rhoa(D\cup E, y_n)$, let 
\[c_n = \frac{\rhoa(D\cup E,y_n)}{\rhoa(D,y_n) \lambda^a_0(\overline{x}_a, (\underline{x})_{E\setminus a})}.
\]
Note $0 \le c_n \le 1$ by our choice of $y_0$. We define $\lambda^a_n$ by 
$\lambda^a_n( (X(b))_{b \in E}) = \lambda^a_0((X(b))_{b \in E})$ and for all $(S_b)_{b \in E}\in supp (\lambda^a_{n-1})\setminus (X(b))_{b \in E}$
\[
\lambda^a_n( (S_b)_{b \in E}) = c_n\lambda^a_{n-1}((S_b)_{b \in E}),\quad 
\lambda^a_n(S_a \cup x_a(y_n) \ , \  (S_b)_{b \in E \setminus a}) = (1-c_n)\lambda^a_{n-1}((S_b)_{b \in E}).
\]
Note first that $\lambda^a_n$ is a probability measure given that $\lambda^a_{n-1}$ is a probability measure.   Note also this modification of $\lambda^a_n$ from $\lambda^a_{n-1}$ does not change the chance that $y_i$ is chosen for any $i  \le n-1$; this only changes the chance that $y_n$ is chosen because adding $x_a(y_n)$ changes only the choice frequency of $y_n$. Moreover, we obtain the desired equalities for the choice frequency of $y_n$ as follows:
\begin{align*}
& \sum_{(S_b)_{b \in E} \in \prod_{b \in E} (2^{X(b)}\setminus \{\emptyset\})} \lambda^a_n((S_b)_{b \in E}) \rhoX(D \cup (\cup_{b\in E} S_b), y_n)\\
& =\lambda^a_0((X(b))_{b \in E} )\rhoX(D \cup (\cup_{b\in E} X(b)), y_n) \\
&+
\sum_{(S_b)_{b \in E} \in supp(\lambda^a_{n-1})\setminus (X(b))_{b \in E}} \Big[ c_n\lambda^a_{n-1}((S_b)_{b \in E})\rhoX(D \cup (\cup_{b \in E}S_b) , y_n)\\
&\quad\quad\quad\quad\quad\quad\quad\quad\quad\quad\quad\quad  + (1-c_n)\lambda^a_{n-1}((S_b)_{b \in E})\rhoX(D \cup (\cup_{b \in E}S_b) \cup x_a(y_n), y_n) \Big] \\
&= \sum_{(S_b)_{b \in E} \in supp(\lambda^a_{n-1})\setminus (X(b))_{b\in E}}
c_n\lambda^a_{n-1}((S_b)_{b \in E})\rhoX(D \cup (\cup_{b \in E}S_b) , y_n) \\
&= \sum_{(S_b)_{b \in E} \in supp(\lambda^a_{n-1})\setminus (X(b))_{b \in E}} c_n\lambda^a_{n-1}((S_b)_{b \in E})\rhoa(D , y_n)\\
&=c_n \lambda^a_0(\overline{x}_a, (\underline{x})_{E\setminus a})\rhoa(D , y_n)\\
&= \rhoa(D\cup E,y_n),
\end{align*}
where the second equality holds because $x_a(y_n) \succ y_n$ for all $\succ\in supp (\muX)$; the second to the last equality holds because $\sum_{(S_b)_{b \in E} \in supp(\lambda^a_{n-1})\setminus (X(b))_{b \in E}}\lambda^a_{n-1}((S_b)_{b \in E})=1-\lambda^a_{n-1}((X(b))_{b \in E})=1-\lambda^a_{0}((X(b))_{b \in E})=\lambda^a_{0}(\overline{x}_a, (\underline{x})_{E\setminus a})
$; and finally the third to the last equality holds because $y_n$ is a maximum element among $D$ with respect to $\succ$ if and only if  $y_n$ is a maximum element among $D\cup (\cup_{b \in E} S_b)$ with respect to $\succ$ (Note that this equivalence holds because $S_b\subseteq \{x_b(y_i)| i < n\}\cup\{\underline{x}_b, \overline{x}_b| b \in \A_N\}$ and $x_b(y_i)$ is the immediate predecessor of $y_i$ for all $ \succ\in supp(\muX)$). It follows that $\lambda^a$ satisfies (\ref{eq:388}) as desired: for all $y \in D$
\begin{equation}\label{eq:la_a_1}
\sum_{(S_b)_{b \in E} \in \prod_{b\in E} (2^{X(b)}\setminus \{\emptyset\})}\lambda^a((S_b)_{b \in E})  \rhoX(D\cup (\cup_{b\in E} S_b), y)=\rhoa(D \cup E, y).
\end{equation}

\noindent\textbf{Step 2:}  In the following, given $(\lambda^a)_{a \in \A_N}$, we will construct $\lambda$ that satisfies (\ref{eq:392}) as well as (\ref{eq:388}).  First observe 
 that for all $a \in \A_N$,
\begin{eqnarray}\label{eq:185}
\sum_{(S_b)_{b \in E} \in \prod_{b \in E} (2^{X(b) }\setminus \{\emptyset\})}\lambda^a((S_b)_{b \in E})  \rhoX(D\cup (\cup_{b\in E} S_b),S_b)=
\left\{
\begin{array}{llll}
1-\sum_{y \in D}\rhoa(D\cup E,y) &\text{ if }b=a,\\
0 &\text{ if $b \neq a$}.
\end{array}
\right.
\end{eqnarray}
We denote the left hand side of the above equation as $ \rhoa^a(D\cup E,b)$.  Define 
\[
\lambda= \sum_{a \in E}\lambda^a \Bigg(\dfrac{\rhoa(D\cup E,a)}{1-\sum_{y \in D}\rhoa(D\cup E,y)}\Bigg).
\]
Since $\sum_{a \in E}\rhoa(D\cup E,a)=1-\sum_{y \in D}\rhoa(D\cup E,y)$ and each $\lambda^a$ is a probability distribution on  $\prod_{a \in E}2^{X(a)}\setminus \{\emptyset\}$, so is $\lambda$.  Moreover, we have the desired equalities for choice frequency of each aggregate $b \in E$ as follows:
\begin{align*}
&\sum_{(S_b)_{b \in E} \in \prod_{b\in E} (2^{X(b)}\setminus \{\emptyset\}) }\lambda((S_b)_{b \in E})\rhoX(D\cup (\cup_{b\in E} S_b),S_b)\\
&=\sum_{(S_b)_{b \in E} \in \prod_{b\in E} (2^{X(b)}\setminus \{\emptyset\}) }\Bigg(
\sum_{a \in E}\lambda^a((S_b)_{b \in E})  \dfrac{\rhoa(D\cup E,a)}{1-\sum_{y \in D}\rhoa(D\cup E,y)}\Bigg) \rhoX(D\cup (\cup_{b\in E} S_b),S_b)
\\
&=\sum_{a\in E} \Bigg(\dfrac{\rhoa(D\cup E,a)}{1-\sum_{y \in D}\rhoa(D\cup E,y)}\Bigg)
\sum_{(S_b)_{b \in E} \in \prod_{b\in E} (2^{X(b)}\setminus \{\emptyset\})}\lambda^a((S_b)_{b \in E})  \rhoX(D\cup (\cup_{b\in E} S_b),S_b)\\
&=\sum_{a \in E} \Bigg(\dfrac{\rhoa(D\cup E,a)}{1-\sum_{y \in D}\rhoa(D\cup E,y)}\Bigg)\rhoa^a(D\cup E,b) \quad (\because \text{Definition of }\rhoa^a(D\cup E,b))\\
&= \Bigg(\dfrac{\rhoa(D\cup E,b)}{1-\sum_{y \in D}\rhoa(D\cup E,y)}\Bigg)\rhoa^b(D\cup E,b)+ \Bigg(\sum_{a \in E \setminus b} \dfrac{\rhoa(D\cup E,a)}{1-\sum_{y \in D}\rhoa(D\cup E,y)}\Bigg)\rhoa^a(D\cup E,b)\\
&=\rhoa(D\cup E,b). \quad\quad (\because (\ref{eq:185}))
\end{align*}
Thus, $\lambda$ satisfies (\ref{eq:392}). Finally, to show that $\lambda$ satisfies (\ref{eq:388}), choose  $z \in D$. Then, we have
\begin{align*}
&\sum_{(S_b)_{b \in E} \in \prod_{b\in E} (2^{X(b)}\setminus \{\emptyset\}) }\lambda((S_b)_{b \in E})\rhoX(D\cup (\cup_{b\in E} S_b),z)\\
&=\sum_{(S_b)_{b \in E} \in \prod_{b\in E} (2^{X(b)}\setminus \{\emptyset\}) }\Bigg(
\sum_{a \in E}\lambda^a((S_b)_{b \in E})  \dfrac{\rhoa(D\cup E,a)}{1-\sum_{y \in D}\rhoa(D\cup E,y)}\Bigg) \rhoX(D\cup (\cup_{b\in E} S_b),z)
\\
&=\sum_{a\in E} \Bigg(\dfrac{\rhoa(D\cup E,a)}{1-\sum_{y \in D}\rhoa(D\cup E,y)}\Bigg)
\sum_{(S_b)_{b \in E} \in \prod_{b\in E} (2^{X(b)}\setminus \{\emptyset\})}\lambda^a((S_b)_{b \in E})  \rhoX(D\cup (\cup_{b\in E} S_b),z)\end{align*}
\begin{align*}
&=\sum_{a\in E} \Bigg(\dfrac{\rhoa(D\cup E,a)}{1-\sum_{y \in D}\rhoa(D\cup E,y)}\Bigg)\rhoa(D \cup E, z)\quad (\because (\ref{eq:la_a_1}))\\
&=\rhoa(D \cup E, z),
\end{align*}
where the last equality holds because $\sum_{a\in E} \rhoa(D\cup E,a)=1-\sum_{y \in D}\rhoa(D\cup E,y)$.

\vspace{-0.2cm}

\subsection{Proof of Theorem \ref{thm:vertex}}\label{sec:proof_vertex}
\vspace{-0.2cm}

First we prove that the set of RU rational choice functions coincides with $\co\{\rho^{\succ}_{\E} \mid \succ \in \L(\A) \text{ and } \E \subseteq  \D \}$.
 For all $\succ \in \L(\X)$ and $ (S_{(A,a)})_{a \in A \in \D}  \in \prod_{(A,a): a\in A \in \D}(2^{X(a)} \setminus \{\emptyset\})$, define $\rho^{\succ}_{(S_{(A,a)})_{a \in A\in \D}}$ as follows: 
\[
\rho^{\succ}_{(S_{(A,a)})_{a\in A\in \D}}(B,b) := 1\Big(\exists x \in S_{(B,b)} \forall y \in \cup_{a \in B \setminus b} S_{(B,a)}, [ x \succ y]\Big)
\]
for all $(B,b)$ such that $b \in B \in \D$; thus, $\rho^{\succ}_{(S_{(A,a)})_{a\in A\in \D}}$ is a choice function where on menu $A$, the aggregate $a$ means $S_{(A,a)}$ for all pairs $(A,a)$ such that $a \in A \in \D$.

\noindent\textbf{Step 1}: For each $\succ \in \L(\X)$ and each collection $(S_{(A,a)})_{a \in A \in \D}$, there exists $\succ_\A \in \L(\A)$ and $\E \subseteq \D$ such that $\rho^{\succ}_{(S_{(A,a)})_{a \in A \in \D}}=\rho^{\succ_\A}_{\E}$.

\begin{proof} Fix each $\succ \in \L(\X)$ and each collection $(S_{(A,a)})_{a \in A \in \D}\in \prod_{(A,a): a\in A \in \D}(2^{X(a)} \setminus \{\emptyset\})$, 
define ${\cal F}=\Big\{A \in \D\ \Big|\rho^{\succ}_{(S_{(A,a)})_{a \in A \in \D}}(A,a_0)=1\Big\}$. That is, $\mathcal{F}$ consists of all choice sets $A$ such that $a_0$ is chosen with the probability one. 

Define $\succ_\A \in \L(\A_A)$ as follows: $a \succ_\A b$ if and only if $x \succ y$ for any $a, b\in \A_A$, where $\{x\}=X(a)$ and $\{y\}=X(b)$. Define the rest of the linear order arbitrarily. For any $D \in \D$, if $D \in \E^c \equiv {\cal F}$, the maximum elements of $\succ$ over $D$ belongs to  $S_{(D,a_0)}$; thus $\rho^{\succ}_{(S_{(A,a)})_{a\in A\in \D}}(D,a_0)=1=\rho^{\succ_\A}_{\E}(D,a_0)$. Moreover, if $D \in \E$, the maximum elements of $\succ$ over $D$ does not belong to  $S_{(D,a_0)}$; thus $\rho^{\succ}_{(S_{(A,a)})_{a\in A\in \D}}(D,b)=1 = \rho^{\succ_\A}_{\E}(D,b)$ for some $b \in D\setminus a_0$. It follows that $\rho^{\succ}_{(S_{(A,a)})_{a\in A\in \D}}=\rho^{\succ_\A}_{\E}$
\end{proof}

In the next step we show the special case where the preference distribution is degenerate.

We denote $\rho$ rationalized with $(\muX,\lambda)$ by $\rho_{(\muX,\lambda)}$ in the followings.

\noindent\textbf{Step 2}: Suppose that $\rhoa$ is rationalized by $(\delta_{\succ},\lambda)$.   Then $\rhoa$ is in the convex hull of choice functions of the form $\rho^{\succ}_{(S_{(A,a)})_{a \in A \in \D}}$. 
\begin{proof}
It is easy to see that $\rho_{(\delta_{\succ},\lambda)}$ is linear in $\lambda$, that is, for any $\alpha \in [0,1]$, composition distributions $\lambda$, $\lambda'$ such that $\lambda'' = \alpha\lambda +(1-\alpha)\lambda'$, and $\succ \in \L(\X)$ we have,  $\alpha\rho_{(\delta_{\succ},\lambda)} +(1-\alpha)\rho_{(\delta_{\succ},\lambda')} = \rho_{(\delta_{\succ},\lambda'')}$.

Now notice that $\rho^{\succ}_{(S_{(A,a)})_{a \in A \in \D}} = \rho_{(\delta_{\succ},\lambda)}$ where $\lambda_A((S_{(A,a)})_{a \in A})= 1$ for all $A$. That is, $\lambda$ puts probability $1$ on the outside option meaning $S_{A}$ when the menu is $A$. These are exactly the degenerate composition distributions. Since every $\lambda$ is the convex combination of degenerate $\lambda$s, this concludes the proof of the step.
\end{proof}

\noindent\textbf{Step 3:} Every RU-rational $\rho$ is in the convex combination of choice functions rationalized by $(\delta_\succ, \lambda)$.

\begin{proof}
Similar to Step 2, $\rho_{(\muX,\lambda)}$ is linear in $\muX$ in the sense that if $\alpha\muX+(1-\alpha) \muX' = \muX''$ then   $\alpha\rho_{(\muX,\lambda)}+(1-\alpha)\rho_{(\muX',\lambda)}= \rho_{(\muX'',\lambda)}$ for a fixed $\lambda$. Thus, since every $\muX$ is the convex combination of degenerate preference distributions $\delta_\succ$, it follows that every RU-rational choice function $\rho$ is the convex combination of choice functions rationalized by degenerate preference distributions $\rho_{(\delta_\succ,\lambda)}$ for fixed $\lambda$.
\end{proof}

Combining the three steps, we obtain that the set of RU-rational choice functions is a subset of $\co\{\rho^{\succ}_{\E} |  \succ \in \L(\A)\text{ and }\E \subseteq 2^\A \}$.\footnote{This holds regardless of the size of $X(a_0)$.} Now we show the other inclusion. By  Theorem \ref{thm:1}, it can be shown that the set  of RU-rational choice functions is convex since the set of choice functions satisfying Limited Monotonicity is convex, the set of Partial-RU choice functions is convex, and the intersection of convex sets is convex. Since $\rho^{\succ}_{\E}$ is RU-rational for all $\succ \in \L(\A)$ and $\E \in \D$ (see Remark \ref{rem:ru_rational}), thus its convex combination is also RU rational; thus the set $\co\{\rho^{\succ}_{\E} |  \succ \in \L(\A)\text{ and }\E \subseteq 2^\A \}$ is contained by the set of RU-rational choice functions. Therefore, the RU rational set coincides with  $\co\{\rho^{\succ}_{\E} |  \succ \in \L(\A)\text{ and }\E \subseteq 2^\A \}$.

 To finish the proof of Theorem \ref{thm:vertex}, we must show that each $\rho^\succ_\E$ is a vertex. Notice that $\rho^{\succ}_{\E}$ cannot be represented as a convex combination of the other RU rational $\rho$s as the values of the vector $\rho^{\succ}_{\E}$ are all $0$ and $1$ and all RU rational vectors are nonnegative. This means that all $\rho^{\succ}_{\E}$ are vertices of RU rational polytope. The fact that the RU rational set is $\co\{\rho^{\succ}_{\E} |  \succ \in \L(\A)\text{ and }\E \subseteq 2^\A \}$ imply that there are no other vertices.

\vspace{-0.2cm}
\subsection{Proof of Theorem \ref{thm:nests}}\label{subsect:proof_number}
\vspace{-0.2cm}

The weak inclusions $RU(n-1) \subseteq RU(n)$ are immediate since when $|X(a_0)| = n$, we can restrict $\lambda$ to put zero weight on subsets containing some fixed element of $X(a_0)$. The rest of the proof is is devoted to showing which inclusions are equalities and which inclusions are strict.

\noindent\textbf{Step 1:} For any $n \ge |\A_A|+1$, $RU(n)=RU(n+1)$ and $RU(n)$ coincides with the set of $\rho$ satisfying Limited Monotonicity and Partial RU-rationality.  
\begin{proof}
The set of choice functions satisfying Limited Monotonicity and Partial-RU rationality is equal to $\bigcup_{n=2}^\infty RU(n)$ by Theorem \ref{thm:1}. Furthermore, in the alternate proof of Theorem \ref{thm:1} in Section \ref{sec:altproofthm1} in the online appendix, we shows that it suffices to use $|\A_A|+1$ number of underlying alternatives belonging to $X(a_0)$. Thus, $RU(|\A_A|+1)$ is also equal to the set of choice functions satisfying Limited Monotonicity and Partial-RU rationality. Combining these facts we see that $RU(|\A_A|+1) \supseteq RU(n)$ for all $n$. Combined with the other inclusion, the step follows.
\end{proof}

\noindent \textbf{Step 2: } $RU(|\A_A|) \subsetneq RU(|\A_A|+1)$. 

\begin{proof}  

It suffices to show that there exists $\rho$ that is in $RU(|\A_A|+1)$ but not $RU(|\A_A|) $. We construct it explicitly.

Label the atomic aggregates $\A_A = \{y_1,y_2,...,y_{|\A_A|}\}$. We will abuse notation and refer to the aggregate and the underlying alternative by the same symbol for atomic aggregates. Define $\rho$ in the following way: for nonempty $D \subseteq \A_A$, let $\rho(D,y_i) = \frac{1}{|D|}$ for all $y_i \in D$. Now for each $n$ such that $1 \le n \le |\A_A|$ let $D_n = \{y_1,...,y_n\}$. Let $\rho(D_n \cup \{a_0\} , y_n) = 0$ and $\rho(D_n \cup \{a_0\} , y_i) = \frac{1}{n}$ for $i<n$ and  $\rho(D_n \cup \{a_0\} , a_0) = \frac{1}{n}$ For all other nonempty $D \subseteq \A_A$, let $\rho(D \cup \{a_0\},y_i) = \rho(D,y_i)$ when $y_i \in D$ and $\rho(D \cup \{a_0\},a_0) = 0$

Note that $\rho$ satisfies Limited Monotonicity by the definition; it also satisfies Partial-RU  as $\rho$ is a uniform choice on for all $D \subset \A_A$. Thus $\rho \in RU(|\A_A|+1)$ by Theorem \ref{thm:1}. It only remains to show that $\rho \notin RU(|\A_A|)$.

Suppose for a contradiction that $\rho$ is rationalized by $(\mu_{\X}, \lambda)$ with $|X(a_0)| = |\A_A|$. Since $\rho(\A_A, y_i) = \frac{1}{|\A_A|}$ for all $y_i$, the distribution $\mu_{\X}$ assigns probability $\frac{1}{|\A_A|}$ to the event that $y_i$ is ranked above all other atomic aggregates. In particular, for each $y_i$, there exists a ranking $\succ_i$ in the support of $\mu_{\X}$ such that $y_i = \max_{\A_A} \succ_i$.

Now consider $D_1 = \{y_1\}$. Since $\rho(D_1 \cup \{a_0\}, y_1) = 0$, alternative $y_1$ is never chosen from $\{y_1, a_0\}$. But under $\succ_1$, alternative $y_1$ is ranked above all other atomic aggregates. For $y_1$ to have zero choice probability from $\{y_1, a_0\}$, there must exist some $x_1 \in X(a_0)$ such that $x_1 \succ_1 y_1$ and $\lambda_{\{y_1, a_0\}}$ assigns positive weight to subsets of $X(a_0)$ containing $x_1$.

Next, consider $D_2 = \{y_1, y_2\}$. Since $\rho(D_2 \cup \{a_0\}, y_1) = \rho(D_2, y_1) = \frac{1}{2}$, adding $a_0$ does not reduce the choice probability of $y_1$. This implies that $\lambda_{D_2 \cup \{a_0\}}$ puts zero weight on subsets of $X(a_0)$ containing $x_1$; otherwise, under $\succ_1$, the element $x_1$ would be chosen over $y_1$, reducing its choice probability below $\frac{1}{2}$. At the same time, $\rho(D_2 \cup \{a_0\}, y_2) = 0$. Since $\succ_2$ ranks $y_2$ above all other atomic aggregates, there must exist some $x_2 \in X(a_0)$ with $x_2 \succ_2 y_2$ such that $\lambda_{D_2 \cup \{a_0\}}$ assigns positive weight to subsets containing $x_2$. Moreover, $x_2 \neq x_1$, since $\lambda_{D_2 \cup \{a_0\}}$ puts zero weight on subsets containing $x_1$.

Proceeding inductively, at each step $n$ we obtain a new element $x_n \in X(a_0)$, distinct from $x_1, \ldots, x_{n-1}$, such that $x_n \succ_n y_n$ under some ranking $\succ_n$ in the support of $\mu_{\X}$ where $y_n = \max_{\A_A} \succ_n$. After $|\A_A|$ steps, we have identified $|\A_A|$ distinct elements $x_1, \ldots, x_{|\A_A|}$ of $X(a_0)$. Since $|X(a_0)| = |\A_A|$, these exhaust $X(a_0)$. Therefore, every element of $X(a_0)$ is ranked above some atomic aggregate (and hence above all atomic aggregates under the corresponding ranking) with positive probability under $\mu_{\X}$.

But this implies that $\rho(\{y_2, a_0\}, a_0) > 0$, since with positive probability an element of $X(a_0)$ is ranked above $y_2$. This contradicts $\rho(\{y_2, a_0\}, a_0) = 0$, and we conclude that $\rho \notin RU(|\A_A|)$.

\end{proof}

\noindent\textbf{Step 3: }$RU(M) \subsetneq RU(M+1)$ for any $M$ such that $2 \le M <|\A_A|+1$.

\begin{proof} Fix $M$ such that $2 \le M <|\A_A|+1$. The weak inclusion is again immediate so we just need an example of $\rho \in RU(M+1)$ but $\rho \notin RU(M)$. 


Fix $M$ and partition the elements of $\A_A$ into the first $M$ elements, $D_M = \{y_1,...,y_M\}$ and the remaining elements $D_{-M} = \{y_{M+1},...,y_{|\A_A|}$\}. Now, define $\rho$ restricted to the subsets of $D_M \cup \{a_0\}$ as in the previous example, as if $D_M$ is the set of all atomic aggregates. Consider $\rho$ as in the previous step by restricting the set of all atomic aggregates  to $D_M$. 
By the exact same logic as in Step 2,   $\rho$ is not rationalizable on $D_M \cup \{a_0\}$ when $|X(a_0)| = M$, but it is rationalizable on $D_M \cup \{a_0\}$ when $|X(a_0)|=M+1$. 

It remains to extend $\rho$ to menus that contain elements of $D_{-M}$ while preserving that  $\rho$ belongs to $RU(M+1)$ but not to $RU(M)$. Note that for any extension, we remains to have  $\rho \notin RU(M)$ because if $\rho$ is not rationalizable on the restricted domain, then it is not not rationalizable on the full  domain. 

We extend $\rho$ as follows: If $D$ contains any $y_i \in D_{-M}$, then $\rho$ always chooses the one with the highest index, that is, $\rho(D,y_{j}) = 1$ where $j = \max\{i| y_i \in D\}$. To see exactly why this is rationalizable, note that there exists a distribution over linear orders  over $D_M \cup X(a_0)$ and $\lambda$ that rationalizes $\rho$ on the restricted domain. We simply extend the preference distribution by adding the elements of $D_{-M}$ to the top of each linear order in the support. 
\end{proof}

\vspace{-0.5cm}

\subsection{Proofs of Lemma \ref{lem:subpolys} and Theorem \ref{thm:errorbound} }\label{sec:caratheodoryproof}
\vspace{-0.2cm}

As mentioned in the body of the paper,  we use Approximate Carath\'{e}odory theorem. (See Theorem 0.0.2 of \cite{vershynin}).

\begin{theorem}{(Approximate Carath\'{e}odory)}\label{thm:approxcaratheodory} Consider a set $T \subseteq \mathbf{R}^N$ contained by the unit Euclidean ball. Then for every $x \in co.(T)$ and $k \in \mathbf{N}$ there exists $x_1,...,x_k \in T$ such that $\Big\|x - \frac{1}{k}\sum_{i=1}^k x_i\Big\| \le \frac{1}{\sqrt{k}}$ where $||\cdot||$ is the Euclidean distance.
\end{theorem}

We also use the next result that follows directly from Theorem \ref{thm:1} and \ref{thm:vertex}:

\begin{corollary}\label{coro:vertexhalfspace}
    A stochastic choice function $\rho$ satisfies Limited Monotonicity and Partial-RU if and only if $\rho \in \co\{\rho^{\succ}_{\E} \mid \succ \in \L(\A) \text{ and } \E \subseteq  \D \}$.
\end{corollary}

\vspace{-0.5cm}
\subsubsection{Proof of Lemma \ref{lem:subpolys}}
\vspace{-0.2cm}

Let $n= |X(a_0)|$ and $\rho = \sum_{i=1}^{n-1} \alpha_i \rho^{\succ_i}_{\E_i}$ for some  weights $\sum_{i=1}^{n-1} \alpha_i = 1$. Label $X(a_0) = \{\underline{x},x_1,...,x_{n-1}\}$. Now, for $\succ_i \in \L(\A)$, define $\tilde{\succ}_i \in \L(\X)$ by putting $x_i$ above all other elements of $\X$ (i.e., $x_i$ is the best element) and for all $j \neq i$, put $x_j$ at the bottom of $\X$ (ranked below all other elements of $\X$), with comparisons between $x_j$ and $x_{j'}$ arbitrary for any $j' \neq i, j$. Also, for $a,b \in \A_A$ let $a \ \tilde{\succ}_i \ b$ if and only if $a \succ_i b$. Finally, for $a \in \A_A$, let $\underline{x} \ \tilde{\succ}_i \ a$ if and only if $a_0 \succ_i a$. All other comparisons may be arbitrary.

 In this way, $\tilde{\succ}_i$ is the same as $\succ_i$ except $\underline{x}$ plays the role of $a_0$ and $x_i$ is at the top. Thus if $a_0$ is composed of only $\underline{x}$, then the behavior follows $\succ_i$. If $a_0$ has $x_i$ in its composition, then whenever the agent's preference is $\tilde{\succ}_i$, the agent deviates to $a_0$.

Now let $\mu_\X(\tilde{\succ}_i) = \alpha_i$ for all $i$ and let $\lambda_A(\{\underline{x}\} \cup \{x_i| A \notin \E_i\}) = 1$ for all $A \in \D$. The pair $(\mu_\X,\lambda)$ rationalizes $\rho$. To see this, let $\rho_{(\mu_\X,\lambda)}$ be the RU-rational stochastic choice function represented with$(\mu_\X,\lambda)$ and  let $S_{a_0} = \{\underline{x}\} \cup \{x_i| A \notin \E_i\}$. Observe that for each $A \subseteq \A$ and an atomic aggregate  $a \in A$, $\rho_{(\mu_\X,\lambda)} (A,a) =  \sum_i \alpha_i 1(a \ \tilde{\succ}_i \  x \ \forall x \in S_{a_0}) = \sum_i \alpha_i \rho^{\succ_i}_{\E_i}(A,a)$, where the first equality holds by the definitions of $\mu_{\X}$ and $\lambda$  and the second equality holds by $\rho^{\succ_i}_{\E_i}(A,a)=1(a \ \tilde{\succ}_i \  x \ \forall x \in S_{a_0})$.

\vspace{-0.2cm}
\subsubsection{Proof of Theorem \ref{thm:errorbound}}

Fix $\rho$ satisfying Limited Monotonicty and Partial-RU and integer $n$ satisfying $ 2  < n  < |\A_A|+1$. Let $T$ be the set of $\frac{\rho^\succ_\E}{\sqrt{|\D|}}$ for all $\succ \in \L(\A)$ and $\E \subseteq \D$. It is easy to see that $T$ is contained by the unit ball. Furthermore, by Corollary \ref{coro:vertexhalfspace}, it follows that $\frac{\rho}{\sqrt{|\D|}} \in co.(T)$. By Theorem \ref{thm:approxcaratheodory}, there exists exists a sequence $(\succ_i,\E_i)_{i=1}^{n-1}$ such that, $\big|\big|\frac{\rho}{\sqrt{|\D|}} - \frac{1}{k}\sum_{i=1}^{n-1} \frac{\rho^{\succ_i}_{\E_i}}{\sqrt{|\D|}} \big|\big| \le \frac{1}{\sqrt{n-1}}$. By Lemma \ref{lem:subpolys}, $\rho' \equiv \frac{1}{k}\sum_{i=1}^{n-1} \frac{\rho^{\succ_i}_{\E_i}}{\sqrt{|\D|}}$ is in $RU(n)$. Thus, the inequality simplifies to $\big|\big|\frac{\rho}{\sqrt{|\D|}} - \frac{\rho'}{\sqrt{|\D|}} \big|\big| \le \frac{1}{\sqrt{n-1}}$, which we rewrite simply as, $\frac{1}{\sqrt{|\D|}}||\rho - \rho' || \le \frac{1}{\sqrt{n-1}}$. Squaring both sides gives the desired inequality.

\vspace{-0.2cm}
\subsection{Proof of Proposition \ref{thm:non-overlapping}}
\vspace{-0.2cm}

To prove Proposition \ref{thm:non-overlapping}, we first provide a lemma:

\begin{lemma}\label{lem:non-overlapping} Let $\rhoa$ be a stochastic choice function. (a) Suppose that  $(\muX, \lambda)$  rationalizes $\rhoa$ and $\muX$ is non-overlapping. Then for any composition distribution $\lambda'$, $(\muX, \lambda')$ rationalizes $\rhoa$;\\ (b) Suppose that there is no non-overlapping $\muX$ such that $(\muX, \lambda)$ rationalizes $\rhoa$ for some $\lambda$.  Then, for any composition distribution $\lambda'$, there does not exist non-overlapping $\muX'$ such that $(\muX', \lambda')$ rationalizes $\rhoa$.
\end{lemma}

We provide the proof of the lemma in Section \ref{sec:proof-non-overlapping} of the online appendix. By using the lemma, we now prove the proposition below.

\noindent\textbf{Only if direction:} It suffices to show this for the case where $\muX$ is deterministic. Suppose $\muX$ is  non-overlapping and $\muX(\succ) = 1$ for some $\succ\in \L(\X)$ and fix $\lambda$.  Label the elements of $\A $ as $a_1,a_2,...a_{|\A|}$ in such a way that if $i<j$ then $x \succ y$ for all $x \in X(a_i)$ and $y \in X(a_j)$. This labeling is possible by the non-overlapping condition.\\
Let $\succ_{\A}$  be defined such that $a_1 \ \succ_{\A} \ a_2 \ \succ_{\A} \ \cdots \succ_{\A} a_{|\A|}$.

It is easy to see that if $D \subseteq \A$ contains any element of $\A$ then $\rhoa$ puts probability $1$ on the $a_i$ with the smallest index. Thus, $\muA$ rationalizes $\rhoa$.

\noindent\textbf{If direction:} There exists $\muA \in \Delta ({\mathscr L}( {\mathscr A}))$ that rationalizes $\rhoa$. 
For each $\succ \in \text{supp }\muA$ consider a ranking $\succ'\in {\mathcal L}(\X)$ that extends $\succ$ by $x \succ' y$ if  $x \in X(a)$ and $y \in X(b)$ for some  $a \succ b$ and fill the order among $X(a)$ arbitrarily for all $a$; then define $\muX(\succ')=\muA(\succ)$. Also define $\lambda$ arbitrarily. Since $\muX$ is non-overlapping, it follows from statement (a)  of Lemma \ref{lem:non-overlapping} that  $(\muX,\lambda)$ rationalizes $\rhoa$.

\vspace{-0.2cm}
\subsection{Proof of Proposition \ref{thm:independent}}
\vspace{-0.2cm}

\noindent\textbf{If direction:} If  $\rhoa$ is rationalizable by an ARUM with $\muA$ over $\A$, then as we showed in Proposition \ref{thm:non-overlapping} that there exists $(\muX,\lambda)$ that rationalizes $\rhoa$ and $\muX$ is non-overlapping. Moreover, as proved in Lemma \ref{lem:non-overlapping} (a), for non-overlapping representation any $\lambda$ works; so we can choose an menu-independent $\lambda$ as desired.

\noindent\textbf{Only if direction:} Suppose that $\rhoa$ is RU-rational with $(\muX, \lambda)$, where $\lambda$ is menu-independent.  Let $\rhoX$ be the random utility function  defined on $\X$ associated with $\muX$. When $\lambda$ is menu-independent,  representation (\ref{eq:rationalization}) in Definition \ref{def:rationalization}  can be simplified as in the following lemma.

\begin{lemma}\label{lem:ind} If $\rho$ is rationalized by $(\muX,\lambda)$ where $\lambda$ is menu-independent, then for all $A \subseteq {\mathscr A}$ and $a \in A$, $\rhoa(A,a)=  \sum_{(S_b)_{b \in {\mathscr A}}} 
\lambda ((S_b)_{b \in \A})$ $\rho_{\X}(\cup_{b \in A} S_b, S_a)$, where $\rho_{\X}(\cup_{b \in A} S_b, S_a) = \muX(\succ\in \L (\X) | \exists x \in S_a, \ \forall y \in \cup_{b \in A \setminus a} S_b, \ x \succ y)$.\footnote{$\lambda$ here is the unconditional distribution guaranteed to exists by the menu independence.}
\end{lemma}

We provide the proof of the lemma in Section \ref{sec:proof_ind} in the online appendix.  Lemma \ref{lem:ind} means that the reduced  stochastic choice function $\rhoa$ is a convex combination of $\rho_{\X}$ for each composition of the aggregates. By using  the lemma, we now complete the proof of the only if direction.
 
 Given $(S_a)_{a \in {\mathscr A}} \subseteq \prod _{a \in {\mathscr A}} X(a) $, for each $(B,b)$ such that $b \in B \subseteq {\mathscr A}$ define $\rho_{(S_a)_{a \in {\mathscr A}}}(B,b)=\rhoX(\cup_{a \in B} S_a, S_b)$.\footnote{Note that $\rho_{(S_a)_{a \in {\mathscr A}}}$ is a data set where $a$ always means $S_a$.} 
Given $(S_a)_{a \in {\mathscr A}} \subseteq \prod _{a \in {\mathscr A}} 2^{X(a)}\setminus \{\emptyset\} $, for each $\succ\in \mathcal{L} (\X)$ define $\succ_{(S_a)_{a \in \A}}$ by $a \succ_{(S_a)_{a \in \A}} b$ if and only if there exists $x \in S_a$ such that $x \succ y$ for all $y \in S_b$. Note that $\succ_{(S_a)_{a \in \A}}$ is a linear order on ${\mathscr A}$. Now we define a probability measure $\mu_{(S_a)_{a \in \A}}$ over linear orders on ${\mathscr A}$ as follows: for any linear order $\succ'$ on ${\mathscr A}$,  define 
 $\mu_{(S_a)_{a \in \A}}(\succ') = \muX(\succ \in  {\mathscr L}(\X)|\succ_{(S_a)_{a \in \A}}=\succ' )$.

Note that, by definition, $\mu_{(S_a)_{a \in \A}}$ rationalizes $\rho_{(S_a)_{a \in \A}}$, that is for any $(A,a)$ such that $a \in A \subseteq \A$, $
\rho_{(S_a)_{a \in \A}}(A,a)=\mu_{(S_a)_{a \in \A}}(\succ' \in {\mathscr L }({\mathscr A})|a \succ' b, \text{ for all } b \in A\setminus a)$. Now let $\muA = \sum_{(S_a)_{a \in {\mathscr A}}} \lambda((S_a)_{a \in \A}) \mu_{(S_a)_{a \in \A}}$. Then for all $(A,a)$ such that $a \in A \in \D$
\begin{align*}
\rhoa(A,a) &=  \sum_{(S_a)_{a \in {\mathscr A}}} \lambda((S_a)_{a \in \A})   \rhoX(\cup_{a \in A} S_a, S_a)  \quad\quad\quad(\because \text{Lemma }\ref{lem:ind})\\
&=\sum_{(S_a)_{a \in {\mathscr A}}} \lambda((S_a)_{a \in \A})  \mu_{(S_a)_{a \in \A}}(\succ' \in {\mathscr L }({\mathscr A})|a \succ' b, \text{ for all } b \in A\setminus a)\\
&=\muA(\succ' \in {\mathscr L }({\mathscr A})|a \succ' b , \text{ for all } b \in A\setminus a).
\end{align*}
Thus $\muA$ rationalizes $\rhoa$.

\begin{singlespace}
\bibliographystyle{new_ecma.bst}
\bibliography{discrete_ru}
\end{singlespace}

\newpage

\section*{**For Online Publication**
}

\section*{Online Appendix}

\section{Omitted Proofs}\label{sec:omitted}

\subsection{Proof of Remark \ref{rem:ru_rational}}\label{pf:ru_vertex}
Fix $\succ \in \L(\A)$ and $\E \subseteq \D$. Choose any two elements of $X(a_0)$ denoted by $\overline{x}$ and $\underline{x}$. Consider the following ranking $\succ' \in \L(\X)$ such that (i) $\succ  $ and $\succ'$ coincide on $\A_A$ (i.e., $a \succ b$ if and only if $x \succ' y$ for all $a,b \in \A_{N}$, where $X(a)=\{x\}$ and $X(b)=\{y\}$); (ii) for any $y \in \A_A$, $\overline{x}\succ' y \succ' \underline{x}$. For all $D\in \D$, define $\lambda_D(\{\overline{x}\})=1$ for any $D \not \in \E$; and $\lambda_D(\{\underline{x}\})=1$ for any $D \in \E$. It is easy to see that $(\delta_{\succ'}, \lambda)$ RU- rationalizes $\rho^{\succ}_{\E}$. Note that this construction only relies on $|X(a_0)| \ge 2$ which is implied by the outside option approach, thus Remark \ref{rem:ru_rational} does not depend on $X(a_0)$.

\subsection{Proof of Remark \ref{rem:vertex}}

We obtain the lower bound of the number of vertices. Fix $\E$ that contains all subsets of $\A$ that contains $a_0$. Note that there are $2^{(2^{|\A_A|} - {|\A_A| \choose 2} -1)}$ distinct such collection $\E$. Let ${\cal C}$ be the collection of such $\E$s. Let $\L'$ be the set of linear orders on $\A$ such that $a_0$ is the worst alternative. There are $|\A_A|!$ distinct such linear orders. 

It can be shown that for any $\E, \E' \in {\cal C}$ and $\succ, \succ' \in \L'$ such that $(\succ, \E)\neq (\succ', \E')$, we have $\rho^{\succ}_{\E}\neq \rho^{\succ'}_{\E'}$. Thus, the number of distinct vertices is at least
\[
|\A_A|! \times 2^{(2^{|\A_A|} - {|\A_A| \choose 2} -1)}.
\]
Finally we show $\rho^{\succ}_{\E}\neq \rho^{\succ'}_{\E'}$ if $(\succ, \E)\neq (\succ', \E')$. To see this suppose by way of contradiction that $\rho^{\succ}_{\E}=\rho^{\succ'}_{\E'}$. Since $\succ, \succ' \in \L'$, the worst elements of $\succ$ and $\succ'$ are $a_0$. Moreover, both $\E$ and $\E'$ contain the sets of the form $\{a,b,a_0\}$ for all $\{a,b\} \subseteq \A_A$. Thus, $\rho^{\succ}_{\E}=\rho^{\succ'}_{\E'}$ implies that $\succ$ and $\succ'$ coincide binary comparison of any $a,b \in \A_A$. It follows that $\succ=\succ'$, which in turn implies $\E=\E'$ given $\rho^{\succ}_{\E}=\rho^{\succ'}_{\E'}$ and the fact that $\E$ contains all sets containing $a_0$.\footnote{There are many redundant vertices. For example, for any $\E$, fix two linear orders $\succ$ and $\succ'$ such that $a_1 \succ \cdots \succ a_{|\A_A|} \succ a_0$ and $a_1 \succ' \cdots \succ' a_{|\A_A|-1} \succ' a_0 \succ'   a_{|\A_A|}$. Then we have $\rho^{\succ}_{\E}=\rho^{\succ'}_{\E}$.}

\subsection{Proof of Remark \ref{rem:non-convex}}\label{sec:non-convexproof}

Fix $n$ such that $2 \le n <|\A_A|+1$. In the proof of Remark \ref{rem:ru_rational}, we only assume that $X(a_0)$ contains at least two elements, one ranked above all other alternatives and ranked below all other alternatives. This means $\rho^{\succ}_{\E} \in RU(2)$ for all $\succ \in \L(\A)$ and $\E \subseteq \D$. Thus we have $\rho^{\succ}_{\E} \in RU(n)$ for all $\succ \in \L$ and $\E \subseteq \D$, which implies statement (i).

We now show statement (ii), the non-convexity of $RU(n)$. Suppose by way of contradiction that $RU(n)$ is convex. Then, statement (i) implies that $RU(n)=\co \{\rho^{\succ}_{\E}| \succ \in \L, \E \subseteq \D\}$. 
However, by Corollary \ref{coro:vertexhalfspace}, the set of stochastic choice functions that are characterized by Limited Monotonicity and Partial RU-rationality is $\co\{\rho^{\succ}_{\E}| \succ \in \L, \E \subseteq \D\}$, which implies that $RU(n)$ is characterized by Limited Monotonicity and Partial RU-rationality. This contradicts with Theorem \ref{thm:nests}.\footnote{The constructive example of non-convexity of $RU(2)$ is as follows: Let $\A_A = \{x,y\}$ and $X(a_0) = \{z,w\}$. Let $x \ \succ_1 \  y \ \ \succ_1 \ a_0$ and $y \ \succ_2\  x  \ \succ_2 \ a_0$. Also, suppose that 
\begin{align*}
&\{x,a_0\} \notin \E_1,  \  \{y,a_0\} \in \E_1, \ \{x,y,a_0\} \in \E_1, \ \ \  
&\{x,a_0\} \notin \E_2, \ \{y,a_0\} \in \E_2,  \ \{x,y,a_0\} \notin \E_2.
\end{align*}
As stated in the body of the paper, $\rho^{\succ_1}_{\E_1}$ and $\rho^{\succ_2}_{\E_2}$ belong to $RU(2)$.
 By the same argument in the proof of Theorem \ref{thm:nests}, it can be shown that  $\rho = \alpha \rho^{\succ_1}_{\E_1}+ (1-\alpha)\rho^{\succ_2}_{\E_2} \not \in RU(2)$ for all $\alpha \in (0,1)$.}

\subsection{Proof of Lemma \ref{lem:non-overlapping} }\label{sec:proof-non-overlapping}

 Statement (a) follows from the fact that if $\muX$ is non-overlapping then for all $(S_a)_{a \in A}\in \prod_{a \in A} (2^{X(a) }\setminus \{\emptyset\})$ and $(S'_a)_{a \in 
A}\in \prod_{a \in A}  (2^{X(a) }\setminus \{\emptyset\})$ we have that $\rhoX(\cup_{a \in A} S_a, S_b)  = \rhoX(\cup_{a \in A} S'_a,S'_b )$ for all for all $b \in A$ where $\rhoX$ is the underlying stochastic choice function on $\X$ rationalized by $\mu \in \Delta (\L(\A_A \cup X(a_0)))$.  
Fix $(S^*_a)_{a \in A}$. Then for all $b \in A$, we have 
\begin{align*}
&\sum_{(S_a)_{a \in A}\in \prod_{a \in A} (2^{X(a) }\setminus \{\emptyset\})} \lambda_{A}((S_a)_{a \in A})\rhoX(\cup_{a \in A} S_a,S_b)\\ 
&= \rhoX(\cup_{a \in A} S^*_a,S^*_b)\sum_{(S_a)_{a \in A}\in \prod_{a \in A} (2^{X(a) }\setminus \{\emptyset\})} \lambda_{A}((S_a)_{a \in A})\\
&= \rhoX(\cup_{a \in A} S^*_a,S^*_b)\\
&=\muX(\succ| \exists x \in S^*_b, \ \forall y \in \cup_{a \in A \setminus b} S_a, \ x \succ y).
\end{align*}
Thus every $\lambda$ produces the same choice function. Thus, by the definition of rationalization (Definition \ref{def:rationalization}), we established statement (a).

To prove statement (b), assume there is no non-overlapping $\muX$ that RU rationalizes $\rhoa$. By way of contradiction assume that there exist $\lambda'$  and non-overlapping $\muX'$ such that $(\muX', \lambda')$ that rationalizes $\rhoa$.  Given this, statement (a) implies a contradiction.

\subsection{Proof of Lemma \ref{lem:ind}}\label{sec:proof_ind}

Fix $A \subseteq \A$ and $a \in A$. Then, 
\begin{align*}
\rhoa(A, a)
&= \sum_{\left(S_b\right)_{b \in A} \in \prod_{b \in A} \left(2^{X(b)} \setminus \{\emptyset\}\right)} 
\lambda_A\left((S_b)_{b \in A}\right) 
\rhoX(\cup_{b \in A} S_b, S_a) \\
&= \sum_{\left(S_b\right)_{b \in A} \in \prod_{b \in A} \left(2^{X(b)} \setminus \{\emptyset\}\right)} 
\lambda \Big( (S_b)_{b \in A} \times \prod_{b \in \A \setminus A} (2^{X(b)}\setminus \{\emptyset\})\Big)
\rhoX(\cup_{b \in A} S_b, S_a) \\
&= \sum_{\left(S_b\right)_{b \in A} \in \prod_{b \in A} \left(2^{X(b)} \setminus \{\emptyset\}\right)} 
\sum_{\left(S_c\right)_{c \in A^c} \in \prod_{c \in A^c} \left(2^{X(c)}\setminus \{\emptyset\}\right)} \lambda((S_b)_{b \in A}, (S_c)_{c \in A^c})
\rhoX(\cup_{b \in A} S_b, S_a) \\
&= \sum_{\left(S_b\right)_{b \in \A} \in \prod_{b \in \A} (2^{X(b)}\setminus \{\emptyset\})} 
\lambda\left((S_b)_{b \in \A}\right)  
\rhoX(\cup_{b \in A} S_b, S_a).
\end{align*}

\section{Generalization of Theorem \ref{thm:vertex}, \ref{thm:nests}, and \ref{thm:errorbound}  to multiple non-atomic aggregates}\label{sec:multiplegeneralizations}

\subsection{Generalization of Theorem \ref{thm:vertex}: Vertex characterization}
We first define the vertices of the RU polytope for multiple non-atomic aggregates.  For each sequence  of disjoint collections $\{\E_a\}_{a \in \A_N}$ in $\D$, where $\E_a \subseteq \D$ for each $a$ and $\E_a \cap \E_b = \emptyset$ for all $a \neq b$, define $\E = \D \setminus (\cup_{a \in \A_N} \E_a)$. The interpretation of $\{\E\} \cup \{\E_a\}_{a \in \A_N}$ is that when $D \in \E$ the agent follows the linear order and when $D \in \E_a$,  the agent deviates from the linear order and chooses $a$.\footnote{In the single non-atomic aggregate case, we defined $\E$ as the set of menus where the agent follows the linear order. Thus the sequence for this case is defined by $\E_{a_0} =  \D \setminus \E$.} Thus we define for each $\succ \in \L(\A)$,
\[
c^{\succ}_{\{\E_a\}_{a \in \A_N}}(D)= \begin{cases}
\max_D \succ \ &\text{ if } D \in \E, \\
a\ &\text { if $D \in \E_a$}
\end{cases}
\]
and define $\rho^\succ_{ \{\E_a\}_{a \in \A_N}}(D,x)$ to be the indicator function of $c^\succ_{\{\E_a\}_{a \in \A_N}}(D) =x$.  The  theorem shows that these are the vertices of the RU polytope for multiple aggregated alternatives.

\begin{theorem} \label{thm:vertexmultiple}
A stochastic choice function $\rho$ is RU-rational if and only if   
\[
\rho \in \co\Big\{\rho^{\succ}_{ \{\E_a\}_{a \in \A_N}} \big| \succ \in \L(\A) \text{ and }  \{\E_a\}_{a \in \A_N}, \E_a \subseteq \D, \forall a \neq b \ \E_a \cap \E_b = \emptyset \Big\}. 
\]
\end{theorem}

\begin{proof}  We first show that every RU-rational choice function is a  convex combination of $\rho^\succ_{ \{\E_a\}_{a \in \A_N}}$.

Remember that for all $\succ \in \L(\X)$ and $(S_{(A,a)})_{a \in A \in \D}  \in \prod_{(A,a): a\in A \in \D}(2^{X(a)} \setminus \{\emptyset\})$, define $\rho^{\succ}_{(S_{(A,a)})_{a \in A\in \D}}$ as follows: $\rho^{\succ}_{(S_{(A,a)})_{a\in A\in \D}}(B,b) := 1(\exists x \in S_{(B,b)} \forall y \in \cup_{a \in B \setminus b} S_{(B,a)}, [ x \succ y])$ for all $(B,b)$ such that $b \in B \in \D$; thus, $\rho^{\succ}_{(S_{(A,a)})_{a\in A\in \D}}$ is a choice function where on menu $A$, the aggregate $a$ means $S_{(A,a)}$ for all pairs $(A,a)$ such that $a \in A \in \D$.

\noindent \textbf{Step 1: }
For all $\succ \in \L(\X)$ and $(S_{(A,a)})_{a \in A \in \D}  \in \prod_{(A,a): a\in A \in \D}(2^{X(a)} \setminus \{\emptyset\})$, there exists $\succ_{\A} \in \L(A)$ and $\{\E_a\}_{a \in \A_N}$ such that $\rho^{\succ}_{(S_{(A,a)})_{a \in A\in \D}}=\rho^{\succ_\A}_{ \{\E_a\}_{a \in \A_N}}$. 
\begin{proof}

Fix each $\succ \in \L(\X)$, each collection $(S_{(A,a)})_{a \in A \in \D}$, and each $a \in A$,  
define $\E_a=\big\{A \in \D\ \big|\rho^\succ_{(S_{(A,a)_{a \in A\in \D}})}(A,a) =1\big\}$. That is, $\E_a$ is the set of menus where $a$ is chosen with probability 1. 
By definition, $\E_a \cap \E_b= \emptyset$, thus $\{\E_a\}_{a \in \A_N}$ is a disjoint sequence of collections in $\D$. 


Define $\succ_\A \in \L(\A_A)$ as follows: for all $a, b \in \A_N$, $a \succ_{\A} b$ if and only if $x \succ b$, where $X(a)=\{x\}$ and $X(b)=\{y\}$; moreover $\succ_\A$ ranks elements of $\A_A$ the same as $\succ$ and puts all elements of $\A_N$ at the bottom in any order.

For any $A \in \D$, if $A \in \E_a$ for some $a$, then $\rho^{\succ}_{(S_{(A,a)})_{a\in A\in \D}}(A,a)=1=\rho^{\succ_\A}_{\E}(A,a)$, thus they agree on $A$. Otherwise, since $\rho^{\succ}_{(S_{(A,a)})_{a\in A\in \D}}$ does not choose any non-atomic aggregate on $A$, it chooses the highest ranked non-atomic aggregate according to $\succ$ which coincides with the largest non-atomic aggregate according to $\succ_{\A}$ by definition and they agree on $A$. Thus the two choice functions agree on every menu and therefore they are equal.
\end{proof}


Exactly by the same logic as in the proof of Theorem \ref{thm:vertex} in Section \ref{sec:proof_vertex}, we obtain the following two steps:

\noindent \textbf{Step 2:} Suppose that $\rhoa$ is RU-rationalized by $(\delta_{\succ},\lambda)$. Then $\rhoa$ is the convex combination of functions of the form $\rho^{\succ}_{(S_{(A,a)})_{a \in A\in \D}}$.

\noindent \textbf{Step 3:} Every RU-rational $\rho$ is a convex combination of choice functions rationalized by $(\delta_\succ, \lambda)$.


Step 3 says that every RU rational $\rho$ is in the convex hull of choice functions rationalized by $(\delta_\succ, \lambda)$. Step 2 says that every choice function rationalized by $(\delta_\succ, \lambda)$ is itself a convex combination of choice functions of the form $\rho^{\succ}_{(S_{(A,a)})_{a \in A\in \D}}$. Finally, Step 1 says that every $\rho^{\succ}_{(S_{(A,a)})_{a \in A\in \D}}$ can be written as $\rho^{\succ}_{ \{\E_a\}_{a \in \A_N}} $ for the appropriate sequence $\{\E_a\}_{a \in \A_N}$. Thus combining the steps yields $
\rho \in \{\rho^{\succ}_{ \{\E_a\}_{a \in \A_N}} \big| \succ \in \L(\A) \text{ and } $ $ \{\E_a\}_{a \in \A_N}, \E_a \subseteq \D, \forall a \neq b \ \E_a \cap \E_b = \emptyset \}$.

Now we show the other inclusion. By the definition, for all $\succ$ and $ \{\E_a\}_{a \in \A_N}$, $\rho^\succ_{ \{\E_a\}_{a \in \A_N}}$  satisfies Limited Monotonicity and Partial-RU. Thus by Theorem \ref{thm:1multiple}, every $\rho^\succ_{ \{\E_a\}_{a \in \A_N}}$ is RU-rational. The RU-rational polytope is convex by Theorem \ref{thm:1multiple}. Thus, the convex hull of $\{\rho^{\succ}_{ \{\E_a\}_{a \in \A_N}} \mid $ $\succ \in \L(\A) \text{ and }\{\E_a\}_{a \in \A_N}, \E_a \subseteq \D, \forall a \neq b \ \E_a \cap \E_b = \emptyset\}$ is a subset of the RU-rational set. This completes the proof.
\end{proof}

\subsection{Generalization of  Theorem \ref{thm:nests}: Exogenous aggregation correspondences}\label{sec:exogenous_multiple}

Let $RU((n_a)_{a \in \A_N})$ be the set of RU-rationalizable choice functions when for each $a \in \A_N$, we have $|X(a)| = n_a$. That is, for each vector $(n_a)_{a \in \A_N} \in \mathbf{N}^{\A_N}$, we define $RU((n_a)_{a \in \A_N})$ to be the set of RU-rationalizable choice functions in the case where each aggregate $a\in \A_N$ contains exactly $n_a$ underlying alternatives. For simplicity, we write $(n)$ to mean the vector in $\mathbf{N}^{\A_N}$ with all entries equal to $n$. For example, $(|\A_A|+2)$ means the vector $(n_a)_{a \in \A_N}$ where $n_a = |\A_A|+2$ for all $a$. We will also write $\ge$ to be the standard entry-wise comparison.

In the following we state and prove our characterization of RU-rationalizable choice functions with multiple non-atomic aggregates and exogenous aggregation correspondences. The result is analogous to Theorem \ref{thm:nests}.

\begin{theorem}\label{thm:multiplenests}
Suppose $\D = 2^{\A}\setminus \{\emptyset\}$. Then,
\begin{itemize}
    \item[(i)] if $(n_a)_{a \in \A_N} \ge (n'_a)_{a \in \A_N}$ then $RU((n_a)_{a \in \A_N}) \supseteq RU((n'_a)_{a \in \A_N})$. Furthermore, if there exists $a \in \A_N$ such that $n'_a < n_a < |\A_A|+2$ then the inequality is strict,
    \item[(ii)] if $(n_a)_{a \in \A_N} \ge (|\A_A|+2)$ then $RU((n_a)_{a \in \A_N})$ is equal to the set of choice functions that satisfy Limited Monotonicity and Partial-RU.
\end{itemize}
\end{theorem}
\begin{proof}
Statement (i): The weak inclusions are immediate. Thus, we only need to   show that $RU((n_a)_{a \in \A_N}) \neq RU((n'_a)_{a \in \A_N})$ if there exists $a \in \A_N$ such that $n'_a < n_a < |\A_A|+2$. We do this by constructing a choice function that is in $RU((n'_a)_{a \in \A_N})$ but not in $RU((n_a)_{a \in \A_N}) $, based on the construction in the proof of Theorem \ref{thm:nests}.

Fix some $a \in \A_N$ such that $n'_a < n_a < |\A_A|+2$. We will label $a = a_0$. Then consider the restricted  domain $\D' = 2^{\A_A\cup \{a_0\}} \setminus \{\emptyset\}$. On this domain, the multiple aggregated alternative case reduces to the outside option setup. Consider the same example of a choice function $\rho' \in RU(n_{a_0})$ but $\rho' \notin RU(n_{a_0}-1)$ in the proof of Theorem \ref{thm:nests} on the restricted domain $\D'$. We extend $\rho'$ to $\rho$ defined on $\D$ in the following way. Let $\rho=\rho'$ on $\D'$. Label the aggregated alternatives by $\A_N = \{a_0,...,a_{|\A_N|}\}$.  Whenever $D \not \in \D'$ (i.e., $D$ contains some element of $\{a_1,...,a_{|\A_N|}\}$), $\rho$ puts probability $1$ on the element of $\A_N$ with the largest index. Then, we have $\rho \in RU((n_a)_{a \in \A_N})$ by definition.  Since $\rho' \notin RU(n_{a_0}-1)$, we have $\rho' \notin RU(n'_{a_0})$ and thus $\rho \not \in RU((n_a')_{a \in \A_N})$. This completes the proof of Statement (i).\\
Statement (ii) follows immediately by the proof of Theorem \ref{thm:1multiple}.
\end{proof}

The theorem says that as the number of underlying alternatives for each aggregate increases, the set of RU-rational choice function expands. Furthermore, when every aggregate has at least $|\A_A|+2$ underlying alternatives, then the RU-rational set is fully characterized by Limited Monotonicity and Partial-RU.\footnote{Notice that with multiple aggregated alternatives $|\A_A|+2 \le |\A|$. Thus we need only $|\A|$ underlying alternatives for each aggregate. It follows that as long as $|\X| \ge |\A|^2$, Theorem \ref{thm:1multiple} holds.}
We can also write as a direct implication of Theorem \ref{thm:multiplenests},
 \begin{align*}
 \bigcup_{(n_a)_{a \in \A_N} \in \mathbf{N}^{\A_N}}& RU((n_a)_{a \in \A_N}) = RU((|\A_A|+2))\\
 &=\{\rho| \ \rho \text{ satisfies Limited Monotonicity \& Partial RU-rationality}\},
 \end{align*}
 which means that we only need to consider that case where $X(a) = |\A_A|+2$ for all $a \in \A_N$ to get the whole RU set.

\subsection{Generalization of Theorem \ref{thm:errorbound}: Approximation results}

We first state and prove the analogous result to Remark  \ref{rem:non-convex}.

\begin{lemma} For any vector of integers $(n_a)_{a \in \A_N}$,
\begin{itemize}
\item[(i)] $\rho^\succ_{\{\E_a\}_{a \in \A_N}} \in RU ((n_a)_{a\in \A_N})$ for all $\succ \in \L(\A)$ and disjoint sequence $ \{\E_a\}_{a \in \A_N}$,
\item[(ii)] If $(n_a)_{a \in \A_N} \not \ge (|\A_A|+1)$ then $RU((n_a)_{a \in \A_N})$ is not convex.
\end{itemize}
\end{lemma}
\begin{proof} We first prove (i).  Given Theorem \ref{thm:multiplenests}, it suffices to show $\rho^\succ_{\{\E_a\}_{a \in \A_N}}\in RU((2))$ for all $\succ \in \L(\A)$ and disjoint sequences $ \{\E_a\}_{a \in \A_N}$. For each $a \in \A_N$ fix $\overline{x}_a \in X(a)$ and $\underline{x}_a \in X(a)$.

We will construct a linear order $\succ_{\X}$ on $\X$ and a composition correspondence $\lambda$ that  rationalizes $\rho^{\succ}_{\{\E_a\}_{a \in \A_N}}$, in other words we will show that $\rho^\succ_{ \{\E_a\}_{a \in \A_N}}=\rho_{(\delta_{\succ_{\X}}, \lambda)}$. We first obtain $\succ_{\X}$ by augmenting $\succ$ as follows. Define $\succ_\X$ by (i) for all $a \in \A_N$ and $b \in \A_A$, $\underline{x}_a \succ_\X b$ if and only if $a \succ b$; (ii)  $\overline{x}_a \succ b$ for all $b \in \A_A$; (iii) $\underline{x}_a \succ_{\X} \underline{x}_b$ if and only if $a \succ b$ for all $a, b \in \A_N$; (iv) $\overline{x}_a \succ_{\X} \underline{x}_b$ for all $a, b \in \A_N$. The other comparisons may be arbitrary. 

Now we define the composition correspondence $\lambda$ as follows: for each $D \in \D$, if $D \notin  \cup_{a \in \A_N} \E_a$, then $\lambda_D$ puts probability $1$ on the event that each  $a \in \A_N$ is composed of just $\underline{x}_a$. 
Otherwise, $D \in \E_b$ for some $b \in \A_N$. Then $\lambda_D$ puts probability $1$ on the event that each  $a \in \A_N$ is composed of just $\underline{x}_a$, except for $b$ which is composed of $\overline{x}_{b}$.

Finally, we show $\rho^\succ_{ \{\E_a\}_{a \in \A_N}}=\rho_{(\delta_{\succ_{\X}}, \lambda)}$. Fix $D \in \D$. Consider the case $D \notin  \cup_{a \in \A_N} \E_a$. Then  all $a \in \A_N$ composed of $\underline{x}_a$. Thus, the definition (i) and (iii) of $\succ_{\X}$ implies that $\rho^\succ_{ \{\E_a\}_{a \in \A_N}}=\rho_{(\delta_{\succ_{\X}}, \lambda)}$. Now consider case in which $D \in \E_b$ for some $b \in \A_N$. Then  all $a \in \A_N\setminus b$ composed of $\underline{x}_a$ and $b$ composed of $\overline{x}_b$. Thus, the definition (ii) and (iv) of $\succ_{\X}$ implies that $\rho^\succ_{ \{\E_a\}_{a \in \A_N}}=\rho_{(\delta_{\succ_{\X}}, \lambda)}$.

The proof of statement (ii) is exactly the same as in the proof of Remark \ref{rem:non-convex} replacing corresponding theorems to the generalized theorems with multiple non-atomic aggregates. 
\end{proof}

Next we prove the analogue of Lemma \ref{lem:subpolys} as follows:

\begin{lemma}\label{lem:multiple_subpolys} For any $n \ge 2$, the convex hull of $\{\rho^{\succ_i}_{\{\E^i_a\}_{a \in \A_N}}\}^{n-1}_{i=1}$ is contained by $RU((n))$ for any sequence $(\succ_i,\{\E^i_a\}_{a \in \A_N})_{i=1}^{n-1}$.
\end{lemma}

\begin{proof} Fix $(\succ_i,\{\E^i_a\}_{a \in \A_N})_{i=1}^{n-1}$. Suppose $\rho = \sum \alpha_i\rho^{\succ_i}_{\{\E^i_a\}_{a \in \A_N}}$ for some $\alpha_i\ge 0$ such that $\sum \alpha_i = 1$. We proceed by constructing a pair $(\mu_\X,\lambda)$ that rationalizes $\rho$.

For each $a \in \A_N$, label the elements of $X(a)$ by $ \{\underline{x}_a,\overline{x}_a^1,\overline{x}_a^2,...,\overline{x}_a^{n-1}\}$. (This is possible because $|X(a)|=n$.)  Then for each $\rho^{\succ_i}_{\{\E^i_a\}_{a \in \A_N}}$, define $\succ_\X^i \in \L(\X)$ in the following way: (i) $\overline{x}^i_a \succ^i_\X x$ for all $x \in \X \setminus \{\overline{x}^i_a\}_{a\in \A_N}$, (ii) $b \succ^i_\X \overline{x}^j_a$ for all $j \neq i$, $b \in \A_A$, (iii) $\underline{x}_a \succ^i_\X b$ if and only if $a \succ^i b$ for all $b \in \A_A$, (iv) $\underline{x}_a \succ^i_\X \underline{x}_b$ for all $b \in \A_N$, and (v) $a \succ^i_{\X} b$ if and only if $a\succ^i b$ for all $a,b \in \A_A$. All other comparisons are arbitrary.

Now we define $\mu_\X \in \Delta (\L(\X))$ by $\mu_\X(\succ^i_\X)= \alpha_i$ for all $i=1,\dots, n-1$. We define $\lambda$ as follows:  if $A \not\in \bigcup_{a \in \A_N}\E_a$, then let $\lambda_A(\{\underline{x}_a| a \in \A_N\}) = 1$; if $A \in \E_a$ for some $a \in \A_N$ then let $\lambda_A(\{\underline{x}_b| b \in \A_N \setminus {a}\} \cup \{\overline{x}_a\}) = 1$.

Filly, we show that  $\rho$ is rationalized by the pair $(\muX , \lambda)$. That is, $\rho_{(\muX,\lambda)}=\rho$. For all $(A,a)$ such that $a \in A \in \D$, $\rho_{(\muX,\lambda)}(A,a)$ equals
\begin{align*}
&\sum_{(S_b)_{b \in A}\in \underset{b\in A}{\prod} (2^{X(b)}\setminus \{\emptyset\})} \lambda_{A}\big((S_b)_{b \in A}\big)\muX\big(\succ\in \L (\X) | \exists x \in S_a, \ \forall y \in \underset{b \in A \setminus a}{\cup} S_b, \ x \succ y\big)\\
&= \sum_{\succ \in \L(\X)} \mu(\succ) \sum_{(S_b)_{b \in A}\in \underset{b\in A}{\prod} (2^{X(b)}\setminus \{\emptyset\})} \lambda_{A}\big((S_b)_{b \in A}\big)1( \exists x \in S_a, \ \forall y \in \underset{b \in A \setminus a}{\cup} S_b, \ x \succ y\big)\\
&= \sum_{i=1}^{n-1} \alpha_i \sum_{(S_b)_{b \in A}\in \underset{b\in A}{\prod} (2^{X(b)}\setminus \{\emptyset\})} \lambda_{A}\big((S_b)_{b \in A}\big)1( \exists x \in S_a, \ \forall y \in \underset{b \in A \setminus a}{\cup} S_b, \ x \succ^i_\X y\big)\\
&= \sum_{i=1}^{n-1} \alpha_i \rho^{\succ_i}_{\{\E^i_a\}_{a \in \A_N}}(A,a) \equiv  \rho(A,a).
\end{align*}
\end{proof}

We are now ready to prove the analogue of Theorem \ref{thm:errorbound} for multiple non-atomic aggregates.
\begin{theorem}\label{thm:errorboundmultiple}
Fix any integer $n$ such that $2 < n <|\A_A|+2$. For all $\rho$ satisfying Limited Monotonicity and Partial-RU, there exists $\rho' \in RU((n))$ such that,
\[
\frac{\|\rho-\rho'\|^2}{|\D|} \le \frac{1}{n-1}.
\]

\end{theorem}

\begin{proof} 
Given Lemma \ref{lem:multiple_subpolys}, the proof  proceeds exactly as in the proof of Theorem \ref{thm:errorbound} using Approximate Carath\'{e}odory Theorem.
\end{proof}

\section{Additional theoretical results}

\subsection{Alternate proof of Theorem \ref{thm:1}}\label{sec:altproofthm1}
Previously we proved  Theorem \ref{thm:1} as a special case of  Theorem \ref{thm:1multiple}. In this section we prove Theorem \ref{thm:1} separately using a construction that allows for a slightly weaker assumption on the cardinality of $X(a_0)$.

The necessity is trivial. We prove the sufficiency. Assume the two conditions and fix an aggregation correspondence $X$ such that $|X(a_0)| = |\A_A| + 1$.

By Partial-RU, there exists $\tilde{\mu} \in \Delta (\L(\A_A))$ that rationalizes $\rho$ on $2^{\A_A} \setminus \{\emptyset\}$.  For each $y \in \A_A$, choose an element of  $X(a_0)$ denoted by $x_0(y)$. We will also fix a unobservable alternative $\underline{x}$ that plays a special role in the following. Note that such alternatives exist because  $|X(a_0)| = |\A_A| + 1$.

 For any ranking $\tilde \succ$ on $\A_A$ that belongs to the support of $\tilde{\mu}$,  we extend $\tilde \succ$ and define a ranking $\succ$ on $\X$ as follows.  First, for each  $y \in \A_A$, place $x_0(y) \in X(a_0)$ just above $y$ and $\underline{x}$ at the bottom for all $\tilde \succ\in supp (\tilde \mu)$. For example, if the ranking $y_1\  \tilde \succ\   y_2\   \tilde \succ\  
 y_3$ is in the support then we extend $\tilde \succ$  to $\succ$ on $\X$ as follows
 \[
 x_0(y_1) \succ y_1 \succ x_0(y_2) \succ y_2 \succ x_0(y_3) \succ y_3 \succ\underline{x}.
 \]
Finally, we define $\muX \in \Delta (\L(\X))$ as follows: for each extended ranking $\succ$ defined above, we let $\muX(\succ)=\tilde{\mu}(\tilde\succ)$. 

Let $\rhoX$ be a random utility model that is rationalized by $\muX$. Now fix $D\subseteq \A_A$. In the following, we will construct $\lambda \in \prod_{b \in D \cup a_0} (2^{X(b)} \setminus \{\emptyset\})$ to establish the following equality for all $b \in D \cup a_0$: 
\[
\sum_{(S_b)_{b \in D \cup a_0} \subseteq \prod_{b \in D \cup a_0} (2^{X(b)} \setminus \{\emptyset\})}\lambda((S_b)_{b \in D \cup a_0} ) \rhoX(\cup_{b \in D \cup a_0} S_b, S_b) =\rhoa(D \cup a_0,b).
\]
Since $D \subseteq \tilde X$, for all $b \in D$, there exists $y_b$ such that $X(b)=\{y_b\}$.

Thus, for any $(S_b)_{b \in D \cup a_0} \subseteq \prod_{b \in D \cup a_0} (2^{X(b)} \setminus \{\emptyset\})$, $S_b$ is determined as $\{y_b\}$ for all $b \in D$. Thus, for simplicity, we will focus on the marginal distribution of $\lambda$ on $2^{X(a_0)} \setminus \{\emptyset\}$  and obtain the following equality: 
\begin{equation}
\sum_{S \subseteq X(a_0)}\lambda(S) \rhoX(\{y_b| b \in D\} \cup S, y_b) =\rhoa(D \cup a_0,b)\quad \text{ for all }b \in D
\end{equation}
and
\begin{equation}\label{eq:pf2}
\sum_{S \subseteq X(a_0)}\lambda(S) \rhoX(\{y_b| b \in D\} \cup S, S) =\rhoa(D \cup a_0,a_0).
\end{equation}
Note that (\ref{eq:pf2}) is implied by (\ref{eq:pf1}) by the fact that the probability must  sum to one. The  purpose of our proof is, thus, to show the existence of $\lambda$ on $2^{X(a_0)} \setminus \{\emptyset\}$ that satisfies the following equality for all $y \in D \subseteq \A_A$,\footnote{Remember that this $\lambda$ is a composition distribution for the choice set $D \cup\{a_0\}$. For another choice set $E \cup\{a_0\}$, one can construct a desirable composition distribution in the same way.}
\begin{equation}\label{eq:pf1} 
\sum_{S \subseteq X(a_0)}\lambda(S) \rhoX(D \cup S, y) =\rhoa(D \cup a_0, y).
\end{equation}

Now we will define $\lambda$ recursively.  We first label elements of $D$ by $y_i$ so that $\frac{\rhoa(D\cup {a}_0,y_0)}{\rhoa(D,y_0)}$ is maximum among $\frac{\rhoa(D\cup a_0,y_i)}{\rhoa(D ,y_i)}$. Define
\[
\lambda_0(\underline{x}) = \frac{\rhoa(D \cup a_0,y_0)}{\rhoa(D ,y_0)}, \quad \lambda_0(X(a_0)) = 1-\lambda_0(\underline{x})
\]
and all other values are zero. Then $\lambda_0$ is a probability measure because $\lambda_0(\underline{x})\le 1$ by Limited Monotonicity.  Moreover, the chance that $y_0$ is chosen is equal to $\rhoa(D\cup a_0, y_0)$ as desired:
\begin{align*}
\sum_{S \subseteq X(a_0)}\lambda_0(S) \rhoX(D \cup S, y_0)
&=\lambda_0(\underline{x})   \rhoX(D\cup \{\underline{x}\},y_0) + \lambda_0(X(a_0))   \rhoX(D\cup X(a_0),y_0) \\
&= \lambda_0(\underline{x})   \rhoX(D\cup \{\underline{x}\},y_0) \quad (\because x(y_0) \succ y_0 \text{ for all } \succ\in supp(\muX))\\
&= \lambda_0(\underline{x})   \rhoX(D ,y_0)  \quad\quad\quad\quad (\because y_0 \succ\underline{x} \text{ for all } \succ\in supp(\muX))\\
&= \lambda_0(\underline{x})   \rhoa(D ,y_0) \\
&=\rhoa(D\cup a_0,y_0),
\end{align*}
where the fourth equality holds because $\rhoX(D ,y_0)=\rhoa(D ,y_0)$ because all elements in $D$ are atomic aggregates.

Given $\lambda_0$, we will  define $\lambda_1$ that satisfies (\ref{eq:pf1}) for $y_0$ and $y_1$. First remember that since $\frac{\rhoa(D \cup {a}_0,y_0)}{\rhoa(D ,y_0)}$ is maximum among $\frac{\rhoa(D \cup {a}_0,y_i)}{\rhoa(D ,y_i)}$, we have
\begin{eqnarray*}   
\lambda_0(\underline{x})\equiv\dfrac{\rhoa(D\cup {a}_0,y_0)}{\rhoa(D ,y_0)} \ge \dfrac{\rhoa(D\cup {a}_0,y_i)}{\rhoa(D ,y_i)}=\dfrac{\rhoa(D\cup {a}_0,y_i)}{\rhoX(D ,y_i)} =\dfrac{\rhoa(D\cup {a}_0,y_i)}{\rhoX(D\cup \{\underline{x}\},y_i)},
\end{eqnarray*}
where the second to the last equality holds because $\rhoX(D ,y_i)=\rhoa(D ,y_i)$; and the last equality holds because $\underline{x}$ is the worst alternative for any $\succ\in supp (\muX)$. Thus $\lambda_0(\underline{x})   \rhoX(D\cup \{\underline{x}\}, y_i) \ge \rhoa(D\cup a_0, y_i)$. By using this inequality, we have 
\begin{align*}
\sum_{S \subseteq X(a_0)}\lambda_0(S) \rhoX(D \cup S, y_i)
&=\lambda_0(\underline{x})   \rhoX(D\cup \{\underline{x}\}, y_i) + \lambda_0(X(a_0))   \rhoX(D\cup X(a_0), y_i) \\
&= \lambda_0(\underline{x})   \rhoX(D\cup \{\underline{x}\}, y_i) \\
&\ge \rhoa(D\cup {a}_0, y_i).
\end{align*}
This inequality suggests that we should decrease the value of $\lambda_0(\underline{x})$ to have the desired equality for $i>0$. Given the observation, we define $\lambda_1$ given $\lambda_0$ as follows. The  idea is to move probability from $\lambda_0(\underline{x})$ to $\lambda_0(\underline{x},x_0(y_1))$ so that the chance that $y_1$ is chosen decreases to $\rhoa(D \cup a_0,y_1)$ as desired.  More precisely, we define $\lambda_1$ as follows:
\[
\lambda_1(\underline{x})= \frac{\rhoa(D\cup {a}_0,y_1)}{\rhoa(D,y_1)}, \quad \lambda_1(\underline{x}, x_0(y_1))= \lambda_0(\underline{x})-\lambda_1(\underline{x}), \quad \lambda_1(X(a_0))=\lambda_0(X(a_0)). 
\]
To see that  $\lambda_1$ is a probability measure, notice that $\lambda_1(\underline{x}, x_0(y_1))\ge 0$ because $\frac{\rhoa(D \cup {a}_0, y_0)}{\rhoa(D ,y_0)} \ge \frac{\rhoa(D \cup {a}_0,y_1)}{\rhoa(D ,y_1)}$. Moreover, we have the desired equalities for $y_1$ and $y_0$ as follows:
\begin{align*}
&\sum_{S \subseteq X(a_0)}\lambda_1(S) \rhoX(D \cup S, y_1)\\
&=\lambda_1(\underline{x})   \rhoX(D\cup \{\underline{x}\}, y_1)+ \lambda_1(\underline{x},x_0(y_1)
) \rhoX(D\cup \{\underline{x},x_0(y_1)\}, y_1) + \lambda_1(X(a_0))   \rhoX(D\cup X(a_0), y_1) \\ 
&= \lambda_1(\underline{x})   \rhoX(D\cup \{\underline{x}\}, y_1)\quad (\because x_0(y_1)\succ y_1 \text{ for all }\succ  \in supp(\mu)) \\
&= \lambda_1(\underline{x})   \rhoX(D , y_1) \quad (\because y_1 \succ\underline{x}\text{ for all }\succ \in supp(\mu)) \\
&= \lambda_1(\underline{x})   \rhoa(D , y_1) \\
&=\rhoa(D\cup a_0, y_1)
\end{align*}
and 
\begin{align}\label{eq:315}
&\sum_{S \subseteq X(a_0)}\lambda_1(S) \rhoX(D \cup S, y_0) \nonumber \\
&=\lambda_1(\underline{x})   \rhoX(D\cup \{\underline{x}\},y_0)+ \lambda_1(\underline{x},x_0(y_1)) \rhoX(D\cup \{\underline{x},x_0(y_1)\},y_0) + \lambda_1(X(a_0))   \rhoX(D\cup X(a_0),y_0) \nonumber \\
&= \big( \lambda_1(\underline{x})+\lambda_1(\underline{x},x_0(y_1)
) \big)   \rhoX(D\cup \{\underline{x}\},y_0)+ \lambda_1(X(a_0))   \rhoX(D\cup X(a_0),y_0)\\
&=\lambda_0(\underline{x})   \rhoX(D\cup \{\underline{x}\},y_0) + \lambda_0(X(a_0))   \rhoX(D\cup X(a_0),y_0) \nonumber\\
&=\lambda_0(\underline{x})   \rhoX(D\cup \{\underline{x}\},y_0) \nonumber \\
&= \rhoa(D\cup a_0,y_0), \quad(\because \text{Definition of } \lambda_0) \nonumber
\end{align}
where the second equality holds because  for all $\succ\in supp(\muX)$, $x_0(y_1)$ is the immediate predecessor of $y_1$ so adding $x_0(y_1)$ changes the only the probability of choice frequency of $y_0$, thus $\rhoX(D\cup \{\underline{x}\},y_0) = \rhoX(D\cup \{\underline{x},x_0(y_1) \},y_0)$. 

In general, at $n$th step, given $\lambda_{n-1}$ that satisfy (\ref{eq:pf1}) for all $y \in \{y_1,\dots, y_{n-1}\}$, we will define $\lambda_n$ that satisfy (\ref{eq:pf1}) for all $y \in \{y_1,\dots, y_{n-1},y_{n}\}$. First define
\[
c_n = \frac{\rhoa(D \cup a_0,y_n)}{\rhoa(D ,y_n) \lambda_0(\underline{x})}.
\]
 Note $0 \le c_n \le 1$ because $\frac{\rhoa(D \cup a_0,y_0)}{\rhoa(D ,y_0)}$ is maximum among $\frac{\rhoa(D \cup a_0,y_i)}{\rhoa(D ,y_i)}$. Let $\lambda_n(X(a_0)) = \lambda_0(X(a_0))$ and for each $S$ in the support of $\lambda_{n-1}$ except $X(a_0)$
 \[
 \lambda_n(S) = c_n\lambda_{n-1}(S), \quad \lambda_n(S \cup x_0(y_n)) = \lambda_{n-1}(S)-\lambda_{n}(S).
 \]
Note first that $\lambda_n$ is a probability measure given that $\lambda_{n-1}$ is a probability measure.   Note also this modification of $\lambda_n$ from $\lambda_{n-1}$ does not change the chance that $y_i$ is chosen for any $i  \le n-1$; this only changes the chance that $y_n$ is chosen because adding $x_0(y_n)$ changes only the choice frequency of $y_n$ as in (\ref{eq:315}). Moreover, we obtain the desired equalities for the choice frequency of $y_n$ as follows:
    \begin{align*}
&\sum_{S \subseteq X(a_0)}\lambda_n(S) \rhoX(D \cup S, y_n)\\    
&=\lambda_{n-1}(X(a_0))\rhoX(D\cup X(a_0),y_n) +\\
&\sum_{S \in supp(\lambda_{n-1})\setminus \{X(a_0)\}} c_n\lambda_{n-1}(S)\rhoX(D\cup S , y_n) + (1-c_n)\lambda_{n-1}(S)\rhoX(D\cup S \cup x_0(y_n), y_n) \\
&= \sum_{S \in supp(\lambda_{n-1})\setminus \{X(a_0)\}} c_n\lambda_{n-1}(S)\rhoX(D\cup S , y_n) \quad (\because x_0(y_n) \succ y_n \text{ for all }\succ\in supp (\muX))\\
&= \sum_{S \in supp(\lambda_{n-1})\setminus \{X(a_0)\}} c_n\lambda_{n-1}(S) \rhoX(D, y_n) \\
&= c_n\sum_{S \in supp(\lambda_{n-1})\setminus \{X(a_0)\}} \lambda_{n-1}(S) \rhoa(D, y_n) \quad (\because \rhoa(D, y_n)=\rhoX(D, y_n) )\\
&=c_n \lambda_0(\underline{x})\rhoa(D, y_n)\\
&(\because \sum_{S \in supp(\lambda_{n-1})\setminus \{X(a_0)\}} \lambda_{n-1}(S)=1-\lambda_{n-1}(X(a_0))=1-\lambda_{0}(X(a_0))=\lambda_{0}(\{\underline{x}\}))\\
&= \rhoa(D\cup a_0, y_n),
\end{align*}
where the last equality holds by the definition of $c_n$; the fourth to the last equality holds because $y_n$ is the maximum element among $D$ with respect to $\succ$ if and only if  $y_n$ is the maximum element among $D \cup S$ with respect to $\succ$ since $S\subseteq \{x_0(y_i)| i < n\}\cup\{\underline{x}\}$ and $x_0(y_i)$ is the immediate predecessor of $y_i$ for all $ \succ\in supp(\muX)$.

\subsection{Technical remarks}\label{sec:remarks}

\begin{remark}\label{rem:domain_assumption}
For Theorem \ref{thm:1}, in order to make sense of Limited Monotonicity, we need the following assumption of the domain of $\rho$: if $\rhoa $ is defined on $(D \cup \{a_0\},\cdot)$, then $\rhoa $ is defined  on $(D,\cdot)$ (i.e., if $D\cup\{a_0\} \in \D$, then $D \in \D$). In this case, Limited Monotonicity should be defined as: for all $D \subseteq \A_A$ such that $D \cup \{a_0\} \in \D$, $\rho$ satisfies $\rho(D, a) \ge \rho(D \cup \{a_0\} ,a)$ for all $a \in D$. 
\end{remark}

\begin{remark}\label{rem:generalru}
    In our notion of RU rationality, we assume that the composition of aggregates and the agent's preferences are independent. A natural extension is to allow the composition distribution to depend on the agent's preferences. That is, a composition distribution $\lambda_{\succ, A}$ is a distribution over ~$\prod_{a \in A}(2^{X(a)}\setminus \{\emptyset\})$ that may jointly depend on the available aggregates $A$ and the agent's preference $\succ 
    \in \L(\X)$. For example, we may allow the case where consumers with a preference for high quality goods tend to have a higher quality outside option.\footnote{We do not, however, allow the composition of the outside option to enter into the consumer's underlying choices. This type of dependence would lead to RUM violations in the underlying choice which  would lead to a model with no observable implications.}
    
    We say that $\rhoa$ is General RU rational if there exists $(\muX, \lambda_{\succ, A})$ such that $\rhoa(A,a)$ equals
    \begin{align}
\sum_{(S_b)_{b \in A}\in \underset{b\in A}{\prod} (2^{X(b)}\setminus \{\emptyset\})} \sum_{\succ \in \L(\X)} 1(\exists x \in S_a, \ \forall y \in \underset{b \in A \setminus \{a\}}{\cup} S_b, \ x \succ y)) \lambda_{\succ,A}((S_b)_{b \in A})\muX(\succ).
\end{align}
Since this model includes the independent case, RU rationality implies General RU rationality. Furthermore, it can be shown that General RU rationality implies both Limited Monotonicity and Partial RU rationality. Thus RU rationality and General RU rationality are equivalent.

\end{remark}

\begin{remark}\label{remark:threshold}

That the threshold equals $|\A_A|+1$ can be seen from how weak Limited Monotonicity is: it only implies $\rho(D,b)\ \ge\ \rho(D\cup\{a_0\},b)\quad\text{for all }b\in\A_A$, and therefore allows essentially arbitrary (and highly heterogeneous) substitution across atomic aggregates. For instance, it permits $\rho(D,a)=\rho(D\cup\{a_0\},a)$ for some $a$ (no substitution) while $\rho(D,b)\gg \rho(D\cup\{a_0\},b)$ for other $b$ (strong substitution). To rationalize such patterns with an underlying RUM, one needs, for each $b\in\A_A$, a distinct underlying alternative that can become available via $a_0$ and is ranked just above $b$, yielding $|\A_A|$ underlying alternatives. If $\rho(D,a)=\rho(D\cup\{a_0\},a)$ holds for some $a$, then $a_0$ must be worse than every atomic aggregate, which requires one additional underlying alternative ranked below all of $\A_A$. Hence the threshold is $|\A_A|+1$.
\end{remark}

\begin{remark}\label{rem:dom_nests}
    We assumed the full domain $\D  = 2^\A \setminus \{\emptyset\}$ for Theorem \ref{thm:nests}. However, the proof only requires that there exists a labeling $\A_A = \{y_1,...,y_{|\A|}\}$ such that for all $1 \le n \le |\A|$, the sets  $D_n = \{y_1,...,y_n\}$ and $D_n \cup \{a_0\}$, as well as any other $D \subset \A_A$ are in $\D$,.
\end{remark}

\begin{remark}\label{rem:simplexintuition}
In the remark, we provide intuition why Lemma \ref{lem:subpolys} holds. 
 Let $\A_A = \{x,y\}$ and $X(a_0) = \{z,w\}$. Let $x \ \succ_1 \  y \ \ \succ_1 \ a_0$ and $y \ \succ_2\  x  \ \succ_2 \ a_0$. Also, suppose that 
\begin{align*}
&\{x,a_0\} \notin \E_1,  \  \{y,a_0\} \in \E_1, \ \{x,y,a_0\} \in \E_1, \ \ \  
&\{x,a_0\} \notin \E_2, \ \{y,a_0\} \in \E_2,  \ \{x,y,a_0\} \notin \E_2.
\end{align*} For the choice on $D = \{x,y,a_0\}$ when we mix $\rho^{\succ_1}_{\E_1}$ and $\rho^{\succ_2}_{\E_1}$, $a_0$ is very good when the agent has preference $\succ_2$ but very bad when the agent has preference $\succ_1$. However, our model does not have the correlation between preferences and $\lambda$. Thus in order to rationalize this choice function (the mixture), in particular, the correlation between the preference relation and the value of the outside option, we need to construct preferences in a way that mimic the correlation pattern as follows. Each vertex gets a special underlying alternative $\overline{x}_i \in X(a_0)$ that plays the role of the outside option \textit{only for $\succ_i$}. That is, $\overline{x}_i$ is ranked at the top of $\succ_i$ when extended to $\X$. In particular, we can make $\overline{x}_i$ ranked very low in any ranking except $\succ_i$ (when extended to $\X$).   With this construction, $\lambda$ has a different effect on each $\succ_i$, which mimics the effect of correlation between $\muX$ and $\lambda$. The $-1$ comes from the fact that we need another underlying alternative in $X(a_0)$ at the bottom in that case where $a_0$ is never chosen.
\end{remark}

\begin{remark}\label{rem:dom_vertex}
Theorem \ref{thm:vertex} does not require the domain assumption that $a_0 \in D$ for all $D \in \D$. The theorem continues to hold with a slight modification of the definition of $c^{\succ}_{\E}$. The adjustment is needed because a set $D$ may fail to contain $a_0$ even when $D \notin \E$. The generalized definition is:
\[
c^{\succ}_{\E}(D)= \begin{cases}
\max_D \succ \ &\text{ if } a_0 \notin D \text{ or if } D \in \E, \\
a_0\ &\text { otherwise.}
\end{cases}
\]
This is the same as before, except that $c^{\succ}_{\E}$ is now defined for menus that do not contain $a_0$. With this modification, Theorem~\ref{thm:vertex} holds without change.
\end{remark}

\begin{remark}\label{rem:separable}
Some readers may think that assuming additive  separable  preferences over $\X$ may help to justify the use of ARUM. For example, \cite{ealesmeat1982} suggests that aggregates should be defined to satisfy separability. This is not the case, as separability does not address unknown composition in a discrete choice setting.

To see this, notice that the non-overlappingness of preferences is equivalent to the condition that any preferences can be represented by the following form of utility function. Assuming $\A=\{a_0,a_1\}$,  $\A_N=\{a_0\}$, $\X=\{x_1, x_2, y\}$ for simplicity, 
\[
w(x_1+x_2,y)
\]
where $X(a_0)=\{x_1,x_2\}$ and $X(a_1)=\{y\}$.  
To see why, refer to Theorem \ref{thm:non-overlapping}, which says that non-overlapping preferences are equivalent to preferences over aggregates. Since $x_1+x_2$ is exactly the quantity of the outside option consumed, $w$ is a utility function over aggregates. 

Thus, additive separability of preferences is not enough to justify ARUM as the general form,
\[
w(u(x_1,x_2), y) ,
\] 
where $w$ is linear, is much weaker than non-overlapping preferences. For example, let $w$ be the sum, and $u(x_1,x_2) = 2x_1+\frac{1}{2}x_2$. This induces the preference $x_1 \succ y \succ x_2$.\footnote{ We abuse notation and conflate the quantity with the good.} Suppose that alternative $x$ in a choice set $D_1$ means $x_1$; alternative $x$ in choice set $D_2$ means $x_2$. Then we observe a cycle that $x \succ y$, which is revealed from a choice from  $D_1$, and $y \succ x$, which is revealed from a choice from  $D_2$. Thus there is no linear order (or utility function) over  aggregates.
\end{remark}

\begin{remark}\label{rem:contrapositive}{Proposition \ref{thm:non-overlapping}} says that ARUM is equivalent to rationalization by non-overlapping preferences (even if the composition distribution is menu dependent) and Proposition \ref{thm:independent} says that ARUM is equivalent to menu independence (even if preferences are overlapping). These are two distinct, but not contradictory characterizations. Because because $(\mu_\X,\lambda)$ is not identified, it is possible for an ARUM is generated by overlapping preferences and/or menu-dependent composition distribution. The two propositions respectively say that there exists a non-overlapping rationalization and a menu-independent rationalization.
\end{remark}

\begin{remark}
To better understand the structure of the RU-rational set, in particular the non-convexity, note that  $\rho$ is bilinear but not linear in $(\muX,\lambda)$.

To see this, denote $\rho$ rationalized with $(\muX,\lambda)$ by $\rho_{(\muX,\lambda)}$. 
Fix a  composition correspondence $X$. With $\muX$ fixed, $\rho_{(\muX,\lambda)}$ is linear in $\lambda$. For any $\alpha \in [0,1]$, composition distributions $\lambda$, $\lambda'$ such that $\lambda'' = \alpha\lambda +(1-\alpha)\lambda'$, we have $\alpha\rho_{(\muX,\lambda)}+(1-\alpha)\rho_{(\muX,\lambda')} = \rho_{(\muX,\lambda'')}$. Similarly, with $\lambda$ fixed,  $\rho_{(\muX,\lambda)}$ is linear in $\muX$ in the sense that if $\alpha\muX+(1-\alpha) \muX' = \muX''$ then   $\alpha\rho_{(\muX,\lambda)}+(1-\alpha)\rho_{(\muX,\lambda)}= \rho_{(\muX'',\lambda)}$.

However, it is not linear because $\alpha\rho_{(\muX,\lambda)}+(1-\alpha)\rho_{(\muX',\lambda')}\neq \rho_{(\muX'',\lambda'')}$ in general. Note that this non-linearity of $\rho_{(\muX, \lambda)}$ does not directly implies  the non-convexity of the RU rational set, although if it were linear then the set would immediately be convex.

\end{remark}

\section{Supplementary material for simulation}\label{sec:additional_fig} 

\subsection{Additional Explanation for Section \ref{sec:simulation} }\label{sec:sim_intuition}

In this section, we provide further explanations of the figures  in Section \ref{sec:simulation}. 

\subsubsection{Figure \ref{fig:bias_pref_struc}}

The intuition is as follows. In this setup, $a_0$ is likely to be composed of $z$ in both $\{x,y,a_0\}$ and $\{x,a_0\}$, since $\lambda_{\{x,y,a_0\}}(\{z\})=\lambda_{\{x,a_0\}}(\{z\})$ is high, whereas $a_0$ is likely to be composed of $w$ in $\{y,a_0\}$, given that $\lambda_{\{y,a_0\}}(\{w\})$ is high. Consequently, when $u(w) > u(z)$, the relative attractiveness of $x$ in $\{x,a_0\}$ and $\{x,y,a_0\}$ is high, leading to the overestimation of $\hat{u}(x)$; and the relative attractiveness of $y$ in $\{y,a_0\}$ is low, leading to an underestimation of $\hat{u}(y)$ and thus a positive bias in the top-left corner. Conversely, when $u(z) > u(w)$, the relative attractiveness of $x$ in $\{x,y,a_0\}$ and $\{x,a_0\}$ is low, leading to an underestimation of $\hat{u}(x)$; and the relative attractiveness of $y$ in $\{y,a_0\}$ is high, leading to an overestimation of $\hat{u}(y)$. This results in a negative bias in the bottom-right corner.

\subsubsection{Figure \ref{fig:distance_lam}}

To understand the figure, recall that we consider three alternatives $x,y$, and $a_0$. Thus, it can be shown that  $\rho$ is ARUM if and only if it satisfies the standard monotonicity condition for all menus. Since $\rho$ is RU-rational, it satisfies Limited Monotonicity, that is monotonicity on $\{x,y\}$ with respect to adding $a_0$. Also, since $\lambda_{\{y,a_0\}} = \lambda_{\{x,y,a_0\}}$, monotonicity is satisfied on menu $\{y,a_0\}$ with respect to adding $x$. Thus the only possible ways to violate monotonicity are (i) $\rho(\{x,y,a_0\} ,a_0) > \rho(\{x,a_0\} ,a_0)$ and (ii) $\rho(\{x,y,a_0\} ,a_0) > \rho(\{y,a_0\} ,a_0)$. The first occurs when $a_0$ is less attractive on $\{x,a_0\}$ than on $\{x,y,a_0\}$  and the second occurs when $a_0$ is more attractive on $\{x,a_0\}$ than on $\{x,y,a_0\}$. 
Since $u(z)$ is high and $u(w)$ is low, the first case occurs when $a_0$ is more likely to be $\{w\}$ (e.g., $\lambda_{\{x,a_0\}}(\{w\})=1$); the second case occurs when $a_0$ is more likely to be $\{z\}$ or $\{z,w\}$ (e.g., $\lambda_{\{x,a_0\}}(\{w\})=0$). These are two cases where the distance is large in Figure
\ref{fig:distance_lam}.

\subsubsection{Figure \ref{fig:distance_pref_struc}}

The intuition behind the figure is as follows. As in the previous example, ARUM is characterized by monotonicity, so we can interpret the figure by identifying the parameters that generate the largest violations of monotonicity. Because $\lambda_{\{x,y,a_0\}} = \lambda_{\{x,a_0\}}$, monotonicity violations only occur when comparing $\{x,y,a_0\}$ and $\{y,a_0\}$.

When $u(z)$ is high and $u(w)$ is low, $a_0$ is relatively more attractive in $\{x,y,a_0\}$ than in $\{y,a_0\}$. This is because $a_0$ is more often composed of $z$ in $\{x,y,a_0\}$, whereas it is more often composed of $w$ in $\{y,a_0\}$. This implies $\rho(\{x,y,a_0\},a_0) > \rho(\{y,a_0\},a_0)$, producing a large violation of monotonicity in the bottom-right corner. By contrast, when $u(w)$ is high and $u(z)$ is low, $a_0$ becomes relatively less attractive in $\{x,y,a_0\}$ than in $\{y,a_0\}$. As a result, $y$ is more attractive in $\{y,a_0\}$, so that $\rho(\{x,y,a_0\},y) > \rho(\{y,a_0\},y)$—again a violation of monotonicity. However, because $y$ is much less attractive than $x$ (i.e., $u(x) > u(y)$), the size of this violation is considerably smaller than in the previous case. This explains why the distances in the top-left corner are smaller than those in the bottom-right.

In the following, we provide the results of simulations using different parameter values.

\subsection{Estimation bias}

\subsubsection{Effect of composition distributions $\lambda$}

 In this section, we replicate the simulation in Subsection~\ref{subsec:effect_lambda_bias} using different parameter values. Overall, the results are qualitatively similar: biases are smaller in the menu-independent cases, and as we move away from independence, the biases tend to increase.

There are, however, some exceptions. Consider the two heatmaps of biases across $\lambda_{\{x,y,a_0\}}$. In these heatmaps, the pattern differs slightly. The blue cell (the independent case) does not necessarily correspond to the smallest bias, and cells farther from independence do not always correspond to larger biases. This finding can be explained by the fact that these two heatmaps are generated by varying the values of $\lambda_{\{x,y,a_0\}}$, which in turn changes the meaning of $a_0$ in the choice set $\{x,y,a_0\}$. In this choice set, both $x$ and $y$ are present, so changes in the interpretation of $a_0$ affect the relative desirability of $x$ and $y$ symmetrically. As a result, such changes do not substantially affect the bias, which is defined as the difference in utilities between $x$ and $y$.

By contrast, when we vary the values of $\lambda_{\{x,a_0\}}$ or $\lambda_{\{y,a_0\}}$, the effect is asymmetric: changes in $\lambda_{\{x,a_0\}}$ affect only $x$, while changes in $\lambda_{\{y,a_0\}}$ affect only $y$. These changes directly influence the relative evaluation of $x$ and $y$, thereby producing larger estimation biases, as expected.

\begin{figure}[H]
    \centering

    \begin{subfigure}[t]{0.45\textwidth}
        \centering
        \includegraphics[width=\linewidth]{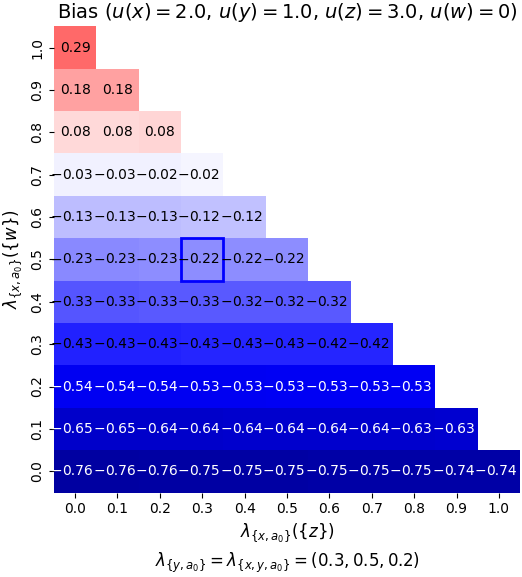}
        \caption{Heatmap of Biases across $\lambda_{\{x,a_0\}}$.}
    \end{subfigure}
\end{figure}

\begin{figure}[H]
    \centering
    \begin{subfigure}[t]{0.45\textwidth}
        \centering
        \includegraphics[width=\linewidth]{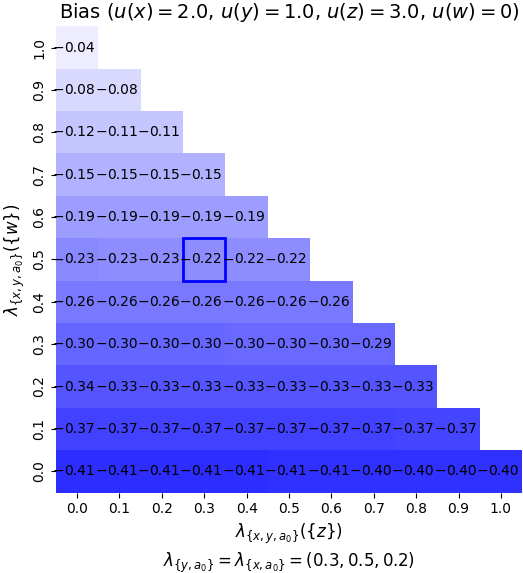}
        \caption*{Heatmap of Biases across $\lambda_{\{x, y, a_0\}}$.}
    \end{subfigure}
    \hfill
    \begin{subfigure}[t]{0.445\textwidth}
        \centering
        \includegraphics[width=\linewidth]{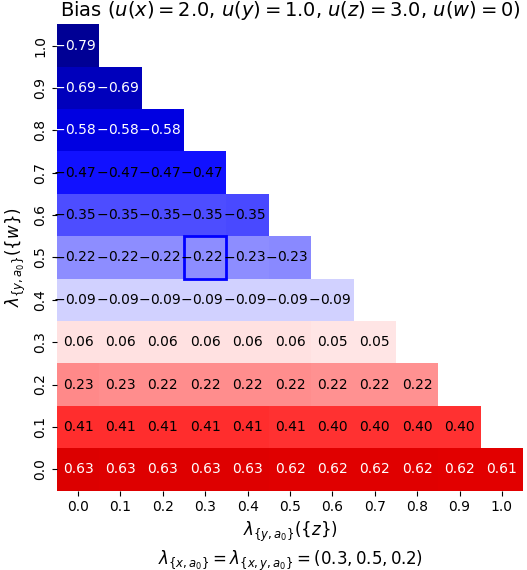}
        \caption*{Heatmap of Biases across $\lambda_{\{y,a_0\}}$.}
    \end{subfigure}

\end{figure}
\begin{figure}[H]
    \centering

    \begin{subfigure}[t]{0.45\textwidth}
        \centering
        \includegraphics[width=\linewidth]{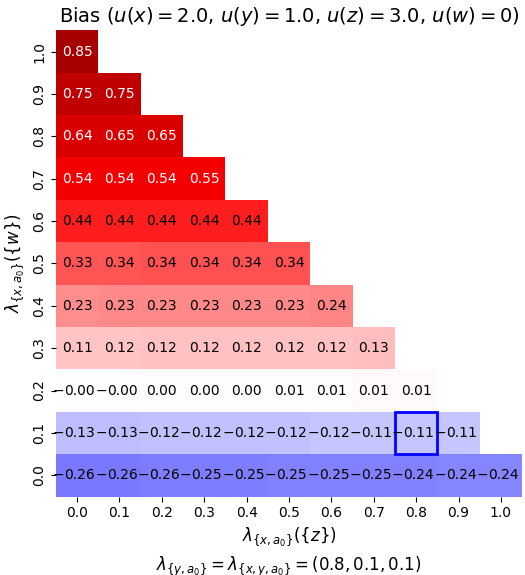}
        \caption{Heatmap of Biases across $\lambda_{\{x,a_0\}}$.}
    \end{subfigure}

\end{figure}

\begin{figure}[H]
    \centering
    
    \begin{subfigure}[t]{0.45\textwidth}
        \centering
        \includegraphics[width=\linewidth]{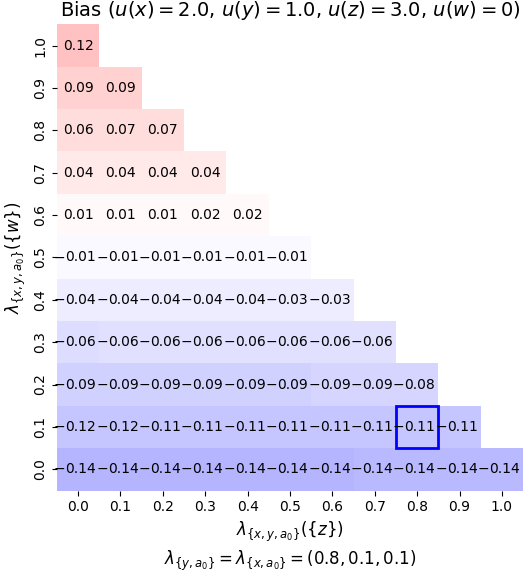}
        \caption*{Heatmap of Biases across $\lambda_{\{x, y, a_0\}}$.}
    \end{subfigure}
    \hfill
    \begin{subfigure}[t]{0.5\textwidth}
        \centering
        \includegraphics[width=\linewidth]{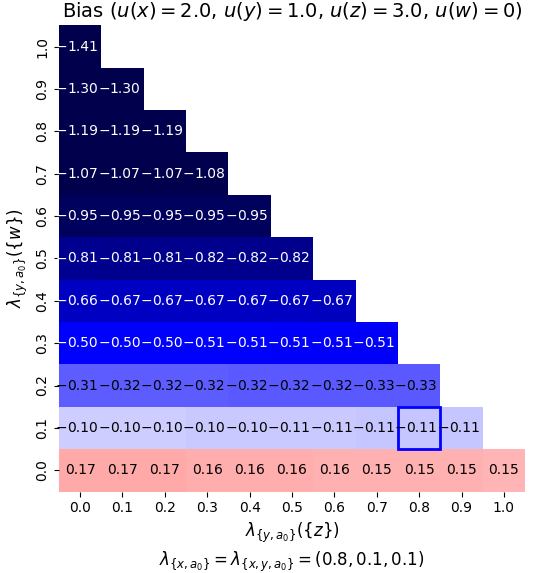}
        \caption*{Heatmap of Biases across $\lambda_{\{y,a_0\}}$.}
    \end{subfigure}

\end{figure}

\subsubsection{Maximum/Minimum, and independent case}\label{subsection:additonal_fig_max_min}

In this section, we replicate the simulation in Subsection \ref{subsection:max_min_biases} under different parameter values. The setup is the same as in Figure \ref{fig:max_min_ind_bias}, except that we use alternative utility levels and different axes. The overall patterns are consistent with those reported in the main body of the paper.

\begin{figure}[H]
    \centering
    \begin{subfigure}[t]{0.48\linewidth}
        \centering
        \includegraphics[width=\linewidth]{figures/flippedminmax1.png}
        
    \end{subfigure}%
    \hfill
    \begin{subfigure}[t]{0.48\linewidth}
        \centering
        \includegraphics[width=\linewidth]{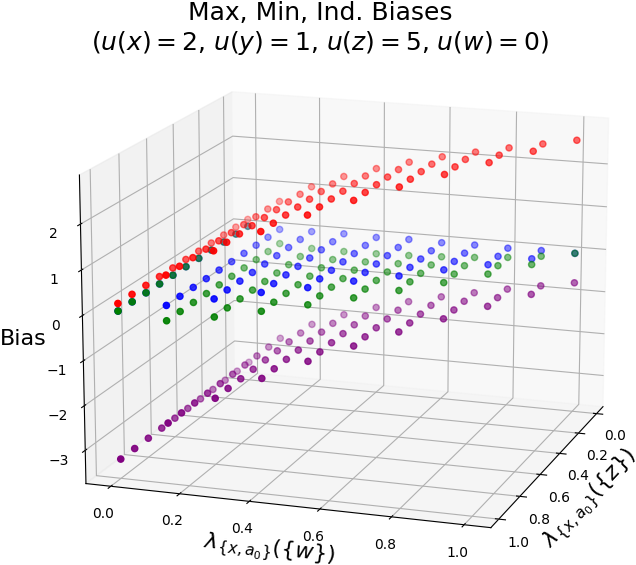}
        
    \end{subfigure}
    \caption*{}

\end{figure}
\begin{figure}[H]
    \centering
    \begin{subfigure}[t]{0.50\linewidth}
        \centering
        \includegraphics[width=\linewidth]{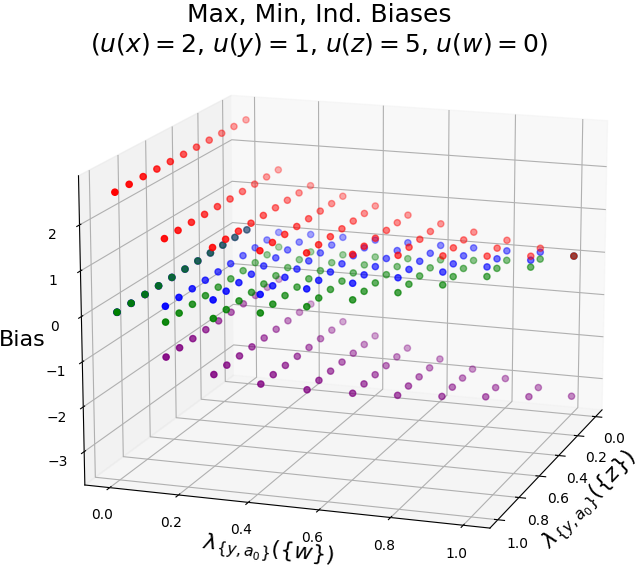}
        
    \end{subfigure}%
    \hfill
    \begin{subfigure}[t]{0.50\linewidth}
        \centering
        \includegraphics[width=\linewidth]{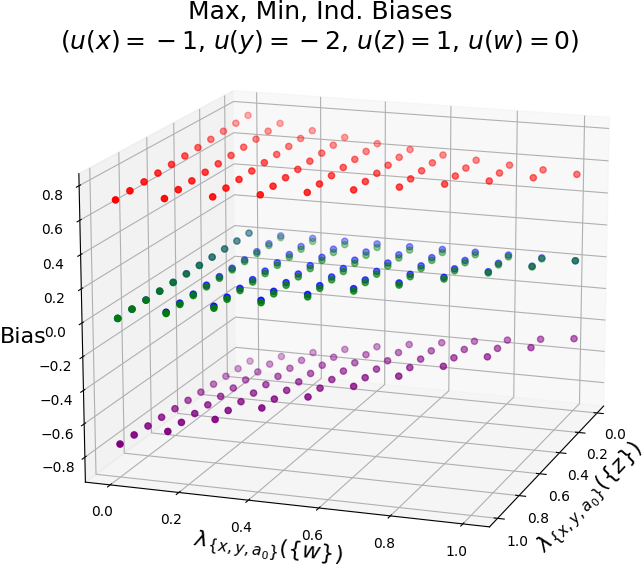}
        
    \end{subfigure}

\end{figure}
\begin{figure}[H]
    \centering
    \begin{subfigure}[t]{0.48\linewidth}
        \centering
        \includegraphics[width=\linewidth]{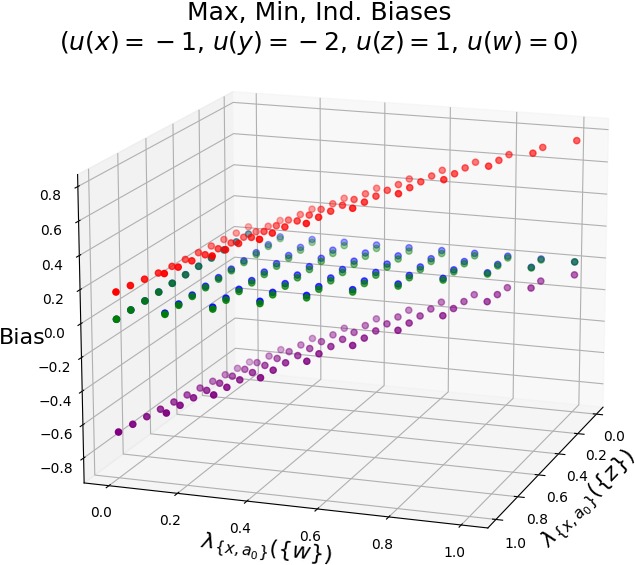}
        
    \end{subfigure}%
    \hfill
    \begin{subfigure}[t]{0.48\linewidth}
        \centering
        \includegraphics[width=\linewidth]{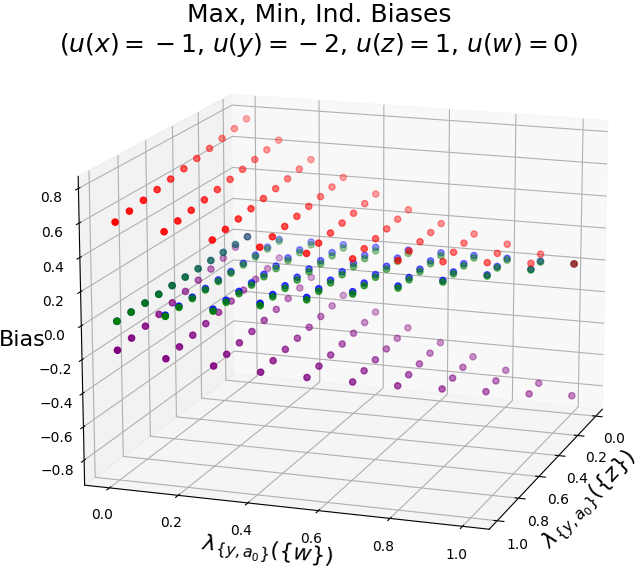}
        
    \end{subfigure}
    \caption*{}

\end{figure}

\clearpage
\subsection{Effect of preference structure on bias}

In this section, we repeat the simulation in Subsection \ref{subsec:effect_pref_bias} with different parameter values. The setup is the same as in Figure \ref{fig:bias_pref_struc} but with different values of $\lambda$.

The figures show the biases across utility values of $z$ and $w$. In general, the bias is largest when overlappingness is most likely (the top left and bottom right corners of each square) and smallest when overlappingness is less likely (the top right and bottom left corners of each square), consistent with Conjecture  (B) and Figure \ref{fig:bias_pref_struc}.

\begin{figure}[H]
    \centering
    \begin{subfigure}[t]{0.48\linewidth}
        \centering
        \includegraphics[width=\linewidth]{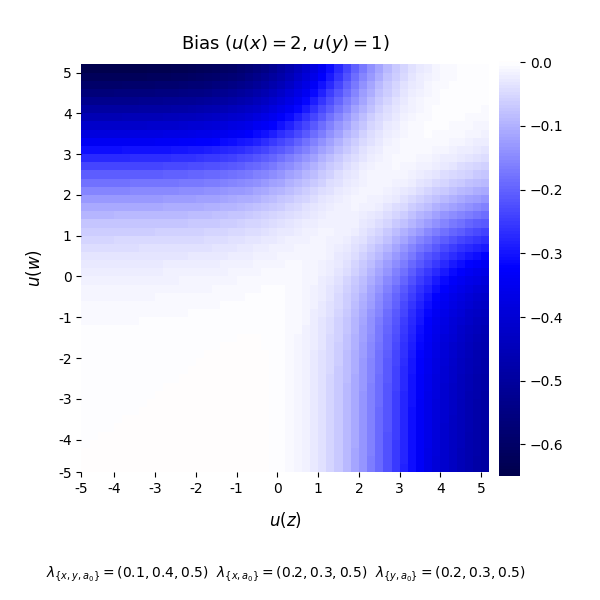}
        
    \end{subfigure}%
    \hfill
    \begin{subfigure}[t]{0.48\linewidth}
        \centering
        \includegraphics[width=\linewidth]{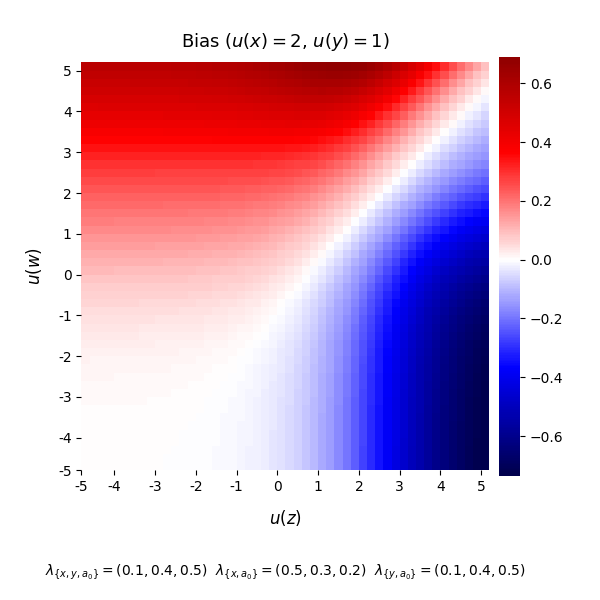}
        
    \end{subfigure}
    \caption*{}
    
\end{figure}
\begin{figure}[H]
    \centering
    \begin{subfigure}[t]{0.475\linewidth}
        \centering
        \includegraphics[width=\linewidth]{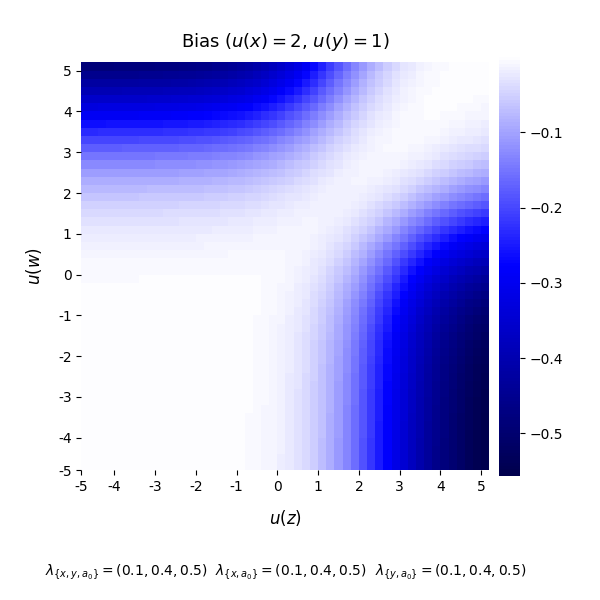}
        
    \end{subfigure}%
    \hfill
    \begin{subfigure}[t]{0.47\linewidth}
        \centering
        \includegraphics[width=\linewidth]{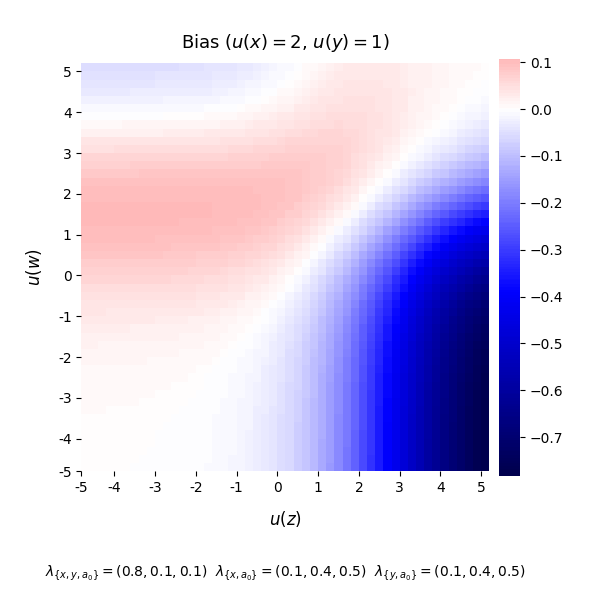}
        
    \end{subfigure}
    
\end{figure}

\newpage

\subsection{Distance to the ARU polytope}

\subsubsection{Effect of composition distribution $\lambda$ on distance}

In this section, we repeat the simulation in Subsection \ref{sec:ditance_lambda} with different parameter values. The setup is the same as in Figure \ref{fig:distance_lam} but with different values of $\lambda$. 

The figures show the distance across values of $\lambda$. The resutls are the same as in  Subsection \ref{sec:ditance_lambda}:  consistent with Proposition \ref{thm:independent}, the distance is zero in the independent case (blue-outlined cell). As $\lambda$ moves farther from the menu independence, the distance increases, confirming Conjecture  (B).
\begin{figure}[H]
    \centering

    \begin{subfigure}[t]{0.45\textwidth}
        \centering
        \includegraphics[width=\linewidth]{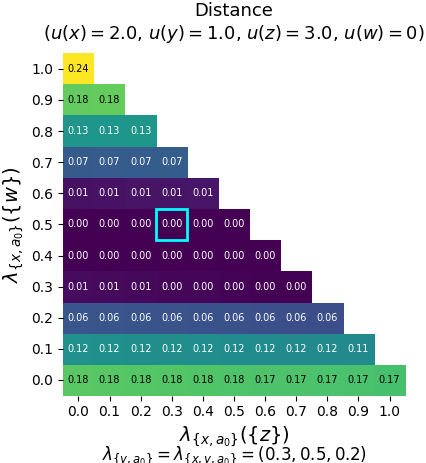}
        \caption{Heatmap of Distance across $\lambda_{\{x,a_0\}}$.}
    \end{subfigure}

\end{figure}    

\begin{figure}[H]
    \centering
    
    \begin{subfigure}[t]{0.45\textwidth}
        \centering
        \includegraphics[width=\linewidth]{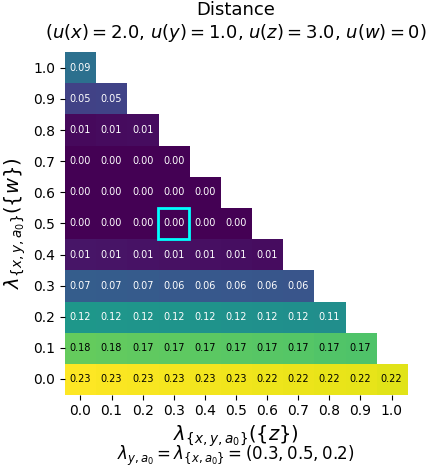}
        \caption*{Heatmap of Distance across $\lambda_{\{x, y, a_0\}}$.}
    \end{subfigure}
    \hfill
    \begin{subfigure}[t]{0.45\textwidth}
        \centering
        \includegraphics[width=\linewidth]{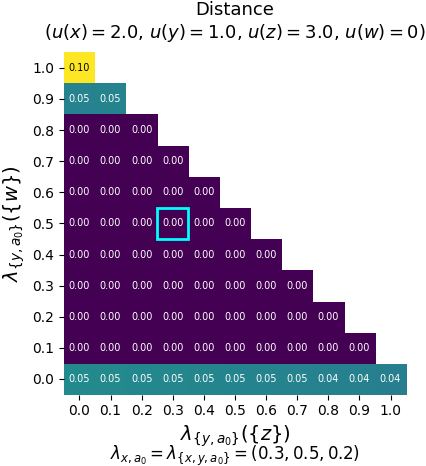}
        \caption*{Heatmap of Distance across $\lambda_{\{y,a_0\}}$.}
    \end{subfigure}
\end{figure}

\begin{figure}[H]
    \centering
    \begin{subfigure}[t]{0.45\textwidth}
        \centering
        \includegraphics[width=\linewidth]{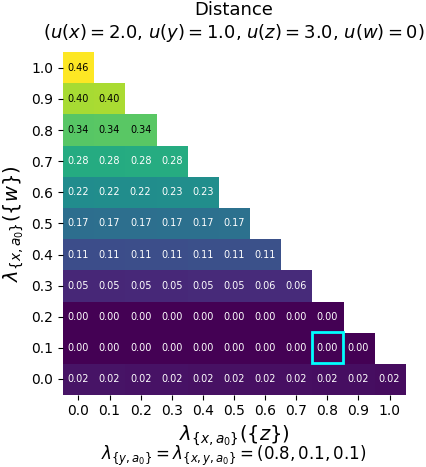}
        \caption{Heatmap of Distance across $\lambda_{\{x,a_0\}}$.}
    \end{subfigure}
\end{figure}

\begin{figure}[H]
    \centering

    \begin{subfigure}[t]{0.45\textwidth}
        \centering
        \includegraphics[width=\linewidth]{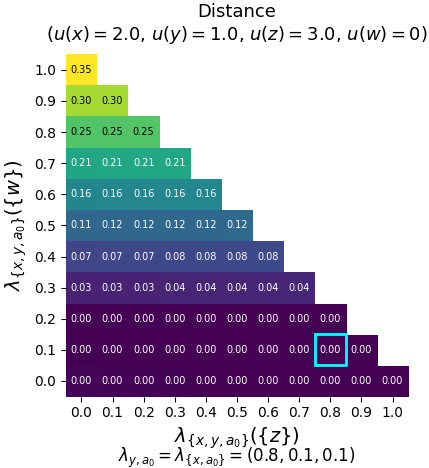}
        \caption*{Heatmap of Distance across $\lambda_{\{x, y, a_0\}}$.}
    \end{subfigure}
    \hfill
    \begin{subfigure}[t]{0.45\textwidth}
        \centering
        \includegraphics[width=\linewidth]{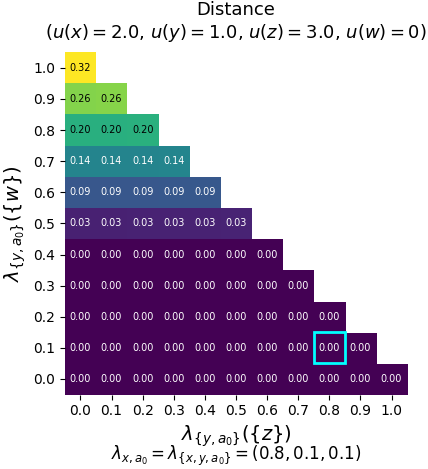}
        \caption*{Heatmap of Distance across $\lambda_{\{y,a_0\}}$.}
    \end{subfigure}
\end{figure}

\subsubsection{Effect of preference structure on distance}
    Finally we repeat the simulation in Section \ref{subsec:effect_pref_distance}. These four figures are the same as Figure \ref{fig:distance_pref_struc} with different values of $\lambda$ --- they show the distance from the ARUM polytope across across utility values of $z$ and $w$. 
    
    In general, the distance is largest when overlappingness is most likely (the top left and bottom right corners of each square) and smallest when overlappingness is less likely (the top right and bottom left corners of each square), consistent with Conjecture  (B) and the results obtained in the body of the paper.
    
    In the first figure ($\lambda_{\{x,y,a_0\}}=\lambda_{\{y,a_0\}}= (0.1,0.4,0.5) $ and $\lambda_{\{x,a_0\}} = (0.2,0.3,0.5)$ and the third figure ($\lambda_{\{x,y,a_0\}}= \lambda_{\{x,a_0\}}= \lambda_{\{y,a_0\}} = (0.1,0.4,0.5)$), the distances are close to zero since $\lambda$ is either menu-independent as in the third figure or close to menu-independent as in the first figure. In particular, when $\lambda$ is menu-independent the distance is always zero which is consistent with Proposition \ref{thm:independent}.

\begin{figure}[H]
    \centering
    \begin{subfigure}[t]{0.48\linewidth}
        \centering
        \includegraphics[width=\linewidth]{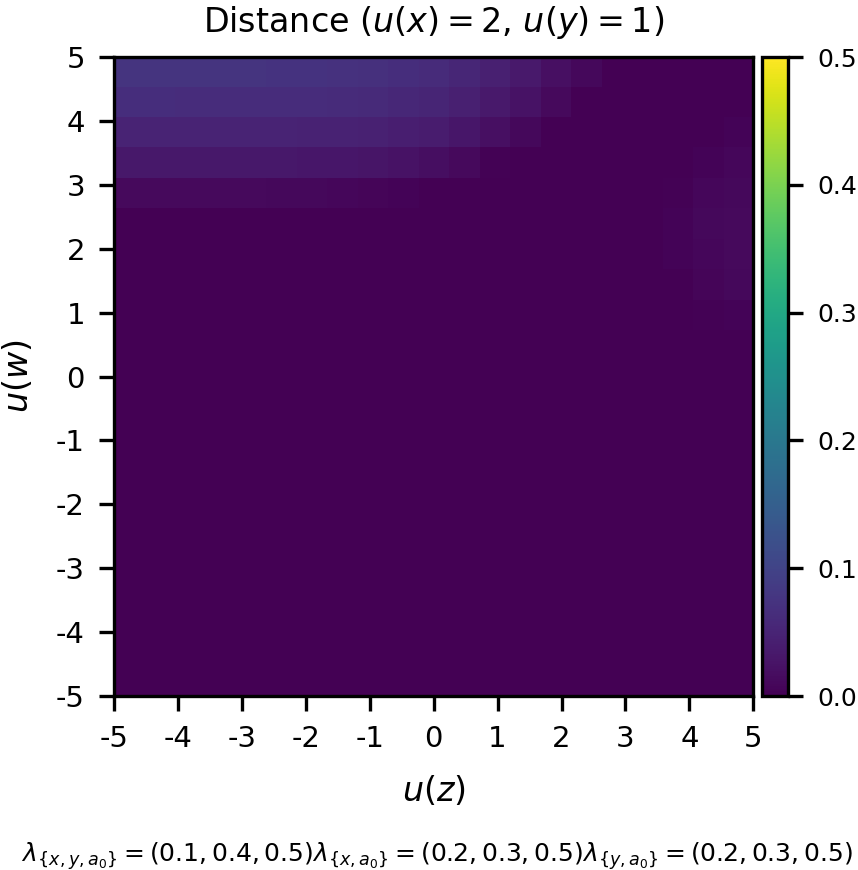}
        
    \end{subfigure}%
    \hfill
    \begin{subfigure}[t]{0.48\linewidth}
        \centering
        \includegraphics[width=\linewidth]{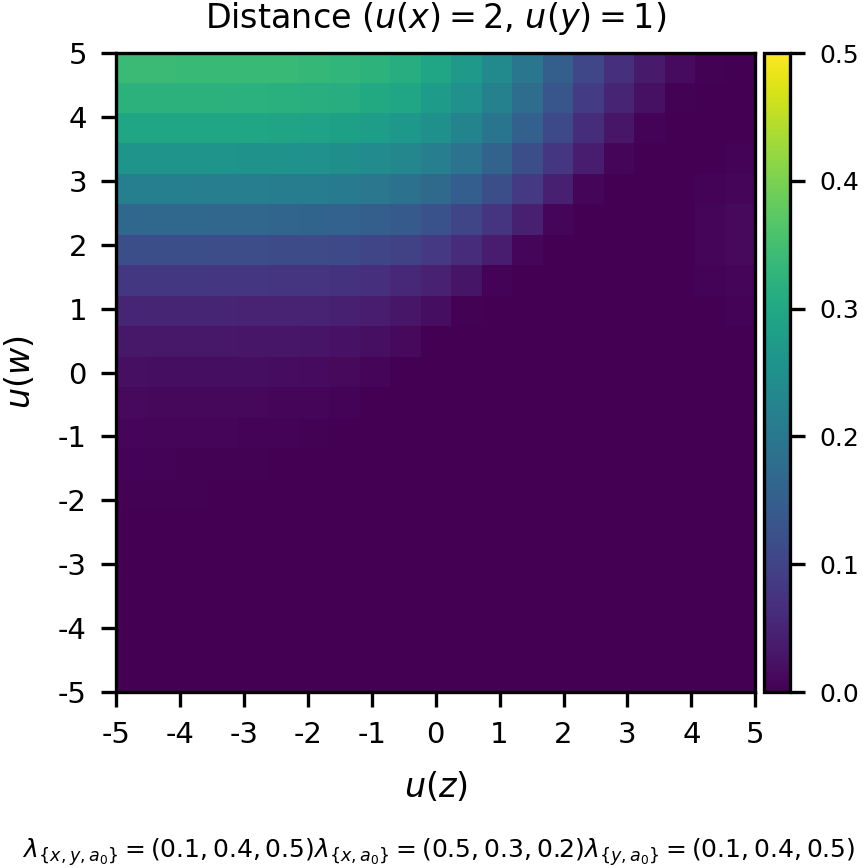}
        
    \end{subfigure}
    \caption*{}

\end{figure}
\begin{figure}[H]
    \centering
    \begin{subfigure}[t]{0.48\linewidth}
        \centering
        \includegraphics[width=\linewidth]{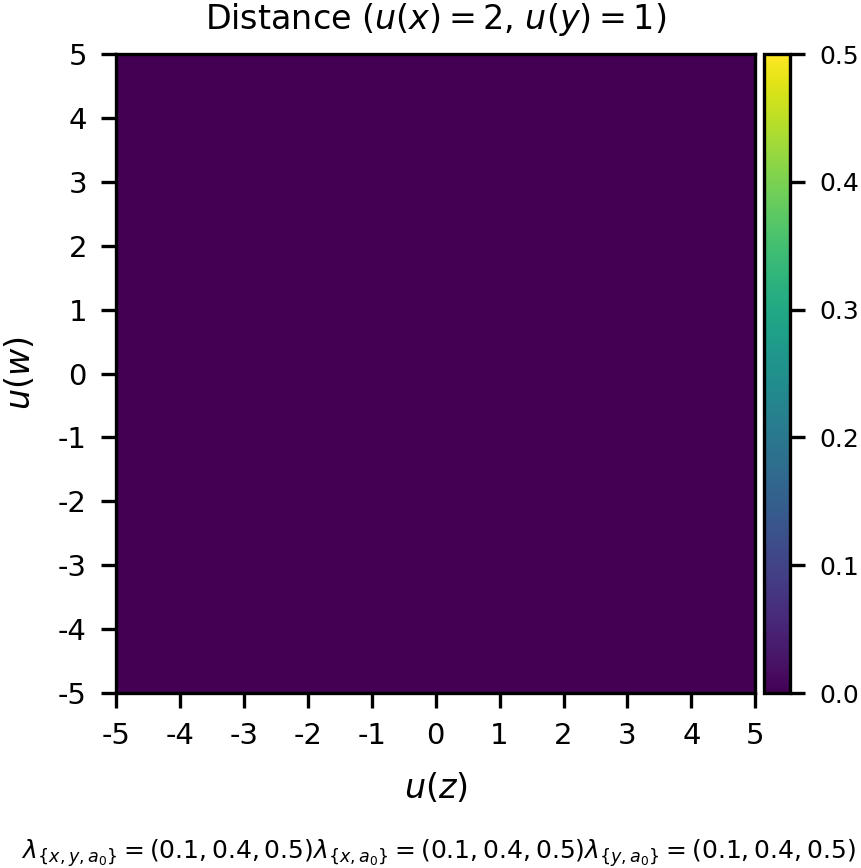}
        
    \end{subfigure}%
    \hfill
    \begin{subfigure}[t]{0.48\linewidth}
        \centering
        \includegraphics[width=\linewidth]{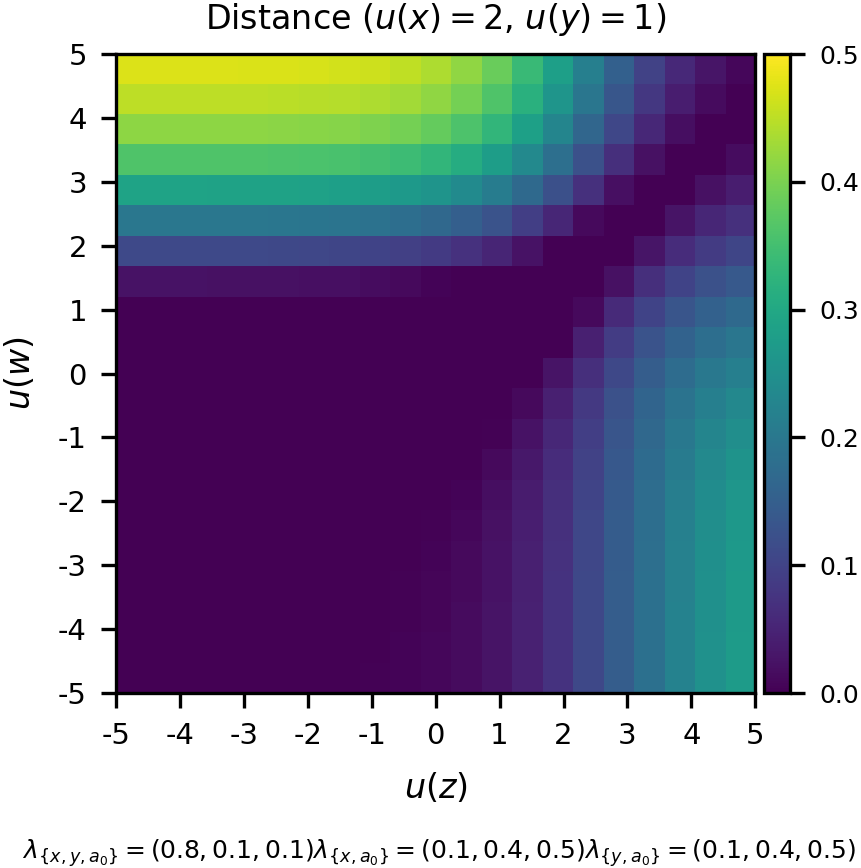}
        
    \end{subfigure}

\end{figure}

\end{document}